%% file: kernelcoco_4.tex
\documentclass[11pt]{article}
\newif\ifanon
\anonfalse
\usepackage[top=1in, bottom=1in, left=1in, right=1in]{geometry} % Management Science: 11pt, 1-inch margins, at least 1.5 spacing (https://pubsonline.informs.org/page/mnsc/submission-guidelines)
\usepackage{setspace}
\usepackage{graphicx}
\graphicspath{ {./pix/} }
\usepackage[page]{appendix}
\usepackage{amsfonts}
\usepackage{amsmath}
\usepackage{amsthm}
\usepackage{amssymb}
\usepackage{url}
\usepackage[colorlinks=true, allcolors=blue]{hyperref}
\usepackage{bm}
\usepackage{xcolor,algorithmicx,algorithm}
\usepackage{booktabs}
\usepackage{adjustbox}
\usepackage{mathtools} % for \coloneqq
\usepackage{siunitx}
\usepackage[font=small]{subcaption}
\usepackage{rotating}
\usepackage{dcolumn}
\usepackage{float}
\usepackage[round,authoryear]{natbib}
\usepackage[colorinlistoftodos]{todonotes}

\DeclareMathOperator{\spn}{span}
\DeclareMathOperator{\Ima}{Im}
\DeclareMathOperator{\diag}{diag}
\DeclareMathOperator{\Diag}{Diag}
\DeclareMathOperator{\vect}{vec}

\DeclareMathOperator{\tr}{tr}
\DeclareMathOperator{\rank}{rank}
\DeclareMathOperator{\var}{Var}
\DeclareMathOperator{\cov}{Cov}

\DeclareMathOperator*\argmin{\arg\!\min}

\newcommand{\muf}{{\bm b}}

\newcommand{\sy}{{\rm sy}}
\newcommand{\id}{{\rm id}}
\newcommand{\hsy}{h^\sy}
\newcommand{\hid}{h^\id}
\newcommand{\Hsy}{\Hcal^\sy}
\newcommand{\Hid}{\Hcal^\id}
\newcommand{\Gsy}{\Gcal^\sy}
\newcommand{\Gid}{\Gcal^\id}

\newcommand{\ksy}{k^\sy}
\newcommand{\kid}{k^\id}
\newcommand{\lsy}{\lambda^\sy}
\newcommand{\lid}{\lambda^\id}
\newcommand{\E}{{\mathbb E}}
\newcommand{\Pa}{{\mathbb P}}
\newcommand{\R}{{\mathbb R}}
\newcommand{\Sy}{{\mathbb S}}
\newcommand{\Sm}{{\mathbb U}}   % parameter space S^{m^sy+1} x S^{m^id} (revision 10: was S^{bm m})

\newcommand{\Acal}{{\mathcal A}}

\newcommand{\Ccal}{{\mathcal C}}
\newcommand{\Dcal}{{\mathcal D}}
\newcommand{\Ecal}{{\mathcal E}}

\newcommand{\Gcal}{{\mathcal G}}
\newcommand{\Hcal}{{\mathcal H}}
\newcommand{\Kcal}{{\mathcal K}}
\newcommand{\Lcal}{{\mathcal L}}
\newcommand{\Mcal}{{\mathcal M}}
\newcommand{\Pcal}{{\mathcal P}}
\newcommand{\Tcal}{{\mathcal T}}
\newcommand{\Rcal}{{\mathcal R}}
\newcommand{\Scal}{{\mathcal S}}

\newcommand{\Xcal}{{\mathcal X}}
\newcommand{\Zcal}{{\mathcal Z}}
\newcommand{\X}{{\bm X}}            % augmented return vector X_{t+1} = [1; x_{t+1}]
\newcommand{\Zt}{{\bm Z}}           % augmented covariate array Z_t = [Delta; z_t]
\newcommand{\Om}{{\bm\Omega}}         % weight matrix of the loss
\newcommand{\aphi}{\bm\Phi}   % augmented feature map (session 3: was \bar{\bm\phi})

\newtheorem{proposition}{Proposition}[section]
\newtheorem{lemma}[proposition]{Lemma}
\newtheorem{theorem}[proposition]{Theorem}
\newtheorem{definition}[proposition]{Definition}
\newtheorem{corollary}[proposition]{Corollary}
\newtheorem{remark}[proposition]{Remark}
\newtheorem{assumption}[proposition]{Assumption}
\newtheorem{exampleemph}[proposition]{Example}
\newenvironment{example}[1][]{\begin{exampleemph}[#1]\upshape}{\end{exampleemph}}

\ifanon
\author{}
\date{}
\else
\author{Damir\ Filipovic\footnote{
 \'Ecole Polytechnique F\'ed\'erale de Lausanne and Swiss Finance Institute, Email: \texttt{damir.filipovic@epfl.ch}}
 \and
 Paul\ Schneider\footnote{
 Universit\`a della Svizzera italiana and Swiss Finance Institute, Email: \texttt{paul.schneider@usi.ch}
 }}

\date{October 8, 2026}
\fi

\begin{document}
%%%%%%%%%%%%%%%%%%%%%%%%%%%%%%%%%%%%%%%%%%%%%%%%%%%%%%%%%%%%%%%%%%%%%%%%

\ifanon
\title{Conditional Moment Kernels: Joint Modeling of Expected Returns and Covariances in Unbalanced Panels}
\else
\title{Conditional Moment Kernels: Joint Modeling of Expected Returns and Covariances in Unbalanced Panels\footnote{A former version of this paper circulated under the title ``Joint Estimation of Conditional Mean and Covariance for Unbalanced Panels''. We thank Sam Cohen, Daniel Kuhn, Andreas Neuhierl, Olivier Scaillet, Raman Uppal and Michael Wolf for helpful comments. We also thank the participants of the Oxford Man Institute Machine Learning in Quantitative Finance Conference (2024), the 17th International Conference on Computational and Financial Econometrics (CFE 2023), the York Asset Pricing Workshop (2024), in particular our discussant Jia Chen, Numerical Methods for Finance and Insurance (2025), the 4th Hong Kong Conference for Fintech, AI, and Big Data in Business, in particular our discussant Dake Zhang, the Futures of Quantitative Finance Seminar at Universit\'e Paris Cit\'e (online, 2026), the Quantitative Finance Conference at the National University of Singapore (2026), and of seminars at the Universities of St.~Gallen, Geneva, Lugano, Cambridge, Rotterdam and Z\"urich and at WU Vienna. We used AI assistants for drafting and editing the text, checking derivations and writing code. Paul Schneider gratefully acknowledges the Swiss National Science Foundation grant 105218\_215528 ``Large-scale kernel methods in financial economics''.}}
\fi

\maketitle

%%%%%%%%%%%%%%%%%%%%%%%%%%%%%%%%%%%%%%%%%%%%%%%%%%%%%%%%%%%%%%%%%%%%%%%%
\begin{abstract}
We propose the conditional moment kernel (CMK), a joint model of the conditional mean and covariance matrix of asset returns given their characteristics. A single kernel function of the characteristics gives the means and second moments of all returns. The implied covariance matrices are positive semidefinite for any number and composition of assets, so the model applies directly to unbalanced panels. Its estimation is a convex program whose size does not depend on the number of assets. A floor at the scale of the idiosyncratic noise keeps the estimated systematic covariance away from zero where the data cannot resolve it, and a trace penalty selects the number of factors. A closed form, the CMK estimator, approximates the solution. For fixed features and hyper-parameters, we prove identification, consistency and finite-sample bounds. In out-of-sample tests on US stocks from 1972 to 2021, the CMK estimator fits the conditional moments significantly better than instrumented principal component analysis (IPCA) with the same characteristics. Almost all of the gain is in the covariance matrix, and nonlinear functions of the characteristics add to it. The implied unconditionally mean--variance efficient portfolios earn annualized Sharpe ratios of 4.0 to 5.1 before costs, close to those of IPCA and declining over time.

\vspace{2ex}


\vspace{1ex}

\end{abstract}

\clearpage

\section{Introduction}\label{sec_introduction}

Conditional asset pricing and portfolio choice rest on the conditional mean and covariance matrix of asset returns \citep{markowitz52,hansenrichard87}. Errors in either moment carry over to the weights of mean--variance efficient portfolios, and the covariance matrix enters them through its inverse, which magnifies the errors in its smallest eigenvalues. The characteristics that predict returns \citep{cochrane11} also drive risk \citep{kozaknagel24}, so the two moments should be modeled jointly as functions of the characteristics. Firms enter and leave, so the panel of stocks is unbalanced, and a month has thousands of stocks. A model of the conditional covariance must therefore return a positive semidefinite (PSD) matrix for every cross section, whatever its size and composition.

Most of the literature models one moment at a time. Conditional means are estimated by return prediction with characteristics \citep{freybergerneuhierlweber20,gukellyxiu20b}, and conditional covariances by large-dimensional estimators for a fixed set of assets \citep{ledoitwolf04,fanliaomincheva13,engleledoitwolf19}. The instrumented principal component analysis (IPCA) of \citet{kellypruittsu19} delivers both moments for changing cross sections, but only as a by-product of a factor model with a fixed number of factors, fitted by a nonconvex least squares problem. What is missing is a model that specifies the two moments as one object, returns PSD covariance matrices for every cross section, chooses the number of factors from the data and is estimated by a convex program with statistical guarantees. This paper provides such a model and shows that it forecasts the covariance matrix of US stocks better than IPCA.

The model, the conditional moment kernel (CMK), works with second moments rather than covariances. Augmenting the return vector by the constant one puts the means into the second-moment matrix of the augmented vector. We model this matrix by a kernel function of the characteristics. Evaluated at the characteristics of any set of stocks, the kernel function gives a PSD matrix. The covariance matrix is a Schur complement of this matrix and is therefore PSD as well. The kernel has a systematic part, which carries the means and the factor structure, and an idiosyncratic part, which adds a variance for each stock.

We represent the kernel by a finite set of functions of the characteristics, the features, for instance features learned by the Nystr\"om method \citep{drineasmahoney05}. The kernel is then determined by a PSD matrix whose size is set by the number of features. We estimate this matrix by weighted least squares on the realized second moments. This is a convex semidefinite program, the CMK program, whose size does not depend on the number of stocks. We specify the program by four choices. First, the weight of the loss whitens the features of each month. The returns of a month then enter the loss only through the coefficients and the residual variance of a cross-sectional regression of the returns on the features. The conditional covariance of these coefficients is the systematic covariance plus a noise term that shrinks with the number of stocks, so every month is about equally informative and gets the same weight. Second, a Tikhonov regularization stabilizes the monthly regressions when the features are nearly collinear. Third, a floor bounds the systematic covariance from below by a multiple of the noise covariance of the monthly coefficients. Without the floor, the estimate sets the systematic variance to zero in directions where it is below the noise level, and portfolio weights, which use the inverse covariance matrix, take this zero at face value. Fourth, a trace penalty, a convex substitute for the rank \citep{rechtfazelparrilo10}, selects the number of factors, which we count as the directions in which the systematic covariance exceeds the floor.

With these choices the solution has an approximate closed form, the CMK estimator. The expected return of a stock is its features times the average of the monthly regression coefficients over the estimation window, as in \citet{famamacbeth73}. The systematic covariance is the sample covariance of these coefficients net of their noise, raised to the floor and shrunk above it by the penalty. The simplest kernel is the linear kernel, whose features are a constant and the characteristics, the instruments of IPCA. With the linear kernel the monthly coefficients are returns of characteristic portfolios, and the CMK estimator is a characteristic-based factor risk model in the tradition of \citet{rosenberg74}, extended to the mean. IPCA with the features as instruments and as many factors as features implies the average and the sample covariance of the same coefficients, without the noise correction, the floor and the penalty, and with a weaker regularization of the monthly regressions. With fewer factors, IPCA replaces the penalty by a rank restriction.

For fixed features and hyper-parameters, and without the floor and the penalty, the minimizer of the population loss reproduces the true conditional moments if the model is well specified and is their best approximation within the model otherwise. In the well-specified case we bound the biases of the floor and the penalty. When the population loss is strongly convex, the solution of the CMK program converges to its minimizer for stationary and ergodic data. For independent data it satisfies a mean squared error bound that decreases at the parametric rate in the number of months, and an exponential finite-sample bound. We also derive the limit of the CMK estimator. The estimated kernel implies traded factor portfolios and an unconditionally mean--variance efficient (UMVE) portfolio with closed-form weights.

We apply the CMK estimator to the monthly returns and 94 characteristics of US stocks of \citet{gukellyxiu20b} from 1962 to 2021, with about 4,300 stocks in the average test month. The number of features, the floor and the penalty are selected in rolling windows, and all measures are out of sample. The main criterion is the Dawid--Sebastiani score, a proper scoring rule for the mean and the covariance matrix, which we report as a differential over a purely idiosyncratic benchmark. With the linear kernel the CMK estimator scores 817, which raises the Gaussian log-likelihood of a test month by about 0.1 per stock over the benchmark. IPCA with the same characteristics scores 774 to 793, depending on the number of factors, and the linear kernel beats it for every number of factors, with $t$-statistics between $3.2$ and $4.8$. Almost all of the gain is in the covariance part of the score, and the mean fits about as well as with IPCA. The global minimum-variance (GMV) portfolio, whose weights do not use the mean, has an out-of-sample volatility of 8.7\% a year with the linear kernel, against 9.0\% to 10.8\% with IPCA. The difference is significant against IPCA with 40, 60 and 95 factors. Adding a Gaussian kernel to the linear one gives the additive kernel, with nonlinear loadings that IPCA with the characteristics as instruments cannot represent. It adds another 60 points ($t=5.4$), almost entirely in the covariance part, but it does not lower the volatility of the GMV portfolio.

A decomposition traces the gain over IPCA with 60 factors. Dropping the rank restriction adds 8 points, replacing IPCA's weaker regularization of the monthly regressions by ours adds 15, the noise correction together with the floor adds 6, and the selection of the floor and the penalty adds 4. The noise correction helps only together with the floor. Alone it lowers the score by 10 points. A floor shaped like the noise covariance of the coefficients scores 9 points higher than one proportional to the identity matrix ($t=8.4$). Equal weights for the months matter mainly for the portfolios. With a weight that grows with the number of stocks, the score is not significantly lower, but the Sharpe ratio of the UMVE portfolio falls from $4.9$ to $3.4$.

The UMVE portfolios earn out-of-sample Sharpe ratios of 4.0 to 5.1, against 4.5 to 4.7 for IPCA. At unit gross exposure they earn 3.9 to 5.0 before and 3.2 to 4.1 after costs of 25 basis points, against 4.6 to 4.8 and 3.6 to 3.8 for IPCA. The costs exclude the borrowing costs of short positions and price impact. From these costs on, every additive specification has a higher Sharpe ratio than every linear and IPCA specification. The five factors of \citet{famafrench15} do not explain these returns. The Sharpe ratios of all models decline over time, from about seven in 1972--1989 to between two and three in 2010--2021. Without the smallest fifth of the stocks the Sharpe ratio of the linear kernel falls by a tenth.  

Our paper contributes to several strands of the literature. IPCA \citep{kellypruittsu19} specifies loadings linear in the characteristics. Autoencoders and deep networks make them nonlinear \citep{gukellyxiu20,fankeliaoneuhierl22,kellykuznetsovmalamudzhang26}. Sieve-based models \citep{connorhagmannlinton12,fanliaowang16,chenroussanovwang23}, the nonparametric mean model of \citet{freybergerneuhierlweber20}, shrinkage of characteristic-based portfolios \citep{kozaknagelsantosh20} and kernel methods \citep{kozak20} are further examples. These models target the first moment or the factor component. Their second moments, where available, are time-series moments of estimated factors, which the CMK represents when its features span the loadings (Remark~\ref{remIPCA}). Commercial risk models such as BARRA build the covariance matrix from cross-sectional regressions of returns on characteristics \citep{rosenberg74,mencheroorrwang11}, and \citet{famafrench20} compare such cross-section factor models with time-series ones. Large covariance estimators, by shrinkage \citep{ledoitwolf04,ledoitwolf20}, factor plus sparse structure \citep{fanliaomincheva13,FAN20195} or dynamic correlations \citep{engleledoitwolf19}, take the asset universe as fixed and return a matrix. The CMK returns a function, evaluated on whichever assets are present. The few joint estimators of conditional means and covariances \citep{gao_11,kirby18,clarkelinn24} proceed sequentially and do not guarantee PSD covariance matrices for changing cross sections. Inference for conditional factor models with large cross sections \citep{gagliardiniossolascaillet16,zaffaroni19,fortinetal24024,fortingagliardiniscaillet24} is complementary, as are the population theories of \citet{kozaknagel24} and \citet{filipovicschneider26fund}, for which the CMK is an estimable specification. On the methods side, the representer theorem behind the CMK program (Section~\ref{sec_jointestim}) specializes that of \citet{diwalejones18} to a loss that is quartic in the feature maps and adds an existence result. We estimate the PSD matrix of the kernel rather than a factorization of it, the reverse of the Burer--Monteiro approach to semidefinite programs \citep{boumalvoroninskibandeira16}. The Nystr\"om landmarks are selected by pivoted Cholesky decompositions \citep{HPS12,che_etal_23}, and the finite-sample bounds rest on the sample average approximation theory of \citet{mil_23}.

The paper is organized as follows. Section~\ref{sec_momentkernels} introduces the model and Section~\ref{sec_jointestim} the CMK program. Section~\ref{sec_estimatorNEW} specializes the CMK program and derives the CMK estimator. Section~\ref{sec_TP} gives the statistical guarantees, and Section~\ref{sec_factorSDF} derives the factor portfolios and the UMVE portfolio. Section~\ref{sec_empiricsNEW} contains the empirical analysis, and Section~\ref{sec_conclusion} concludes. The appendices contain the proofs and the solver, and the Online Appendix a simulation study and additional empirical results.

\section{Model}\label{sec_momentkernels}

We consider periods $(t,t+1]$, $t=0,1,2,\dots$, for example months. In period $(t,t+1]$ there are $N_t$ assets $i=1,\dots,N_t$ with returns $x_{t+1,i}$ and covariates $z_{t,i}$ observed at time $t$. We write $\E_t$ for the conditional expectation given the information at time $t$. The covariates take values in a fixed covariate space $\Zcal$. The returns may be gross or excess returns. We model the conditional first and second moments of the returns as functions $\mu:\Zcal\to\R$ and $q:\Zcal\times\Zcal\to\R$ of the covariates,
\begin{equation}\label{eqmodelmuq}
    \E_t[x_{t+1,i}] = \mu(z_{t,i}) ,\qquad
    \E_t[x_{t+1,i}x_{t+1,j}] = q(z_{t,i}, z_{t,j}) ,
\end{equation}
so that $\cov_t[x_{t+1,i},x_{t+1,j}]=q(z_{t,i},z_{t,j})-\mu(z_{t,i})\mu(z_{t,j})$. We stack the data of a period into the return vector $\bm x_{t+1}\coloneqq[x_{t+1,i}]_{i\le N_t}\in\R^{N_t}$ and the covariate array $\bm z_t\coloneqq[z_{t,i}]_{i\le N_t}\in\Zcal^{N_t}$, and we write $\mu(\bm z_t)\coloneqq[\mu(z_{t,i})]_i$ and $q(\bm z_t,\bm z_t^\top)\coloneqq[q(z_{t,i},z_{t,j})]_{i,j}$. We assume that distinct assets of a cross section have distinct covariates, $z_{t,i}\ne z_{t,j}$ for $i\ne j$, so that each asset has its own idiosyncratic component (Section~\ref{sec_syid}). Section~\ref{sec_mkdef} defines moment kernels, Section~\ref{sec_syid} splits them into systematic and idiosyncratic risk, and Section~\ref{sec_featuremap} represents them by feature maps.

\paragraph{Notation.} $\Sy^n$ is the space of symmetric $n\times n$ matrices and $\Sy^n_+$ the cone of PSD matrices. $\|\cdot\|_2$ is the Euclidean norm of a vector and the spectral norm of a matrix, $\|\cdot\|_F$ the Frobenius norm and $\langle\cdot,\cdot\rangle_F$ the Frobenius inner product. For a matrix $\bm A$, $\bm A^\top$, $\bm A^+$, $\Ima\bm A$ and $\rank\bm A$ are its transpose, Moore--Penrose pseudoinverse, column space and rank. For a vector $\bm v$, $\diag(\bm v)$ is the diagonal matrix with diagonal $\bm v$. For a square matrix $\bm A$, $\diag(\bm A)$ is the vector of its diagonal entries and $\Diag(\bm A)\coloneqq\diag(\diag(\bm A))$. $\Diag(a,\bm A)$ is the block-diagonal matrix with the blocks $a$ and $\bm A$. $\bm I_n$ is the identity matrix, written $\bm I$ where the dimension is clear, $\bm e_0$ the first standard basis vector, $\bm E_{00}\coloneqq\bm e_0\bm e_0^\top$, and $1_{\{\cdot\}}$ an indicator. For a linear operator $\Acal$ on a Hilbert space, $\|\Acal\|$ is its operator norm.

\subsection{Moment Kernels}\label{sec_mkdef}

For every cross section the matrix $q(\bm z_t,\bm z_t^\top)-\mu(\bm z_t)\mu(\bm z_t)^\top$ must be PSD. Kernel functions yield PSD matrices by definition. Recall that a symmetric function $k$ on $\Xcal\times\Xcal$ is a \emph{kernel function}, or positive-type function, if the matrix $[k(v_i,v_j)]_{i,j}$ is PSD for every finite array $(v_i)$ in $\Xcal$ \citep[Section 2.2]{pau_rag_16}. To cover the mean as well, we extend the covariate space by an auxiliary point $\Delta\notin\Zcal$, $\Zcal_\Delta\coloneqq\Zcal\cup\{\Delta\}$, and augment every return vector by the constant one and every covariate array by $\Delta$,
\begin{equation}\label{eqaugmented}
    \X_{t+1}\coloneqq \begin{bmatrix} 1 \\ \bm x_{t+1}\end{bmatrix}\in\R^{N_t+1},\qquad
    \Zt_t\coloneqq \begin{bmatrix} \Delta \\ \bm z_{t}\end{bmatrix}\in\Zcal_\Delta^{N_t+1},
\end{equation}
indexed by $i=0,1,\dots,N_t$.

\begin{definition}\label{defmomentkernel}
A \emph{moment kernel} is a kernel function $q$ on $\Zcal_\Delta\times\Zcal_\Delta$ with $q(\Delta,\Delta)=1$. Its \emph{mean function} and \emph{covariance function} are
\begin{equation}\label{eqmuc}
 \mu(z)\coloneqq q(z,\Delta),\qquad c(z,z')\coloneqq q(z,z')-\mu(z)\mu(z') ,\qquad z,z'\in\Zcal_\Delta .
\end{equation}
\end{definition}

The covariance function is the Schur complement of $q$ with respect to $\Delta$ and therefore a kernel function \citep[Theorem 4.5]{pau_rag_16}. Conversely, a function $\mu$ on $\Zcal$ and a kernel function $c$ on $\Zcal\times\Zcal$, extended by $\mu(\Delta)\coloneqq1$ and $c(\Delta,\cdot)\coloneqq0$, define the moment kernel $q\coloneqq c+\mu\otimes\mu$, with mean function $\mu$ and covariance function $c$. Moment kernels therefore correspond one to one to such pairs $(\mu,c)$. In matrix form the model \eqref{eqmodelmuq} reads
\begin{equation}\label{eqmodelaug}
   \E_t[\X_{t+1}\X_{t+1}^\top] = q(\Zt_t,\Zt_t^\top) = \begin{bmatrix} 1 & \mu(\bm z_t)^\top \\ \mu(\bm z_t) & q(\bm z_t,\bm z_t^\top)\end{bmatrix}.
\end{equation}
This is the CMK. The conditional mean $\bm\mu_t\coloneqq\mu(\bm z_t)$ is the first column of \eqref{eqmodelaug} below the $(0,0)$ entry, and the conditional covariance matrix $\bm\Sigma_t\coloneqq c(\bm z_t,\bm z_t^\top)$ is the Schur complement of its $(0,0)$ entry. Both are determined for every cross section, and $\bm\Sigma_t$ is PSD.

\subsection{Systematic and Idiosyncratic Risk}\label{sec_syid}

We decompose the moment kernel into a \emph{systematic} and an \emph{idiosyncratic} component,
\begin{equation}\label{eqsyid}
    q(z,z') = q^\sy(z,z') + q^\id(z,z') ,
\end{equation}
two kernel functions on $\Zcal_\Delta\times\Zcal_\Delta$. The systematic component satisfies $q^\sy(\Delta,\Delta)=1$ and carries the mean, $\mu(z)=q^\sy(z,\Delta)$, and the common covariation. The idiosyncratic component vanishes off the diagonal of $\Zcal\times\Zcal$ and at $\Delta$. The conditional covariance matrix of a cross section decomposes as $\bm\Sigma_t=\bm\Sigma^\sy_t+\bm\Sigma^\id_t$, with $\bm\Sigma^\sy_t\coloneqq q^\sy(\bm z_t,\bm z_t^\top)-\bm\mu_t\bm\mu_t^\top$ and the diagonal matrix $\bm\Sigma^\id_t\coloneqq q^\id(\bm z_t,\bm z_t^\top)$. The decomposition is identified when $q^\sy$ is determined by its values off the diagonal, for instance when $\Zcal$ has no isolated points and $q^\sy$ is continuous. Then $q^\id(z,z)$ is the jump of $q$ at the diagonal, the nugget effect of spatial statistics \citep{cressie93}.

\subsection{Feature Maps}\label{sec_featuremap}

Let $\Ccal$ be an infinite-dimensional separable Hilbert space, for example the space of square-summable sequences, with a unit vector $p$. A pair of \emph{feature maps} $h=(\hsy,\hid)$ consists of two maps $h^\tau:\Zcal\to\Ccal$, $\tau\in\{\sy,\id\}$, extended by $h^\tau(\Delta)\coloneqq0$. They define the moment kernel
\begin{equation}\label{qcondefNEW}
    q_{h}(z, z') \coloneqq \underbrace{\langle \hsy(z) + p 1_{z = \Delta}, \hsy(z') + p 1_{z' = \Delta} \rangle_\Ccal}_{\text{systematic component } q^{\sy}_h(z, z')} + \underbrace{\|\hid(z)\|_\Ccal^2 \, 1_{z = z'}}_{\text{idiosyncratic component } q^{\id}_h(z, z')} ,
\end{equation}
with $q_h(\Delta,\Delta)=\|p\|_\Ccal^2=1$. Its mean function is $\mu_h(z)=\langle\hsy(z),p\rangle_\Ccal$ for $z\in\Zcal$. Write $\hsy(z)=\mu_h(z)p+\hsy_\perp(z)$ with $\hsy_\perp(z)\perp p$. The systematic covariance is $\langle\hsy_\perp(z),\hsy_\perp(z')\rangle_\Ccal$.

\begin{lemma}\label{lemfeaturemap}
\begin{enumerate}
\item\label{lemfeaturemap1} Every moment kernel $q$ with the structure \eqref{eqsyid} is of the form $q=q_h$ if the reproducing kernel Hilbert space (RKHS) of its systematic component is separable.
\item\label{lemfeaturemap2} Let $O^\sy$ and $O^\id$ be linear isometries of $\Ccal$ with $O^\sy p=p$. Then $q_{h}=q_{\tilde h}$ for $\tilde h\coloneqq (O^\sy\hsy,O^\id\hid)$.
\end{enumerate}
\end{lemma}

The representation is general, but the feature maps are not identified, since isometries that fix $p$ leave the kernel unchanged. They are nuisance parameters. Section~\ref{sec_reparam} therefore parametrizes the estimation by a Gram matrix, which these isometries do not change.

\begin{remark}[Factor structure]\label{remfactor}
If $\hsy_\perp$ takes values in a $K$-dimensional subspace of $\Ccal$ with orthonormal basis $e_1,\dots,e_K$, then $\bm\Sigma^\sy_t=\bm\beta(\bm z_t)\bm\beta(\bm z_t)^\top$ with loadings $\beta_k(z)\coloneqq\langle\hsy_\perp(z),e_k\rangle_\Ccal$. This is a conditional factor structure in the sense of \citet{chamberlainrothschild83}. The number of factors is the rank of the systematic covariance kernel, and Section~\ref{sec_factorSDF} makes the factors observable as portfolios.
\end{remark}

\section{Estimation}\label{sec_jointestim}

We estimate the moment kernel by a least squares regression of the outer product of the augmented return vector on the kernel matrix. Section~\ref{sec_loss} defines the loss, and Section~\ref{sec_reparam} turns the estimation into a convex semidefinite program, the CMK program. Sections~\ref{sec_estimator} to~\ref{sec_uniqueness} give the estimated moments, the sufficient statistics, and existence and uniqueness.

\subsection{Loss Function}\label{sec_loss}

By \eqref{eqmodelaug}, the kernel matrix is the conditional expectation of the outer product of the augmented return vector. We measure the distance between the two by a weighted squared norm. Let $\bm\Omega_t\in\Sy^{N_t+1}_+$ be a \emph{weight matrix} that is a function of the cross section $\bm z_t$. The loss of a kernel $q$ at the \emph{data point} $\xi_t\coloneqq(\bm x_{t+1},\bm z_t)$ is
\begin{equation}\label{lossdef}
     \Lcal(q,\xi_{t})\coloneqq \big\| \bm\Omega_t^{1/2}\big(\X_{t+1}\X_{t+1}^\top - q(\Zt_t,\Zt_t^\top)\big)\bm\Omega_t^{1/2}\big\|_F^2 .
\end{equation}
Let $q_0$ be the true moment kernel. Since $\Om_t$ and the kernel matrices are known at time $t$, the law of iterated expectations gives
\begin{equation}\label{eqlossdecomp}
   \E[\Lcal(q,\xi_t)]=\E\big[\|\Om_t^{1/2}(q_0-q)(\Zt_t,\Zt_t^\top)\Om_t^{1/2}\|_F^2\big]+\E[\Lcal(q_0,\xi_t)] ,
\end{equation}
so the true kernel minimizes the expected loss for every weight. The \emph{Frobenius weight} $\Om_t=\bm I_{N_t+1}$ counts the squared error of the conditional mean twice, in the first row and in the first column, and that of the conditional second-moment matrix once. The \emph{scalar weight} $\Om_t=N_t^{-1/2}\bm I_{N_t+1}$, as in the scaled Frobenius norm of \citet{ledoitwolf04}, weights the months by their size. Section~\ref{sec_weights} derives the weight of the empirical analysis.

\subsection{Convex Reparametrization}\label{sec_reparam}

We take the feature maps from an RKHS. Let $\ksy$ and $\kid$ be scalar kernels with separable RKHS $\Gsy$ and $\Gid$ on $\Zcal$, and let $\Hcal^\tau\cong\Gcal^\tau\otimes\Ccal$ be the RKHS of $\Ccal$-valued functions with the operator-valued kernel $k^\tau(z,z')I_\Ccal$ \citep[Chapter 6]{pau_rag_16}. For a window of $T$ months $t=0,\dots,T-1$ and penalty parameters $\lsy,\lid\ge0$, the kernel ridge regression problem is
\begin{equation}\label{krr1old}
    \underset{h\in\Hsy\times\Hid }{\text{minimize } } \frac{1}{T}\sum_{t=0}^{T-1}\Lcal(q_h,\xi_{t})+\lsy \| \hsy\|_{\Hsy}^2+\lid \| \hid\|_{\Hid}^2 .
\end{equation}
It is not convex in $h$, because $\langle\hsy(z),\hsy(z')\rangle_\Ccal$ is neither convex nor concave in $\hsy$, and by Lemma~\ref{lemfeaturemap}\ref{lemfeaturemap2} its solutions are not unique. For $\lsy,\lid>0$ a representer theorem (Theorem~\ref{thmreprNEW} in Appendix~\ref{app_proofs}) shows that it has a solution and that every solution lies in the span of the kernel sections $k^\tau(\cdot,z_{t,i})$ at the sample covariates. This span grows with the size of the panel. We replace it by a subspace of small dimension $m^\tau$, spanned by an orthonormal system of \emph{features} $\bm\phi^\tau\coloneqq[\phi^\tau_1,\dots,\phi^\tau_{m^\tau}]$ in $\Gcal^\tau$,
\begin{equation}\label{eqorthonormal}
    \langle \phi^\tau_i,\phi^\tau_j\rangle_{\Gcal^\tau} = 1_{i=j},\qquad i,j=1,\dots,m^\tau,
\end{equation}
and we restrict \eqref{krr1old} to the feature maps $h^\tau_\phi(\cdot)=\bm\phi^\tau(\cdot)\bm\gamma^\tau$ with coefficients $\bm\gamma^\tau\in\Ccal^{m^\tau}$. A data-driven choice is the Nystr\"om method \citep{drineasmahoney05}. A pivoted Cholesky decomposition of the kernel matrix of the sample covariates selects $m^\tau$ \emph{landmark covariates} $\bm\upsilon^\tau=[\upsilon^\tau_1,\dots,\upsilon^\tau_{m^\tau}]$ \citep{HPS12,che_etal_23} and returns the orthonormal system
\begin{equation}\label{eqnystrom}
   \bm\phi^\tau(\cdot) \coloneqq k^\tau(\cdot,\bm\upsilon^{\tau\top})\,\bm\Upsilon^\tau,\qquad \bm\Upsilon^\tau\in\R^{m^\tau\times m^\tau},\ \bm\Upsilon^\tau{\bm\Upsilon^\tau}^\top = k^\tau(\bm\upsilon^\tau,\bm\upsilon^{\tau\top})^{-1},
\end{equation}
which spans the kernel sections at the landmarks. The restriction is exact when the landmarks exhaust the sample covariates. Any other orthonormal system may be used instead.

We extend the features by $\bm\phi^\tau(\Delta)\coloneqq\bm0$ and define the augmented features $\aphi^\sy(z)\coloneqq[1_{z=\Delta},\ \bm\phi^\sy(z)]$ and the feature matrices of a cross section,
\begin{equation}\label{eqPhit}
   \bm\phi^\tau_t\coloneqq\bm\phi^\tau(\bm z_t)\in\R^{N_t\times m^\tau},\qquad
   \bm\Phi^\sy_t\coloneqq \aphi^\sy(\Zt_t)=\begin{bmatrix} 1 & \bm 0^\top \\ \bm 0 & \bm\phi^\sy_t\end{bmatrix},\qquad
   \bm\Phi^\id_t\coloneqq \begin{bmatrix} \bm 0^\top \\ \bm\phi^\id_t\end{bmatrix}.
\end{equation}
The parameter space is the closed convex \emph{feasible set}
\begin{equation}\label{eqDcal}
   \Dcal \coloneqq \big\{\bm U=(\bm U^\sy,\bm U^\id)\in\Sy^{m^\sy+1}_+\times \Sy^{m^\id}_+ : \bm U^\sy_{00}=1\big\},
\end{equation}
with rows and columns of $\bm U^\sy$ indexed by $0,1,\dots,m^\sy$. A parameter $\bm U\in\Dcal$ defines the moment kernel $q_{\bm U}(z,z')\coloneqq\aphi^\sy(z)\bm U^\sy\aphi^\sy(z')^\top+\bm\phi^\id(z)\bm U^\id\bm\phi^\id(z')^\top1_{z=z'}$, whose matrix on a cross section is linear in $\bm U$,
\begin{equation}\label{eqFop}
   q_{\bm U}(\Zt_t,\Zt_t^\top) = \bm\Phi^\sy_t \bm U^\sy \bm\Phi^{\sy\top}_t + \Diag\big(\bm\Phi^\id_t \bm U^\id \bm\Phi^{\id\top}_t\big).
\end{equation}
The penalized loss in terms of $\bm U$ is
\begin{equation}\label{eqRcalU}
    \Rcal(\bm U,\xi_{t}) \coloneqq \Lcal(\bm U,\xi_{t})+ \lsy \tr( \bm U^\sy)+ \lid \tr( \bm U^\id),\qquad \Lcal(\bm U,\xi_t)\coloneqq\Lcal(q_{\bm U},\xi_t).
\end{equation}
The coefficients $\bm\gamma=(\bm\gamma^\sy,\bm\gamma^\id)$ have the Gram matrices
\begin{equation}\label{eqUdef}
     {\bm U}^\sy(\bm\gamma^\sy)  \coloneqq     \begin{bmatrix}
    1 & \langle p, {\bm\gamma^\sy}^\top  \rangle_\Ccal \\ \langle p, {\bm\gamma^\sy} \rangle_\Ccal & \langle  {\bm\gamma^\sy} , {\bm\gamma^\sy}^\top\rangle_\Ccal
\end{bmatrix}    , \qquad
{\bm U}^\id(\bm\gamma^\id)  \coloneqq \langle  {\bm\gamma^\id} , {\bm\gamma^\id}^\top\rangle_\Ccal ,
\end{equation}
where $\langle\bm\gamma,\bm\gamma^\top\rangle_\Ccal\coloneqq[\langle\gamma_i,\gamma_j\rangle_\Ccal]_{i,j}$, and $\bm U(\bm\gamma)\coloneqq(\bm U^\sy(\bm\gamma^\sy),\bm U^\id(\bm\gamma^\id))\in\Dcal$.

\begin{theorem}[Convex reparametrization]\label{thmUrepr}
Let $\lsy,\lid\ge 0$ and $h_\phi(\bm\gamma)\coloneqq(\bm\phi^\sy\bm\gamma^\sy,\bm\phi^\id\bm\gamma^\id)$.
\begin{enumerate}
\item\label{thmUrepr1} For every $\bm\gamma$, $q_{h_\phi(\bm\gamma)}=q_{\bm U(\bm\gamma)}$, and the objective of \eqref{krr1old} at $h_\phi(\bm\gamma)$ equals $\frac1T\sum_t\Rcal(\bm U(\bm\gamma),\xi_t)-\lsy$.
\item\label{thmUrepr2} The map $\bm\gamma\mapsto\bm U(\bm\gamma)$ is onto $\Dcal$.
\item\label{thmUrepr3} Hence $h_\phi(\bm\gamma)$ solves \eqref{krr1old} restricted to the features if and only if $\bm U(\bm\gamma)$ solves the convex problem
\begin{equation}\label{krr1convex}
    \underset{\bm U\in\Dcal }{\text{minimize } } \Rcal_T(\bm U)\coloneqq \frac{1}{T}\sum_{t=0}^{T-1}\Rcal(\bm U,\xi_{t}),
\end{equation}
and the two problems induce the same moment kernels.
\end{enumerate}
\end{theorem}

The theorem replaces the nonconvex problem in the nuisance parameters by a convex quadratic semidefinite program in the rotation-free parameter $\bm U$. We call \eqref{krr1convex} the \emph{CMK program}. Orthonormality of the features turns the ridge penalty into the linear trace penalty, $\|\bm\phi^\sy\bm\gamma^\sy\|^2_{\Hsy}=\tr\bm U^\sy-1$ and $\|\bm\phi^\id\bm\gamma^\id\|^2_{\Hid}=\tr\bm U^\id$. By part \ref{thmUrepr2}, the semidefinite constraints ensure that $\bm U$ is the Gram matrix of some coefficients $\bm\gamma$.

\subsection{Estimated Moments}\label{sec_estimator}

We write the systematic parameter in block form,
\begin{equation}\label{eqUblock}
   \bm U^\sy=\begin{bmatrix} 1 & \muf^\top \\ \muf & \bm V \end{bmatrix},\qquad \bm C\coloneqq\bm V-\muf\muf^\top,
\end{equation}
so that $\bm U^\sy\succeq0$ if and only if the Schur complement $\bm C$ is PSD. The kernel $q_{\bm U}$ has the mean function $\bm\phi^\sy(z)\muf$ and the covariance function
\begin{equation}\label{qcondefEST}
  c_{\bm U}(z,z') = \bm\phi^\sy(z)\bm C\bm\phi^\sy(z')^\top    + \bm\phi^\id(z)\bm U^\id\bm\phi^\id(z')^\top 1_{z=z'} .
\end{equation}
A solution $\bm U$ of the CMK program \eqref{krr1convex} gives on a cross section the conditional moments of $q_{\bm U}$, the mean, the covariance matrix and the second-moment matrix
\begin{equation}\label{eqmutSigmatest}
     \bm \mu_t = \bm\phi^\sy_t\muf,\qquad
    \bm \Sigma_t = \underbrace{\bm\phi^\sy_t\bm C\bm\phi^{\sy\top}_t}_{\bm\Sigma^\sy_t} + \underbrace{\Diag (\bm\phi^\id_t \bm U^\id \bm\phi^{\id\top}_t )}_{\bm\Sigma^\id_t},\qquad
    \bm Q_t = \bm\Sigma_t+\bm\mu_t\bm\mu_t^\top .
\end{equation}
The rank of $\bm\Sigma^\sy_t$ is at most $\rank\bm C$ for every cross section, in line with Remark~\ref{remfactor}.

\begin{example}[Isotropic idiosyncratic specification]\label{ex:simplediagonal}
The empirical analysis uses $m^\id=1$, the constant kernel $\kid\equiv1$, whose RKHS consists of the constant functions, the orthonormal feature $\bm\phi^\id\equiv1$, and $\bm U^\id = u^\id\ge0$. Then $\bm\Sigma_t^\id = u^\id\bm I_{N_t}$. Online Appendix~\ref{app_emp} estimates a specification in which the idiosyncratic variance depends on the characteristics.
\end{example}

\subsection{Sufficient Statistics}\label{sec_scalability}

Each cross section enters the loss through a few statistics whose size does not depend on $N_t$. We state this for the isotropic specification. Lemma~\ref{lemsuffstatgen} in Appendix~\ref{app_proofs} gives the general case. Let $\bm J_t\coloneqq\Diag(0,\bm I_{N_t})$, so that $q_{\bm U}(\Zt_t,\Zt_t^\top)=\bm\Phi^\sy_t\bm U^\sy\bm\Phi^{\sy\top}_t+u^\id\bm J_t$.

\begin{lemma}[Sufficient statistics]\label{lemsuffstat}
In the isotropic specification, with the statistics
\begin{equation}\label{eqsuffstatdef}
\begin{gathered}
   a_t\coloneqq\X_{t+1}^\top\Om_t\X_{t+1},\qquad \bm s_t\coloneqq\bm\Phi^{\sy\top}_t\Om_t\X_{t+1},\qquad \bm M_t\coloneqq \bm\Phi^{\sy\top}_t\Om_t\bm\Phi^\sy_t,\\
   r_t\coloneqq\X_{t+1}^\top\Om_t\bm J_t\Om_t\X_{t+1},\qquad \bm D_t\coloneqq\bm\Phi^{\sy\top}_t\Om_t\bm J_t\Om_t\bm\Phi^\sy_t,\qquad n_t\coloneqq\tr(\Om_t\bm J_t\Om_t\bm J_t),
\end{gathered}
\end{equation}
the loss is
\begin{equation}\label{eqsuffstatiso}
   \Lcal(\bm U,\xi_t) = a_t^2 - 2\,\bm s_t^\top\bm U^\sy\bm s_t - 2u^\id r_t + \tr(\bm U^\sy\bm M_t\bm U^\sy\bm M_t) + 2u^\id\tr(\bm U^\sy\bm D_t) + n_t (u^\id)^2 .
\end{equation}
For the Frobenius weight, $\bm M_t=\bm\Phi^{\sy\top}_t\bm\Phi^\sy_t$, $\bm D_t=\bm M_t-\bm E_{00}$ and $n_t=N_t$.
\end{lemma}

The statistics are computed in one pass over the panel, at a cost of $O(N_t(m^\sy)^2)$ per month for the weights of this paper. The CMK program then has $O((m^\sy)^2)$ variables, whatever the number of stocks, and its coefficients are sums of the monthly statistics.

\subsection{Existence and Uniqueness}\label{sec_uniqueness}

We regard the parameters as elements of the Hilbert space $\Sm\coloneqq\Sy^{m^\sy+1}\times\Sy^{m^\id}$ with $\langle\bm U,\bm U'\rangle\coloneqq\tr(\bm U^\sy\bm U'^\sy)+\tr(\bm U^\id\bm U'^\id)$ and norm $\|\bm U\|$. A data point $\xi=(\bm x,\bm z)$ has a cross section of some size $N$, augmented arrays $\X$ and $\Zt$, and weight $\Om$, which we assume positive definite. The penalized loss is a convex quadratic on $\Sm$ (Lemma~\ref{lemquadratic} in Appendix~\ref{app_proofs}),
\begin{equation}\label{eqRquadratic}
   \Rcal(\bm U,\xi)=\tfrac12\langle\bm U,\Acal(\xi)\bm U\rangle+\langle\bm B(\xi),\bm U\rangle+\Lcal(\bm 0,\xi) ,
\end{equation}
with the PSD Hessian $\langle\bm U,\Acal(\xi)\bm U\rangle=2\|\Om^{1/2}q_{\bm U}(\Zt,\Zt^\top)\Om^{1/2}\|_F^2$ and $\bm B(\xi)$ the gradient at $\bm 0$, which contains the gradient $\bm\Lambda\coloneqq(\lsy\bm I_{m^\sy+1},\lid\bm I_{m^\id})$ of the trace penalty. Every point of $\Dcal$ has $\bm U^\sy_{00}=1$, so the differences of points of $\Dcal$ lie in $\Sm^\circ\coloneqq\{\bm U\in\Sm:\bm U^\sy_{00}=0\}$. We call a convex quadratic on $\Sm$ with Hessian $\Acal$ \emph{$\alpha$-strongly convex} if $\langle\bm U,\Acal\bm U\rangle\ge\alpha\|\bm U\|^2$ for all $\bm U\in\Sm^\circ$, and strongly convex if this holds for some $\alpha>0$. A strongly convex quadratic has a unique minimizer on every nonempty closed convex subset of $\Dcal$. The definition applies to $\Rcal(\cdot,\xi)$, to $\Rcal_T$ and to the population loss of Section~\ref{sec_TP}. The curvature of the sample objective is that of the linear map
\begin{equation}\label{eqsamplemap}
    \bm U\mapsto\big(q_{\bm U}(\Zt_t,\Zt_t^\top)\big)_{t=0}^{T-1}.
\end{equation}
Its null space is the intersection of the null spaces of the single-month maps $\bm U\mapsto q_{\bm U}(\Zt_t,\Zt_t^\top)$, so cross sections that are individually too small can jointly identify $\bm U$. The map is injective on $\Sm$ if and only if it is injective on $\Sm^\circ$, since $q_{\bm U}(\Delta,\Delta)=\bm U^\sy_{00}$.

\begin{proposition}[Existence and uniqueness]\label{coruniqueness}
Let $\lsy,\lid\ge0$.
\begin{enumerate}
\item\label{corun1} The CMK program \eqref{krr1convex} has a solution if $\lsy,\lid>0$ or if the map \eqref{eqsamplemap} is injective.
\item\label{corun2} All solutions produce the same in-sample conditional moment matrices $q_{\bm U}(\Zt_t,\Zt_t^\top)$, $t<T$.
\item\label{corun3} If the map \eqref{eqsamplemap} is injective, the solution and the estimated moment kernel are unique, and $\Rcal_T$ is strongly convex.
\end{enumerate}
\end{proposition}

The trace penalty is linear in $\bm U$ and adds no curvature, so uniqueness rests on the cross sections alone. In the isotropic specification a single month identifies $\bm U$ if and only if its feature matrix has full column rank and $N_t>m^\sy$ (Lemma~\ref{lemstrongconvex} in Appendix~\ref{app_proofs}). The rank condition fails when some features are collinear within a month, as in the early years of the empirical analysis, when some characteristics are not yet available. Injectivity then depends on whether the row spaces of the feature matrices vary enough across months. It fails if a direction $\bm v$ has $\bm\phi^\sy_t\bm v=\bm0$ in every month of the window, as for a characteristic that is missing throughout.

\section{Weighting and Regularization}\label{sec_estimatorNEW}

Four choices specialize the CMK program: the weight, a Tikhonov regularization of the single-month estimates, the spectral floor and the trace penalty. Section~\ref{sec_twostage} approximates the resulting program by a two-stage closed form, the CMK estimator. Throughout, the specification is isotropic and $N_t>m^\sy$ for every month. We write $\mathrm{df}_t\coloneqq N_t-\rank\bm\phi^\sy_t$ for the \emph{residual degrees of freedom} of month $t$.

\subsection{Weighting the Months}\label{sec_weights}

Let $\bm G_t\coloneqq\bm\phi^{\sy\top}_t\bm\phi^\sy_t$ be the Gram matrix of the features, $\bm P_t\coloneqq\bm\phi^\sy_t\bm\phi^{\sy+}_t$ the orthogonal projector on $\Ima\bm\phi^\sy_t$, and $\bm\Pi_t\coloneqq\bm\phi^{\sy+}_t\bm\phi^\sy_t$ the projector on the row space. Completing the square in \eqref{eqsuffstatiso} gives, for any weight,
\begin{equation}\label{eqUhat}
   \Lcal(\bm U,\xi_t)=\big\|\bm M_t^{1/2}\big(\bm U^\sy-\hat{\bm U}_t\big)\bm M_t^{1/2}\big\|_F^2+\text{terms without }\bm U^\sy,\qquad
   \hat{\bm U}_t\coloneqq\bm M_t^+\big(\bm s_t\bm s_t^\top-u^\id\bm D_t\big)\bm M_t^+
\end{equation}
(Appendix~\ref{app_prooflemblocks}). The matrix $\hat{\bm U}_t$ is the estimate of $\bm U^\sy$ that month $t$ provides on its own, given $u^\id$. The sample loss combines the $T$ single-month estimates, each measured in the norm that the weight induces. Under the Frobenius weight, $\bm M_t=\Diag(1,\bm G_t)$, $\bm D_t=\Diag(0,\bm G_t)$ and
\begin{equation}\label{eqUhatgw}
    \hat{\bm U}_t=\begin{bmatrix}1&\bm f_{t+1}^\top\\ \bm f_{t+1}&\bm f_{t+1}\bm f_{t+1}^\top-u^\id\bm G_t^+\end{bmatrix},\qquad
    \bm f_{t+1}\coloneqq\bm\phi^{\sy+}_t\bm x_{t+1} .
\end{equation}
The vector $\bm f_{t+1}$ holds the coefficients of the cross-sectional least squares regression of the returns on the features. The term $u^\id\bm G_t^+=\bm\phi^{\sy+}_t(u^\id\bm I_{N_t})\bm\phi^{\sy+\top}_t$ is the idiosyncratic covariance matrix in these coordinates.

The weight of a month should reflect the error of its single-month estimate $\hat{\bm U}_t$. Under the model with parameter $(\muf_0,\bm V_0,u^\id_0)$ and $\bm C_0\coloneqq\bm V_0-\muf_0\muf_0^\top$,
\begin{equation}\label{eqcovf}
   \cov_t[\bm f_{t+1}]=\bm\Pi_t\bm C_0\bm\Pi_t+u^\id_0\bm G_t^+ .
\end{equation}
The first term is the systematic covariance in the coordinates of the features and does not shrink with $N_t$. The second is of order $1/N_t$ when the cross sections are \emph{regular}, that is, when the nonzero eigenvalues of $\bm G_t/N_t$ are bounded away from zero. In the directions with systematic variance the single-month estimates therefore have errors of the same order whatever the size of the cross section, and every month should get the same weight there. The Frobenius weight instead lets the weight of a month grow with $N_t$ in the mean term and with $N_t^2$ in the second-moment term, since $\bm G_t$ grows with $N_t$. The following weight gives every month the same weight. It makes $\bm M_t$ the identity on the row space of the features, and it keeps the identity on the complement of the feature span, where the loss measures the residual variance.

\begin{lemma}[Gram-whitened loss]\label{lemblocks}
For the weight
\begin{equation}\label{eqomegagw}
   \Om^{\rm gw}_t\coloneqq\Diag\Big(\kappa,\ \bm\phi^\sy_t(\bm G_t^+)^2\bm\phi^{\sy\top}_t+\bm I_{N_t}-\bm P_t\Big),\qquad \kappa>0,
\end{equation}
the statistics of Lemma~\ref{lemsuffstat} are $\bm M_t=\Diag(\kappa,\bm\Pi_t)$, $\bm D_t=\Diag(0,\bm G_t^+)$ and $\bm M_t^+\bm s_t=(1,\bm f_{t+1})$, and for $\bm U\in\Dcal$ the loss is, up to terms without $\bm U$,
\begin{equation}\label{eqlossgram}
   \Lcal^{\rm gw}(\bm U,\xi_t)= 2\kappa\,\|\bm\Pi_t(\bm f_{t+1}-\muf)\|_2^2+\big\|\bm\Pi_t(\bm f_{t+1}\bm f_{t+1}^\top-\bm V)\bm\Pi_t-u^\id\bm G_t^+\big\|_F^2
   +\mathrm{df}_t\,(u^\id-\hat u_t)^2 ,
\end{equation}
with $\hat u_t\coloneqq\|(\bm I_{N_t}-\bm P_t)\bm x_{t+1}\|_2^2/\mathrm{df}_t$ the residual variance of the month. The weight is positive definite.
\end{lemma}

We call \eqref{eqomegagw} the \emph{Gram-whitened weight}. The loss \eqref{eqlossgram} has three terms, the mean term, the second-moment term and the residual term. Each month enters the first two with the same weight, whatever its size, and $\kappa$ balances them. The residual term counts the residual degrees of freedom of a month, as under the Frobenius weight. The matrix $\bm\phi^\sy_t(\bm G_t^+)^2\bm\phi^{\sy\top}_t$ is the pseudoinverse of the kernel matrix $\bm\phi^\sy_t\bm\phi^{\sy\top}_t$. The Gram-whitened weight treats all directions of the feature span alike. The generalized least squares weight of \citet{browne74}, the \emph{Browne weight}, accounts for the dependence of the error variance on the direction and is efficient under Gaussian returns. Section~\ref{sec_robustness} compares the two.

\subsection{Tikhonov Regularization}\label{sec_tikhonov}

By \eqref{eqcovf}, the coordinate of $\bm f_{t+1}$ in the direction of the $k$th eigenvector $\bm v_k$ of $\bm G_t$ with $\lambda_k(\bm G_t)>0$ has conditional variance $\bm v_k^\top\bm C_0\bm v_k+u^\id_0/\lambda_k(\bm G_t)$. A single month determines the coordinates poorly in the \emph{weak directions}, those with a small eigenvalue. We compute the single-month estimates with the \emph{Tikhonov-regularized pseudoinverse} \citep{hoerlkennard70,hansen98}
\begin{equation}\label{eqtikhonov}
   \bm\phi^{\sy+}_{t,\varrho}\coloneqq(\bm G_t+\varrho_t\bm I_{m^\sy})^{-1}\bm\phi^{\sy\top}_t,\qquad \bm G^+_{t,\varrho}\coloneqq\bm\phi^{\sy+}_{t,\varrho}\bm\phi^{\sy+\top}_{t,\varrho},\qquad \varrho_t\coloneqq\varrho\,\tr(\bm G_t)/m^\sy,
\end{equation}
for $\varrho>0$, and set $\bm\phi^{\sy+}_{t,0}\coloneqq\bm\phi^{\sy+}_t$ and $\bm G^+_{t,0}\coloneqq\bm G^+_t$, the limits as $\varrho\to0$. In the direction $\bm v_k$ the regularized coordinate is the exact one times $\lambda_k(\bm G_t)/(\lambda_k(\bm G_t)+\varrho_t)$. The factor is near one for eigenvalues much larger than $\varrho_t$, near zero for much smaller ones, and zero on the null directions, as for the pseudoinverse. The empirical analysis uses $\varrho=10^{-3}$ in every window, so that $\varrho_t$ is a thousandth of the average eigenvalue of $\bm G_t$ and the construction does not depend on the scale of the features.

Let $\bm f^\varrho_{t+1}\coloneqq\bm\phi^{\sy+}_{t,\varrho}\bm x_{t+1}$ be the \emph{regularized coordinates}. For $\varrho>0$ we set $\mathrm{df}_t\coloneqq N_t-m^\sy$ instead of $N_t-\rank\bm\phi^\sy_t$, and the two agree in months of full rank. Let $\hat u_{t,\varrho}\coloneqq\bm x_{t+1}^\top(\bm I_{N_t}-\bm\phi^\sy_t\bm\phi^{\sy+}_{t,\varrho})\bm x_{t+1}/\mathrm{df}_t$ be the regularized residual variance. The \emph{regularized Gram-whitened loss} is
\begin{equation}\label{eqlossrho}
   \Lcal^{\rm gw}_\varrho(\bm U,\xi_t)\coloneqq 2\kappa\,\|\bm f^\varrho_{t+1}-\muf\|_2^2+\big\|\bm f^\varrho_{t+1}\bm f^{\varrho\top}_{t+1}-\bm V-u^\id\bm G^+_{t,\varrho}\big\|_F^2+\mathrm{df}_t\,\big(u^\id-\hat u_{t,\varrho}\big)^2 .
\end{equation}
Unlike \eqref{eqlossgram} it contains no projection $\bm\Pi_t$. In the null directions of $\bm G_t$ the regularized coordinates are zero, and the loss pulls $\muf$ and $\bm V$ toward zero there. For $\varrho=0$ and months of full rank it is \eqref{eqlossgram}. For $\varrho>0$ it is not a weighted loss \eqref{lossdef}. On $\Dcal$ it has the form \eqref{eqsuffstatiso}, up to a constant, with the statistics $\bm M_t=\Diag(\kappa,\bm I_{m^\sy})$, $\bm s_t=(\kappa,\bm f^\varrho_{t+1})$, $\bm D_t=\Diag(0,\bm G^+_{t,\varrho})$, $r_t=\bm f^{\varrho\top}_{t+1}\bm G^+_{t,\varrho}\bm f^\varrho_{t+1}+\mathrm{df}_t\hat u_{t,\varrho}$ and $n_t=\|\bm G^+_{t,\varrho}\|_F^2+\mathrm{df}_t$. Its quadratic part is positive definite on $\Sm^\circ$, so the CMK program with this loss is strictly convex in every window, also in months with rank-deficient features. The regularized coordinates are biased toward zero in the weak directions (Section~\ref{sec_limitcmk}).

\subsection{Spectral Floor}\label{sec_regularization}

The sample loss estimates $\bm V$ from the second-moment blocks $\bm f^\varrho_{t+1}\bm f^{\varrho\top}_{t+1}-u^\id\bm G^+_{t,\varrho}$ of the single-month estimates. In the direction $\bm v_k$ of a month the block is the difference of two quantities of order $u^\id_0/\lambda_k(\bm G_t)$, and so is its error. Averaging over the window reduces this error, but the smallest eigenvalues of the average are biased downward. When $m^\sy$ is of the order of $T$, the sample second moment of $\bm f_{t+1}$ has eigenvalues spread far below and above those of its expectation, whatever the true $\bm V$. The unconstrained minimizer of the loss then has a matrix $\bm C$ with negative eigenvalues in the weak directions, and the constraint $\bm U^\sy\succeq\bm0$ sets them to zero. The estimate assigns zero systematic variance to these directions, although the data only show that it is below the noise level. Every quantity that involves $\bm\Sigma_t^{-1}$, such as the portfolio weights of Section~\ref{sec_factorSDF}, takes this zero at face value. We therefore bound the systematic covariance from below at the scale of the noise. The idiosyncratic covariance in the coordinates of the features, averaged over the window, is
\begin{equation}\label{eqXi}
   \hat{\bm F}\coloneqq\hat u\,\frac1T\sum_{t=0}^{T-1}\bm G^+_{t,\varrho},\qquad \hat u\coloneqq\frac{\sum_{t=0}^{T-1}\bm x_{t+1}^\top(\bm I_{N_t}-\bm\phi^\sy_t\bm\phi^{\sy+}_{t,\varrho})\bm x_{t+1}}{\sum_{t=0}^{T-1}\mathrm{df}_t} ,
\end{equation}
where $\hat u$, the pooled residual variance, minimizes the residual term of \eqref{eqlossrho} summed over the window. For $\eta\ge0$ we replace $\Dcal$ by
\begin{equation}\label{eqDgamma}
   \Dcal_\eta\coloneqq\big\{\bm U\in\Dcal:\ \bm C\succeq\eta\,\hat{\bm F}\big\}
   =\big\{\bm U\in\Dcal:\ \bm U^\sy-\bm F_\eta\succeq\bm0\big\},\qquad \bm F_\eta\coloneqq\Diag(0,\eta\hat{\bm F}),
\end{equation}
where the second form follows from the Schur complement. We call $\eta\hat{\bm F}$ the \emph{spectral floor}. The floor is a statistic of the window, and it vanishes as the cross sections grow: $\hat{\bm F}$ is of order $1/N$ for regular cross sections and bounded $\hat u$. For $\varrho=0$ and a Gram matrix that is constant over the window the floor says $\bm\Sigma^\sy_t\succeq\eta\hat u\bm P_t$: the systematic covariance dominates $\eta$ times the idiosyncratic variance in every direction of the feature span. A single-month estimate errs there by at least the idiosyncratic variance. For $\varrho>0$ the floor has, in the direction of the $k$th eigenvector of the constant Gram matrix $\bm G$, the eigenvalue $\eta\hat u\lambda_k(\bm G)/(\lambda_k(\bm G)+\varrho_t)^2$, which is $\eta\hat u/\lambda_k(\bm G)$ times $(\lambda_k(\bm G)/(\lambda_k(\bm G)+\varrho_t))^2$. This factor is near one in the strong directions and near zero in the directions with $\lambda_k(\bm G)$ much smaller than $\varrho_t$, where the regularized coordinates are themselves near zero. For nonsingular Gram matrices and $\varrho=0$, a reparametrization $\bm\phi^\sy\mapsto\bm\phi^\sy\bm A$ with an invertible $\bm A$ transforms $\bm C$ and $\hat{\bm F}$ alike, so $\Dcal_\eta$ describes the same moment kernels in every parametrization, and $\eta$ has the same meaning in each. Under orthogonal $\bm A$ the weight, the penalty, the floor and the Tikhonov regularization are unchanged, so the estimators of Section~\ref{sec_twostage} depend on the feature space and not on its orthonormal basis.

A simpler floor is proportional to the identity matrix. The \emph{identity-shaped floor} $\eta(\tr\hat{\bm F}/m^\sy)\bm I_{m^\sy}$ has the same trace, but it is invariant only under orthogonal reparametrizations. On a cross section it is proportional to the kernel matrix of the features, large in the strong directions, where the data are informative, and small in the weak ones. Section~\ref{sec_robustness} compares the two shapes.

\subsection{Trace Penalty}\label{sec_penalty}

The trace penalty acts as a convex penalty on the rank. On $\Dcal_\eta$ write $\bm C=\eta\hat{\bm F}+\tilde{\bm C}$ with the \emph{excess} $\tilde{\bm C}\succeq\bm0$ of the systematic covariance over the floor. Then
\begin{equation}\label{eqpenaltydecomp}
   \tr(\bm U^\sy)-1 = \|\muf\|_2^2 + \eta\tr\hat{\bm F}+\tr\tilde{\bm C} .
\end{equation}
The trace penalty is thus a ridge penalty on the mean plus the trace of the excess, plus a constant. On PSD matrices the trace equals the nuclear norm, the convex relaxation of the rank \citep{rechtfazelparrilo10}. The floor sets the level below which the data carry little information about a direction. The penalty sets the excess to zero where it does not reach a threshold. We call the floor and the penalty the \emph{regularizers}. For $\eta>0$ and $\hat{\bm F}$ positive definite, which holds when $\hat u>0$ and the row spaces of the months span $\R^{m^\sy}$, the estimate of $\bm C$ has full rank. We therefore count the factors by the rank of the excess, the number of directions whose systematic variance exceeds the floor. We report it as the \emph{number of factors}. For $\eta=0$ the excess is $\bm C$, as in Remark~\ref{remfactor}. Proposition~\ref{proprank} in Appendix~\ref{app_proofproprank} gives the Karush--Kuhn--Tucker (KKT) conditions and bounds the number of factors. The penalty sets eigenvalues of the excess exactly to zero, as the lasso sets coefficients to zero. It also guarantees a solution of the CMK program when the cross sections do not identify $\bm U$ (Proposition~\ref{coruniqueness}).

\subsection{CMK Estimator}\label{sec_twostage}

The CMK program with the four choices is
\begin{equation}\label{krr1floor}
    \underset{\bm U\in\Dcal_\eta }{\text{minimize } } \frac1T\sum_{t=0}^{T-1}\Lcal^{\rm gw}_\varrho(\bm U,\xi_t)+\lsy\tr\bm U^\sy .
\end{equation}
We call its solution the \emph{one-step estimator}. This program nearly separates. The residual term involves $u^\id$ only, and its minimizer is $\hat u$. The second-moment term involves $u^\id$ as well, through $\bm G^+_{t,\varrho}$, which varies across months. Fixing $u^\id=\hat u$ and the mean coefficient at the average $\bar{\bm f}$ below gives a closed form. Let
\begin{equation}\label{eqAbar}
   \bar{\bm f}\coloneqq\frac1T\sum_{t<T}\bm f^\varrho_{t+1},\qquad \bar{\bm A}\coloneqq\frac1T\sum_{t<T}\bm f^\varrho_{t+1}\bm f^{\varrho\top}_{t+1}-\hat{\bm F} ,
\end{equation}
the mean of the regularized coordinates and their second moment net of the idiosyncratic noise. For a general $u^\id$ write $\bar{\bm A}(u^\id)\coloneqq\frac1T\sum_{t<T}\bm f^\varrho_{t+1}\bm f^{\varrho\top}_{t+1}-u^\id\frac1T\sum_{t<T}\bm G^+_{t,\varrho}$, so that $\bar{\bm A}=\bar{\bm A}(\hat u)$. At $u^\id=\hat u$ the sample loss is $2\kappa\|\muf-\bar{\bm f}\|_2^2+\|\bm V-\bar{\bm A}\|_F^2$ up to a constant.

\begin{proposition}[CMK estimator]\label{propclosedform}
Fix $u^\id=\hat u$ and $\muf=\bar{\bm f}$, and write $\bm V=\bar{\bm f}\bar{\bm f}^\top+\eta\hat{\bm F}+\tilde{\bm C}$. Program \eqref{krr1floor} reduces to $\min_{\tilde{\bm C}\succeq\bm0}\|\tilde{\bm C}-(\bar{\bm A}-\bar{\bm f}\bar{\bm f}^\top-\eta\hat{\bm F})\|_F^2+\lsy\tr\tilde{\bm C}$, whose unique solution is
\begin{equation}\label{eqclosedform}
   \tilde{\bm C}=\Pcal_+\big(\bar{\bm A}-\bar{\bm f}\bar{\bm f}^\top-\eta\hat{\bm F}-\tfrac{\lsy}{2}\bm I_{m^\sy}\big) ,
\end{equation}
with $\Pcal_+$ the projection on the PSD cone. Let $(\check{\muf},\check{\bm V},\check u)$ be the one-step estimator. Then $\check{\bm V}-\bar{\bm A}(\check u)+\frac{\lsy}{2}\bm I$ is PSD and $\check{\muf}=\big(\bm I+(\check{\bm V}-\bar{\bm A}(\check u)+\frac{\lsy}{2}\bm I)/\kappa\big)^{-1}\bar{\bm f}$. If its excess $\check{\bm V}-\check{\muf}\check{\muf}^\top-\eta\hat{\bm F}$ has full rank and $\check u=\hat u$, then $\check{\muf}=\bar{\bm f}$ and $\check{\bm V}-\check{\muf}\check{\muf}^\top=\eta\hat{\bm F}+\tilde{\bm C}$ with $\tilde{\bm C}$ of \eqref{eqclosedform}.
\end{proposition}

The \emph{CMK estimator} is the two-stage closed form $\muf=\bar{\bm f}$, $\bm C=\eta\hat{\bm F}+\tilde{\bm C}$ and $u^\id=\hat u$. The penalty is a soft threshold at $\lsy/2$ on the eigenvalues of $\bar{\bm A}-\bar{\bm f}\bar{\bm f}^\top-\eta\hat{\bm F}$, and the number of factors is the number of eigenvalues above the threshold. The mean coefficient $\bar{\bm f}$ is the average of the monthly cross-sectional regression coefficients, as in \citet{famamacbeth73}, Tikhonov-regularized for $\varrho>0$, and is the same for every floor and penalty. Mean and covariance come from one kernel on the same features, and the noise correction and the floor use the same monthly regressions. The CMK estimator does not depend on $\kappa$. The one-step estimator differs from it in two respects. It estimates $u^\id$ also from the second-moment term (Corollary~\ref{corisou} in Appendix~\ref{app_proofproprank}), and, when the excess is rank deficient, it can shrink the mean toward zero in the null directions of the excess. The empirical analysis uses the CMK estimator because it is a closed form. The selection of the hyper-parameters in every window then needs no iterative solver, and Section~\ref{sec_limitcmk} gives its limit. Out of sample it fits at least as well as the one-step estimator. In the comparison of Online Appendix~\ref{app_emp}, for the linear kernel in the full feature space with $\eta=1$ and $\lsy=0$ on every fourth window, the two estimators score the same from 1985 on. Before 1985, with rank-deficient months and many weak directions, the score differential (Section~\ref{sec_predictiveNEW}) of the CMK estimator is 42 points higher than that of the one-step estimator in the median month.

\begin{remark}[IPCA]\label{remIPCA}
IPCA \citep{kellypruittsu19} with the features as instruments specifies $\bm x_{t+1}=\bm\phi^\sy_t\bm\Gamma_\beta\bm g_{t+1}$ plus noise, with a loading matrix $\bm\Gamma_\beta\in\R^{m^\sy\times K}$ and latent factors $\bm g_{t+1}\in\R^K$. It minimizes the sum of squared residuals over $\bm\Gamma_\beta$ and the factors, alternating a factor step, cross-sectional regressions for the factors given the loadings, with an update of the loadings. The objective does not change when $\bm\Gamma_\beta$ is multiplied from the right by a nonsingular matrix and the factors from the left by its inverse, and IPCA selects one solution by a normalization. The moments implied by IPCA are those of the parameter
\[
   \bm U^\sy=\begin{bmatrix}1&\bar{\bm g}^\top\bm\Gamma_\beta^\top\\ \bm\Gamma_\beta\bar{\bm g}&\bm\Gamma_\beta\big(\widehat{\cov}(\bm g)+\bar{\bm g}\bar{\bm g}^\top\big)\bm\Gamma_\beta^\top\end{bmatrix}\in\Sy^{m^\sy+1}_+,
\]
with $\bar{\bm g}$ and $\widehat{\cov}(\bm g)$ the sample mean and the sample covariance matrix of the factors with divisor $T$. This parameter does not depend on the normalization. We call its computation from the estimated factors the \emph{moment step} of IPCA. The rank restriction is $\rank\bm C\le K$, which is nonconvex for $K<m^\sy$, and the trace penalty is its convex relaxation. Given the loadings, the factors are the cross-sectional regression coefficients $\bm g_{t+1}=(\bm\Gamma_\beta^\top\bm G_t\bm\Gamma_\beta)^{-1}\bm\Gamma_\beta^\top\bm\phi^{\sy\top}_t\bm x_{t+1}$. For $K=m^\sy$, a nonsingular $\bm\Gamma_\beta$ and months of full rank, $\bm\Gamma_\beta\bm g_{t+1}=\bm f_{t+1}$, and the implied moments are $\muf=\bar{\bm f}$ and $\bm V=\frac1T\sum_t\bm f_{t+1}\bm f_{t+1}^\top$. This is the CMK estimator with $\varrho=0$, without the noise correction $\hat{\bm F}$ in $\bar{\bm A}$, without the floor and without the penalty. Under the model it estimates $\bm V_0+u^\id_0\E[\bm G_t^{-1}]$ instead of $\bm V_0$. Equivalently, it adds $\hat{\bm F}$ in every direction to the noise-corrected second moment $\bar{\bm A}$. IPCA regularizes its factor step as well, by a smaller Tikhonov regularization (Online Appendix~\ref{app_emp}). Section~\ref{sec_ipcacomparison} decomposes the difference between IPCA and the CMK estimator along these steps.
\end{remark}

\section{Statistical Guarantees}\label{sec_TP}

The CMK program \eqref{krr1convex} minimizes the sample average $\Rcal_T$ of the penalized loss. By \eqref{eqlossdecomp}, the true moment kernel minimizes the expected loss. For fixed features and hyper-parameters, Section~\ref{sec_population} characterizes the target of the CMK program, the minimizer of the population loss. Without penalty and floor, this target reproduces the true conditional moments if the model is well specified and is the best approximation within the model otherwise. Section~\ref{sec_population} also bounds the bias of the penalty and the floor. Section~\ref{sec_guarantees} shows that the solution of the CMK program converges to the target, at rate $T^{-1/2}$ under independence. Remark~\ref{remthmasy} extends this to the one-step estimator of Section~\ref{sec_estimatorNEW} with the floor held fixed. Section~\ref{sec_limitcmk} derives the limit of the CMK estimator, which the empirical analysis uses.

\subsection{Identification}\label{sec_population}

Characteristics are persistent, so the data of different months are dependent. Assumption~\ref{assdata} allows for this dependence.

\begin{assumption}\label{assdata}
The data points $\{\xi_t\}_{t\ge0}$ are strictly stationary and ergodic, with the common distribution $\Pa$ of a generic data point $\xi=(\bm x,\bm z)$, and
\begin{equation}\label{cond_moments}
    \E\big[ \|\Om(\xi)\|_2^2\big(1+\|\bm\phi^\sy(\bm z)\|_F^4+\|\bm\phi^\id(\bm z)\|_F^4+\|\bm x\|_2^4\big)\big]<\infty ,
\end{equation}
where $\Om(\xi)$ is the weight, a function of $\bm z$.
\end{assumption}

Under Assumption~\ref{assdata} the coefficients of \eqref{eqRquadratic} are integrable (Lemma~\ref{lempoploss} in Appendix~\ref{app_proofs}), and the population loss
\begin{equation}\label{eqpoploss}
    \Ecal(\bm U) \coloneqq \E[\Rcal(\bm U,\xi)] = \tfrac{1}{2}\langle\bm U,\Acal\bm U\rangle + \langle\bm B,\bm U\rangle +  \E[\Lcal(\bm 0,\xi)],\qquad \Acal\coloneqq\E[\Acal(\xi)],\quad \bm B\coloneqq\E[\bm B(\xi)],
\end{equation}
is a convex quadratic on $\Sm$. Let $\Mcal(\xi)\coloneqq\E[\X\X^\top\mid \bm z]$ be the conditional second-moment matrix of the augmented return vector given the cross section. Under the model \eqref{eqmodelaug} it equals $\E_t[\X_{t+1}\X_{t+1}^\top]$. In general it is the conditional expectation of $\E_t[\X_{t+1}\X_{t+1}^\top]$ given $\bm z_t$, and the CMK program fits $q_{\bm U}$ to $\Mcal(\xi)$. The model is \emph{well specified} if $\Mcal(\xi)=q_{\bm U_0}(\Zt,\Zt^\top)$ $\Pa$-a.s.\ for some $\bm U_0\in\Dcal$, and \emph{misspecified} otherwise.

\begin{proposition}[Identification]\label{propident}
Let Assumption~\ref{assdata} hold.
\begin{enumerate}
\item\label{propident1} For every $\bm U\in\Sm$,
\begin{equation}\label{eqpoplossdecomp}
   \Ecal(\bm U)=\E\big[\|\Om^{1/2}(\X\X^\top-\Mcal(\xi))\Om^{1/2}\|_F^2\big]+\E\big[\|\Om^{1/2}(\Mcal(\xi)-q_{\bm U}(\Zt,\Zt^\top))\Om^{1/2}\|_F^2\big]+\langle\bm\Lambda,\bm U\rangle .
\end{equation}
Without penalty, the minimizers $\bm U$ of $\Ecal$ over $\Dcal$ are those whose kernel matrices are the $\Om$-weighted least squares projections of $\Mcal(\xi)$ on the model, the best approximations in the metric of the weight.
\item\label{propident2} If the model is well specified, then $\Ecal(\bm U)=\Ecal(\bm U_0)+\tfrac12\langle\bm U-\bm U_0,\Acal(\bm U-\bm U_0)\rangle+\langle\bm\Lambda,\bm U-\bm U_0\rangle$. Without penalty, $\bm U_0$ minimizes $\Ecal$ over $\Dcal$, every minimizer reproduces the true conditional moments, and $\bm U_0$ is the unique minimizer if $\Ecal$ is strongly convex.
\item\label{propident3} If in addition $\Ecal$ is $\alpha$-strongly convex for some $\alpha>0$, the minimizer $\bm U^\ast_{\bm\Lambda}$ of $\Ecal$ over $\Dcal$ satisfies $\|\bm U^\ast_{\bm\Lambda}-\bm U_0\|\le 2\|\bm\Lambda\|/\alpha$.
\item\label{propident4} Under the assumptions of \ref{propident3} without penalty, let $\bm F_{\rm fix}\in\Sy^{m^\sy}_+$ be a fixed floor and $\Dcal_{\rm fix}\coloneqq\{\bm U\in\Dcal:\bm C\succeq\bm F_{\rm fix}\}$. The minimizer $\bm U^\ast_{\rm fix}$ of $\Ecal$ over $\Dcal_{\rm fix}$ satisfies $\|\bm U^\ast_{\rm fix}-\bm U_0\|\le(\|\Acal\|/\alpha)^{1/2}\|\bm F_{\rm fix}\|_F$.
\end{enumerate}
\end{proposition}

Part \ref{propident1} is the population version of \eqref{eqlossdecomp}. A minimizer in part \ref{propident1} is a pseudo-true parameter of a misspecified M-estimator in the sense of \citet{white82}. Well-specification is a condition on the conditional moments. It holds for cross sections of any size and composition. Parts \ref{propident3} and \ref{propident4} bound the two biases of the regularizers. For a floor of the form \eqref{eqDgamma}, held fixed, $\|\eta\hat{\bm F}\|_F$ is of order $1/N$ for regular cross sections.

\subsection{Consistency and Finite-Sample Bounds}\label{sec_guarantees}

The sample objective $\Rcal_T$ is a random convex quadratic whose coefficients, the sample averages $\Acal_T$ and $\bm B_T$, converge with probability one to those of $\Ecal$. If $\Ecal$ is strongly convex, the solution of the CMK program therefore converges to the minimizer of $\Ecal$. Theorem~\ref{thmasy} states this and adds a mean squared error bound and a finite-sample bound under independence over time.

\begin{theorem}\label{thmasy}
Let Assumption~\ref{assdata} hold.
\begin{enumerate}
    \item\label{thmasy1} Consistency: Assume that $\Ecal$ is strongly convex. Then $\Ecal$ has a unique minimizer $\bm U^\ast$ over $\Dcal$. With probability one, $\Rcal_T$ has a unique minimizer $\bm U^\ast_T$ over $\Dcal$ for all sufficiently large $T$, and $\bm U^\ast_T\to\bm U^\ast$ as $T\to\infty$.

\item\label{thmasy2} Mean squared error bound: Assume further that the $\xi_t$ are independent and identically distributed (i.i.d.), that $\Rcal(\cdot,\xi)$ is $\alpha$-strongly convex for $\Pa$-a.e.\ $\xi$, for some $\alpha>0$, and that
\begin{equation}\label{assSM}
    \E\big[ \|\bm\zeta(\xi)\|^2\big]\le\sigma^2,\qquad \bm\zeta(\xi)\coloneqq(\Acal(\xi) -\Acal)\bm U^\ast + \bm B(\xi)-\bm B = \nabla\Rcal(\bm U^\ast,\xi)-\nabla\Ecal(\bm U^\ast) .
\end{equation}
Then the minimizers $\bm U^\ast_T$ are unique for every $T\ge1$, and $\E\big[ \|\bm U^\ast_T - \bm U^\ast\|^2\big]\le  \sigma^2/(\alpha^2 T)$.

\item\label{thmasy3} Finite-sample bound: Assume further that
\begin{equation}\label{assEM}
    \E\big[ \exp(\varsigma^{-2}\|\bm\zeta(\xi)\|^2)\big]\le \exp (1)
\end{equation}
for some $\varsigma>0$. Then for all $\varepsilon>0$ and $T\ge1$, $\|\bm U^\ast_T-\bm U^\ast\|\ge\varepsilon$ with probability at most $2\exp(-\varsigma^{-2}T\varepsilon^2\alpha^2/3)$. Equivalently, for any $\delta\in (0,1)$, with probability at least $1-\delta$,
\[ \|\bm U^\ast_T - \bm U^\ast\| \le  \frac{\sqrt{3\log(2/\delta)}\,\varsigma}{ \alpha \sqrt{T}}.\]
\end{enumerate}
\end{theorem}

The proofs rest on an elementary inequality for minimizers of two convex quadratics over a closed convex set (Lemma~\ref{lembasicineq}). The distance of the minimizers is bounded by the norm of the sample average of the gradient noise $\bm\zeta(\xi_t)$, divided by the curvature. Consistency then follows from the ergodic theorem, the mean squared error bound from the variance of an average of i.i.d.\ terms, and the finite-sample bound from the sample average approximation theory of \citet{mil_23}.

\begin{remark}[Scope]\label{remthmasy}
(i) The theorem holds for $\Dcal_\eta$ with a fixed floor and for any loss that is a convex quadratic \eqref{eqRquadratic} with integrable coefficients, in particular for the regularized Gram-whitened loss \eqref{eqlossrho}, whose target differs from $\bm U_0$ for $\varrho>0$. For \eqref{eqlossrho} the loss of month $t$ is $\alpha_t$-strongly convex with $\alpha_t\ge2\min\{\kappa,\mathrm{df}_t/(1+\mathrm{df}_t+\|\bm G^+_{t,\varrho}\|_F^2)\}$, also in months with rank-deficient features. These $\alpha_t$ are bounded below uniformly, as parts \ref{thmasy2} and \ref{thmasy3} require, when $\|\bm G^+_{t,\varrho}\|_F$ is bounded, for instance when $\varrho_t$ is bounded away from zero. The theorem thus covers the one-step estimator \eqref{krr1floor} with a fixed floor. The floor of the empirical analysis is estimated on the window. Its limit, $\eta$ times the limit of $\hat u$ times $\E[\bm G^+_{t,\varrho}]$, exists by the ergodic theorem, and the theorem covers the one-step estimator with this limit floor. The theorem treats the floor as fixed. Proposition~\ref{proptwostage} covers the estimated floor for the CMK estimator. (ii) The theorem is for fixed features and hyper-parameters. The number of features is fixed as $T$ grows, whereas in the empirical analysis it ranges from 20 to 295 with $T=108$, the regime that motivates the floor (Section~\ref{sec_regularization}). In the empirical analysis the features are learned on the training window and the hyper-parameters are selected on validation months. This adaptive estimator is outside the theory and is assessed out of sample. (iii) The bounds transfer to the estimated conditional moments, and they hold for cross sections of bounded but arbitrary and varying size. Their constants depend on the size through the weight and the feature matrices. Appendix~\ref{app_scope} gives the details and the extension of part \ref{thmasy2} to weakly dependent data.
\end{remark}

\subsection{Limit of the CMK Estimator}\label{sec_limitcmk}

The CMK estimator of Section~\ref{sec_twostage} is a continuous function of a few window averages, so its properties follow from the ergodic theorem. Write $\bar e$ and $\bar d$ for the window averages of the numerator and of the degrees of freedom $\mathrm{df}_t$ of $\hat u$ in \eqref{eqXi}.

\begin{proposition}[Limit of the CMK estimator]\label{proptwostage}
Let the data points be strictly stationary and ergodic, let $\mathrm{df}_t>0$ for all $t$, and fix $\eta,\lsy,\varrho\ge0$.
\begin{enumerate}
\item\label{proptwostage1} The CMK estimator is a continuous function of $\bar{\bm f}$, $\frac1T\sum_{t<T}\bm f^\varrho_{t+1}\bm f^{\varrho\top}_{t+1}$, $\frac1T\sum_{t<T}\bm G^+_{t,\varrho}$, $\bar e$ and $\bar d$. If these have finite expectations, it converges with probability one to the same function of the expectations.
\item\label{proptwostage2} Let $\varrho=0$, let the features have full rank in every month, and let the model be well specified with parameter $(\muf_0,\bm V_0,u^\id_0)$. Then the limit in \ref{proptwostage1} is $\muf_0$, $\eta\bm F_\infty+\Pcal_+(\bm C_0-\eta\bm F_\infty-\tfrac{\lsy}{2}\bm I)$ and $u^\id_0$, with $\bm F_\infty\coloneqq u^\id_0\E[\bm G_t^{-1}]$. It equals $(\muf_0,\bm C_0,u^\id_0)$ when $\lsy=0$ and $\bm C_0\succeq\eta\bm F_\infty$.
\item\label{proptwostage3} If the data points are i.i.d.\ and the summands of the averages in \ref{proptwostage1} have finite second moments, the CMK estimator is within $O_P(T^{-1/2})$ of its limit.
\end{enumerate}
\end{proposition}

In a factor model with fewer factors than features and $u^\id_0>0$, $\bm C_0$ is singular and $\bm F_\infty$ is positive definite, so the condition $\bm C_0\succeq\eta\bm F_\infty$ fails for $\eta>0$. For $\lsy=0$ the limit of $\bm C$ in part \ref{proptwostage2} is then $\bm C_0+\Pcal_+(\eta\bm F_\infty-\bm C_0)$. The bias $\Pcal_+(\eta\bm F_\infty-\bm C_0)$ is PSD, its eigenvalues, ordered by size, are at most those of $\eta\bm F_\infty$ by Weyl's inequality, and it is of order $1/N$ for regular cross sections. The moments implied by IPCA with $K=m^\sy$ have the bias $\bm F_\infty$ in every direction (Remark~\ref{remIPCA}).

For $\varrho>0$ the limit differs from the true parameter even if the model is well specified. The regularized coordinates are attenuated toward zero in the directions whose eigenvalue is not large relative to $\varrho_t$, and $\hat u$ is biased upward, since the residual keeps the attenuated part of the systematic component. These biases do not vanish as $T$ grows. Online Appendix~\ref{app_emp} measures them.

\section{Factor Model and UMVE Portfolio}\label{sec_factorSDF}

The moment kernel $q_{\bm U}$ of a parameter $\bm U\in\Dcal$ implies a conditional factor model with traded factor portfolios and a UMVE portfolio. Proposition~\ref{propGLS} states both, and Section~\ref{sec_empiricsNEW} uses it to turn the estimated moments into factor portfolios and the UMVE portfolio. The returns $\bm x_{t+1}$ are excess returns, that is, payoffs with zero price. Fix a cross section, with the moments $\bm\mu_t$, $\bm\Sigma_t$ and $\bm Q_t$ of \eqref{eqmutSigmatest}. In Proposition~\ref{propGLS}, conditional expectations, covariances and Sharpe ratios are those of the model with kernel $q_{\bm U}$, that is, with $\E_t[\X_{t+1}\X_{t+1}^\top]=q_{\bm U}(\Zt_t,\Zt_t^\top)$. For an estimated kernel the statements hold in the model, and they hold in the data if the estimate reproduces the true conditional moments. The statements use pseudoinverses and need no invertibility. Their only standing assumption is
\begin{equation}\label{eqassid}
   \Ima\bm\phi^\sy_t\subseteq\Ima\bm\Sigma^\id_t ,
\end{equation}
that is, an asset with zero idiosyncratic variance has zero loadings. It holds in the isotropic specification with $u^\id>0$.

A \emph{managed portfolio} is a portfolio $\bm w_t^\top\bm x_{t+1}$ with time-$t$ measurable weights. It is \emph{conditionally mean--variance efficient} (CMVE) if it attains the maximal conditional Sharpe ratio among managed portfolios in every state in which this ratio is positive. It is \emph{UMVE} if it attains the maximal unconditional Sharpe ratio among managed portfolios.

\begin{proposition}[Factor portfolios and UMVE portfolio]\label{propGLS}
Let \eqref{eqassid} hold, and define
\begin{equation}\label{GLSfac}
\begin{gathered}
   \bm f_{t+1}\coloneqq \bm W_t^\top\bm x_{t+1},\qquad
   \bm W_t^\top\coloneqq \big(\bm\phi_t^{\sy\top}\bm\Sigma^{\id+}_t\bm\phi^\sy_t\big)^+\bm\phi_t^{\sy\top}\bm\Sigma^{\id+}_t,\qquad
   \bm\epsilon_{t+1}\coloneqq\bm x_{t+1}-\bm\phi^\sy_t\bm f_{t+1},\\
   \theta_t\coloneqq\bm\mu_t^\top\bm\Sigma_t^+\bm\mu_t,\qquad
   x^\ast_{t+1}\coloneqq \bm\mu_t^\top\bm Q_t^{+}\bm x_{t+1}.
\end{gathered}
\end{equation}
\begin{enumerate}
\item\label{propGLS1} The factor portfolios $\bm f_{t+1}$ have the regression loadings $\bm\phi^\sy_t$, that is, $\bm\phi^\sy_t=\E_t[\bm x_{t+1}\bm f_{t+1}^\top]\E_t[\bm f_{t+1}\bm f_{t+1}^\top]^+$. Their residuals satisfy $\E_t[\bm\epsilon_{t+1}]=\bm 0$ and $\cov_t[\bm f_{t+1},\bm\epsilon_{t+1}]=\bm 0$. The maximal conditional squared Sharpe ratio of managed portfolios is $\theta_t$, and the factor portfolios attain it, $\theta_t=\E_t[\bm f_{t+1}]^\top\cov_t[\bm f_{t+1}]^+\E_t[\bm f_{t+1}]$.
\item\label{propGLS3} The portfolio with weights $\bm\Sigma_t^+\bm\mu_t$ is CMVE. The portfolio $x^\ast_{t+1}$ has weights $\bm Q_t^+\bm\mu_t=\bm\Sigma_t^{+}\bm\mu_t/(1+\theta_t)$ and $\E_t[x^\ast_{t+1}]=\E_t[(x^\ast_{t+1})^2]=\theta_t/(1+\theta_t)$. It is CMVE.
\item\label{propSDF3} If the model with kernel $q_{\bm U}$ holds for all $t$, with stationary $\bm z_t$, then every square-integrable managed portfolio $y_{t+1}=\bm w_t^\top\bm x_{t+1}$ satisfies
\begin{equation}\label{eqSRuncond}
   \E[y_{t+1}]^2\le \frac{\bar\theta}{1-\bar\theta}\var[y_{t+1}],\qquad \bar\theta\coloneqq\E\Big[\frac{\theta_t}{1+\theta_t}\Big]\in[0,1),
\end{equation}
with equality for $y=x^\ast$. If $\bar\theta>0$, the portfolio $x^\ast$ is thus UMVE, with maximal unconditional squared Sharpe ratio $\bar\theta/(1-\bar\theta)$.
\end{enumerate}
\end{proposition}

Proposition~\ref{propGLS} specializes results of \citet{filipovicschneider26fund} on linear factor models and of \citet{hansenrichard87} on conditioning information to the CMK. With the linear kernel of Section~\ref{sec_dataprotocol}, the CMK is the setting of \citet{kozaknagel24}, in which expected returns and loadings are linear in the characteristics. The factor portfolios are the generalized least squares coefficients of the returns on the loadings with weights $\bm\Sigma^{\id+}_t$, which \citet{filipovicschneider26fund} call moment-weighted factor portfolios. In the isotropic case they coincide with the least squares coefficients $\bm f_{t+1}=\bm\phi^{\sy+}_t\bm x_{t+1}$ of \eqref{eqUhatgw}. By part \ref{propGLS1}, the loadings are the betas of the factors, residual risk is unpriced, and an investor who trades only the $m^\sy$ factor portfolios attains the same maximal conditional Sharpe ratio as one who trades all $N_t$ stocks. We call $x^\ast_{t+1}$ the UMVE portfolio and the portfolio with weights $\bm\Sigma_t^+\bm\mu_t$ the CMVE portfolio. The UMVE portfolio differs from the CMVE portfolio by the scaling $1/(1+\theta_t)$, which reduces the position in states with a high conditional Sharpe ratio. Part \ref{propSDF3} predicts the unconditional Sharpe ratio of the UMVE portfolio from the estimated conditional moments. Section~\ref{sec_assetpricing_empirical} compares the prediction with the realized Sharpe ratio.

\section{Empirical Analysis}\label{sec_empiricsNEW}

We evaluate the CMK estimator out of sample on US stocks. Section~\ref{sec_dataprotocol} describes the data and the implementation. Sections~\ref{sec_predictiveNEW} to~\ref{sec_robustness} report the fit, the comparison with IPCA, the portfolios and the robustness checks.

\subsection{Data and Implementation}\label{sec_dataprotocol}

We use the monthly panel of \citet{gukellyxiu20b} from January 1962 to December 2021, 720 months. The cross section of month $t$ consists of the US common stocks in the Center for Research in Security Prices (CRSP) at the end of month $t$ with at least 30\% of the characteristics observed. The covariate vector $z_{t,i}\in\R^{94}$ holds the 94 firm characteristics of \citet{gukellyxiu20b} at the end of month $t$, rank-transformed to $(-1,1)$ within the cross section, with missing values set to zero. This is the instrument set of IPCA. The return of month $t+1$ is the CRSP return, compounded with the delisting return when the stock delists. Stocks without a return are dropped.\footnote{The characteristics are ranked before these stocks are dropped. Given the covariates, the moments that the model implies for the stocks with a return are a block of its moments for the whole cross section, so the forecasts in the measures of Section~\ref{sec_predictiveNEW} involve no look-ahead from the drop. The portfolio weights are computed on the stocks with a return and therefore contain this look-ahead. The drop removes 13 of 2.7 million stock-months.} Excess returns are returns minus the one-month Treasury bill rate. The cross section has fewer than 100 stocks in the first months of 1962, about 3,700 on average and about 4,300 in the test months below, and at most 6,848 in 1997.

The CMK uses two kernels, the linear and the additive kernel. We write $m$ for $m^\sy$ from here on. The linear kernel $k^{\rm lin}(z,z')=1+\langle z,z'\rangle_2$ has the affine functions of the characteristics as RKHS and the feature matrix $[\bm 1,\bm z_t]$, the instrument matrix of IPCA. Its features are the projections of $[1,z]$ on the $m\in\{20,40,95\}$ dominant eigenvectors of the pooled Gram matrix of $[\bm 1,\bm z_t]$ over the training window, and $m=95$ is the full feature space. The additive kernel $k^{\rm lin}+k^{\rm G}$ adds the Gaussian kernel $k^{\rm G}(z,z')=\exp(-\|z-z'\|_2^2/(2\ell^2))$, which allows nonlinear loadings. We select 200 landmarks by pivoted Cholesky on 20,000 stock-months from the first 96 months of the panel and append their Nystr\"om features \eqref{eqnystrom} to the 95 linear ones. Together they form an orthonormal system for the additive kernel, since the RKHS of the Gaussian kernel contains no nonzero affine function \citep{minh10}. We use the $m\in\{152,202,295\}$ dominant directions of the training window, and $m=295$ is the full feature space. The bandwidth is $\ell^2\in\{\ell_0^2,\ell_0^2/4\}$, with $\ell_0^2=36.4$ the median squared distance between covariate vectors. The feature matrix $[\bm 1,\bm z_t]$ has rank 95 in every month from 1985 on and rank 71 to 94 before, when some characteristics are not yet available.\footnote{In the windows of the test months up to 1985 some characteristics are missing in every month of the window, so the full linear feature space contains directions $\bm v$ with $\bm\phi^\sy_t\bm v=\bm 0$ throughout. In these directions the coordinates and the estimate are zero.}

Each test month $t$ has a training window of the 96 months $t-108,\dots,t-13$ and a validation window of the 12 months $t-12,\dots,t-1$. We compute the CMK estimator of Section~\ref{sec_twostage} with the Tikhonov parameter $\varrho=10^{-3}$ (Online Appendix~\ref{app_emp} gives the eigenvalues it affects). The Tikhonov regularization enters only the estimation. The measures and the portfolios use the exact factor portfolios \eqref{GLSfac} and the model moments of the test month. The model is estimated on the training window for every combination of $m$, $\eta\in\{0,\tfrac12,1,2\}$ and $\lsy\in\{0,2\cdot10^{-4},6\cdot10^{-4},2\cdot10^{-3},6\cdot10^{-3}\}$. For the hyper-parameters with the best average validation score, the model is then re-estimated, features included, on the 108 training and validation months and evaluated on the returns of the test month. The score is the proper scoring rule of \citet{dawidsebastiani99} for the pair of conditional moments,
\begin{equation}\label{eq_scoringloss}
\Scal(\bm{x}, \bm{\mu}, \bm{\Sigma}) \coloneqq \log \det \bm{\Sigma} + (\bm{x} - \bm{\mu})^{\top} \bm{\Sigma}^{-1} (\bm{x} - \bm{\mu}),
\end{equation}
which is minus twice the Gaussian log-likelihood up to a constant (Online Appendix~\ref{app_emp} gives its computation). The rows of the tables labeled by $m$ select $(\eta,\lsy)$ in every window. The ``selected'' rows also select $m$ for the linear kernel, and the pair $(\ell,m)$ for the additive kernel, with the linear kernel with $m=95$ as a further candidate. The windows roll forward monthly. We report the 594 test months from July 1972 to December 2021, whose training windows consist of months with more stocks than the largest specification has features.\footnote{The first 18 months, January 1962 to June 1963, have fewer stocks than features. Keeping them out of the training windows leaves the 594 test months, the same for every specification.}

We compare the CMK estimator with two benchmarks. The first is the purely idiosyncratic model $\bm\mu^{\rm bm}_t=\bm 0$, $\bm\Sigma^{\rm bm}_t=u^{\rm bm}\bm I_{N_t}$, with $u^{\rm bm}$ the average of $\|\bm x_{t+1}\|_2^2/N_t$ over the 108 months. The second is IPCA \citep{kellypruittsu19} with $K\in\{20,40,60,80,95\}$ factors and the instruments $[\bm 1,\bm z_t]$. Its moments are those of Remark~\ref{remIPCA}, $\bm\mu_t=[\bm 1,\bm z_t]\bm\Gamma_\beta\bar{\bm g}$ and $\bm\Sigma_t=[\bm 1,\bm z_t]\bm\Gamma_\beta\widehat{\cov}(\bm g)\bm\Gamma_\beta^\top[\bm 1,\bm z_t]^\top+u^{\rm IPCA}\bm I_{N_t}$ with the residual variance $u^{\rm IPCA}$. We estimate IPCA on the same windows and evaluate it with the same code (Online Appendix~\ref{app_emp}). We use IPCA because it is the closest model that gives both moments for changing cross sections. Its covariance matrices are a by-product of its factor model, which \citet{kellypruittsu19} evaluate by the fit of returns.

The specification was developed on this panel. The weight, the floor and its shape, the Tikhonov parameter, the grids, the split of the window and the construction of the features were chosen by comparing out-of-sample results that include the test months reported here. Online Appendix~\ref{app_jkp} repeats the study on a second panel with the specification unchanged.

\subsection{Statistical Fit}\label{sec_predictiveNEW}

We evaluate the moments of a test month against the realized returns. The measures depend only on the forecast moments and the realized returns, so they treat estimators with different losses alike. The predictive $R^2$ of the conditional mean and the predictive $R^{2,2}$ of the conditional second moment,
\begin{equation}\label{eq_kellyoos}
 R^{2}\coloneqq 1 - \frac{\sum_{t} N_t^{-1} \| \bm x_{t+1} - \bm\mu_t \|_2^2}{\sum_{t} N_t^{-1} \| \bm x_{t+1} \|_2^2},\qquad
 R^{2,2}\coloneqq 1 - \frac{\sum_{t} N_t^{-1} \| \bm x_{t+1} \bm x_{t+1}^\top - \bm Q_t \|_F^2}{\sum_{t} N_t^{-1} \| \bm x_{t+1} \bm x_{t+1}^\top - u^{\rm bm} \bm I_{N_t} \|_F^2},
\end{equation}
with sums over the test months, compare each moment with the benchmark, as in \citet{kellypruittsu19}. With the weight $N_t^{-2}$ instead of $N_t^{-1}$ in $R^{2,2}$, under which every month counts equally, $R^{2,2}$ changes by at most 0.03 percentage points (Online Appendix~\ref{app_emp}). The \emph{score differential} $\Scal^{\rm diff}$ is the average over the test months of $\Scal(\bm x_{t+1},\bm 0,u^{\rm bm}\bm I_{N_t})-\Scal(\bm x_{t+1},\bm\mu_t,\bm\Sigma_t)$. It evaluates the two moments jointly, and higher is better. A score differential of 817 means that the moments raise the Gaussian log-likelihood of the test month over the benchmark by about 409, or about $0.1$ per stock. The total $R^2$ of the factor portfolios \eqref{GLSfac},
\begin{equation}\label{eq_totalrsqr}
 R^{2,\bm f}\coloneqq 1 - \text{average over the test months of } \| \bm x_{t+1} - \bm\phi^\sy_t \bm f_{t+1}\|^2_2/\| \bm x_{t+1}\|_2^2,
\end{equation}
measures the variation of the test month's returns that the loadings explain, like the total $R^2$ of \citet{kellypruittsu19} but averaged over the months. The projection on the feature span also picks up idiosyncratic variance, of expected size $u^\id\rank\bm\phi^\sy_t$, so $R^{2,\bm f}$ rises mechanically with the dimension of the feature space. The adjusted total $R^2$ subtracts $\hat u\rank\bm\phi^\sy_t/\|\bm x_{t+1}\|_2^2$ from the total $R^2$ of every month. The projected variance share $\rho^{\bm f}_t\coloneqq\tr(\bm\phi^\sy_t\cov_t[\bm f_{t+1}]\bm\phi^{\sy\top}_t)/\tr\bm\Sigma_t$ is the predicted counterpart of the total $R^2$.\footnote{By Lemma~\ref{lemfactordetails} in Appendix~\ref{app_proofs}, $\tr\bm\Sigma^\sy_t/\tr\bm\Sigma_t\le\rho^{\bm f}_t\le\tr\bm\Sigma^\sy_t/\tr\bm\Sigma_t+m/N_t$.} We test paired differences of monthly scores with Newey--West standard errors with 12 lags. This is the test of \citet{dieboldmariano95}, which \citet{giacominiwhite06} justify for comparing forecasting methods, estimation included, with rolling windows of fixed length.

\begin{table}
\begin{center}
\resizebox{\textwidth}{!}{\input{tables/rev11_tab1_stat.tex}}
\end{center}
\caption{\label{tab_stat}Out-of-sample statistical performance on the 594 test months from July 1972 to December 2021. ``\# factors'' is the average rank of the excess over the floor (an upper bound for the additive kernel, Online Appendix~\ref{app_emp}), and for IPCA the average rank of the factor covariance. ``score'' is the score differential, and $\rho^{\bm f}$ the projected variance share. ``diff.\ lin.'' and ``diff.\ Gauss.'' are the average monthly score differences to the linear kernel with $m=95$ and to the selected additive kernel, with Newey--West $t$-statistics. Every CMK row selects $(\eta,\lsy)$ by validation. IPCA has $K$ factors and the same instruments. ``linear + Gaussian'' is the additive kernel, and $\ell=1\times$ and $0.5\times$ stand for $\ell_0$ and $\ell_0/2$.}
\end{table}

Predictability in Table~\ref{tab_stat} is modest for the mean and stronger for the second moment. $R^2$ is between $0.2$ and $0.4$\% and $R^{2,2}$ between $0.8$ and $1.0$\% ($1.03$\% for the selected additive kernel). The score differential is positive for every specification. For the selected additive kernel it is positive in 94\% of the test months and in every calendar year. The loadings explain 14--24\% of the variation of returns in the test month, and 13--17\% after the adjustment. The model attributes 18--25\% of the conditional variance to systematic risk, $\tr\bm\Sigma^\sy_t/\tr\bm\Sigma_t$ on average, a little less than the share $\rho^{\bm f}$ in the table.

For the linear kernel every measure improves with the number of features. The score differential rises from 763 at $m=20$ to 817 at $m=95$, and the validation picks $m=95$ in 98\% of the windows when it selects $m$. At $m=95$ it selects a positive floor in 88\% of the windows and a positive penalty in 44\%, and the estimate has 22 factors, against 45 without the regularizers.

The Gaussian features raise the score differential to 838--872, 21--55 points above the linear kernel, a gain that grows with $m$ and is significant in every specification. With the bandwidth and $m$ selected, the score differential is 878. This is 60 points, or 7\%, above the linear kernel ($t=5.4$). The selected specification has 35 factors, $R^2$ rises from $0.29$ to $0.38$\% and $R^{2,2}$ from $0.87$ to $1.03$\%. The penalty is at the upper end of its grid in most windows, so the counts of the additive kernel are upper bounds (Online Appendix~\ref{app_emp}). In 60-month rolling windows the difference between the additive and the linear kernel is positive throughout (Figure~\ref{fig_rolling}). It is positive in each sub-period, 1972--1989, 1990--2009 and 2010--2021, with $t$-statistics of $3.2$, $4.1$ and $1.9$, and it is 94 points in the 57 crash months shaded in Figure~\ref{fig_rolling} ($t=3.2$, Table~\ref{tab_periods} in Online Appendix~\ref{app_emp}).

\begin{figure}
\centering
\includegraphics[width=0.62\textwidth]{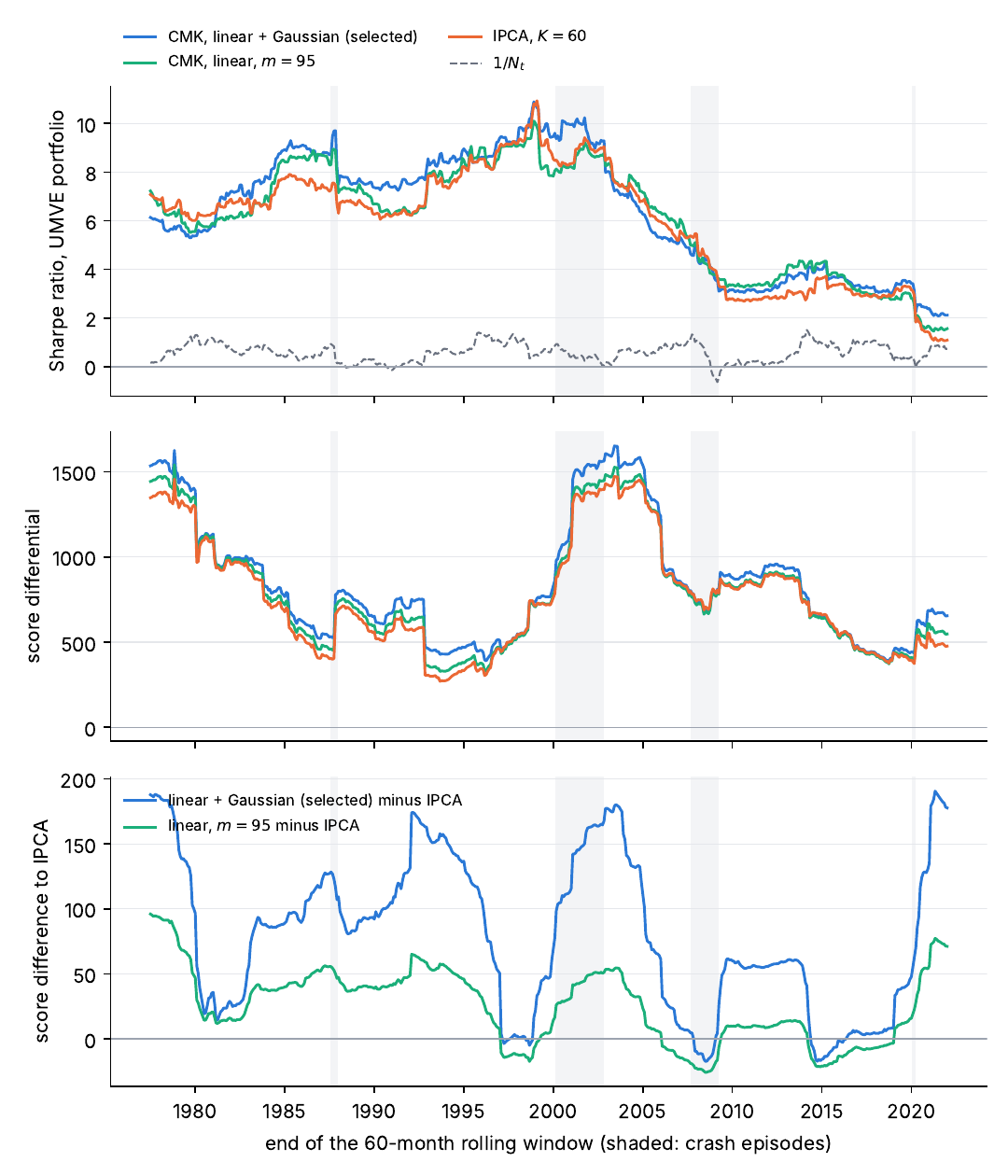}
\caption{\label{fig_rolling}Rolling out-of-sample performance over 60-month windows. Top: annualized Sharpe ratio of the UMVE portfolio for the additive kernel (bandwidth and $m$ selected), the linear kernel ($m=95$), IPCA ($K=60$) and the equally weighted portfolio. Middle: score differential. Bottom: score difference between the CMK specifications and IPCA. Shaded: the 1987 crash, the dot-com bubble, the global financial crisis and the COVID-19 crash, dated by \citet{adr_nat_qur_23}.}
\end{figure}

\subsection{Comparison with IPCA}\label{sec_ipcacomparison}

IPCA and the linear kernel share the feature space. With $K=95$, the full instrument space, IPCA scores 793, and with $K=20$ it scores 774. The linear kernel with $m=95$ is 25 to 43 points ahead of IPCA for every $K$, with $t$-statistics between $3.2$ and $4.8$. The fit of the mean is about the same. The UMVE Sharpe ratio of the linear kernel, $4.9$, exceeds IPCA's, $4.5$ to $4.7$. The $p$-values of the \citet{ledoitwolf08} test, not tabulated for $K\ne60$, are $0.10$ and $0.15$ against $K=20$ and $40$, $0.02$ against $K=60$, and below $0.01$ against $K=80$ and $95$.

Setting the means to zero in the score separates a covariance part and a mean part of the score differential (Online Appendix~\ref{app_emp}). Against IPCA with $K=60$, the linear kernel with $m=95$ gains 32 points in the covariance part ($t=3.7$) and 0.3 points in the mean part ($t=0.7$). The GMV portfolio tests the covariance forecast without the mean, since its weights $\bm\Sigma_t^{-1}\bm 1/(\bm 1^\top\bm\Sigma_t^{-1}\bm 1)$ do not use it. It tests the covariance matrix only in the direction $\bm\Sigma_t^{-1}\bm 1$, whereas the score tests it in every direction. Its out-of-sample volatility is 8.7\% a year for the linear kernel with $m=95$, against 9.0\% to 10.8\% for IPCA and 20.8\% for the equally weighted portfolio. Over the full sample the difference to IPCA is significant for $K=40$, $60$ and $95$, with $t$-statistics between $2.3$ and $3.2$, but not for $K=20$ and $80$. From 1985 on it is significant only for $K=20$ and $40$. The GMV return divided by its predicted standard deviation has a variance of 1.9 for the linear kernel and 2.3 to 4.3 for IPCA, against one for a calibrated forecast (Table~\ref{tab_gmv} in Online Appendix~\ref{app_emp}).

\begin{table}
\begin{center}
\input{tables/rev19_tab3_decomp.tex}
\end{center}
\caption{\label{tab_decomp}Decomposition of the gain over IPCA. Each row changes one step, from IPCA with $K=60$ to the CMK estimator in the full linear feature space, $m=95$. Score differential, its change from the row above with the paired Newey--West $t$-statistic, and the annualized Sharpe ratio (SR) of the UMVE portfolio. Rows 3 to 5 have $\lsy=0$, rows 3 and 4 have $\eta=0$.}
\end{table}

Table~\ref{tab_decomp} decomposes the gain over IPCA with $K=60$ along the steps of Remark~\ref{remIPCA}. The 35 additional factors of $K=95$ add 8 points. IPCA's moment step in the full space, with the Tikhonov regularization of Section~\ref{sec_tikhonov} in place of IPCA's weaker one, adds 15 points ($t=2.4$). The gain is 11 points from 1985 on ($t=1.5$) and 29.5 points before ($t=2.1$), when some months are rank deficient. Subtracting the idiosyncratic noise from the second moment costs 10 points. The floor at $\eta=1$ adds 16 points, and the selection of the floor and the penalty another 4. The noise correction helps only together with the floor. The moment step adds $\hat{\bm F}$ to the noise-corrected estimate in every direction, whereas the floor raises the noise-corrected estimate to $\hat{\bm F}$ only where it falls below. The Sharpe ratio rises with the floor.

The selected additive kernel exceeds IPCA by 85 to 104 points ($t$ between $4.5$ and $5.2$). Against IPCA with $K=20$ and $60$ the difference is positive in each sub-period and significant except in 2010--2021. Its Sharpe ratio, $5.1$, is above IPCA's ($p=0.01$ against $K=60$). This comparison mixes the effect of the estimator with that of the features. A comparison on the same features would require the nonconvex IPCA fit in 295 dimensions. Holding the estimator fixed separates the two effects. The Gaussian features add 60 points. They raise the covariance part of the score by 58 points, but they do not lower the volatility of the GMV portfolio, which is 9.0\% a year.

\subsection{Portfolio Performance}\label{sec_assetpricing_empirical}

Proposition~\ref{propGLS} turns the moments into two portfolios. The UMVE portfolio has weights $\bm Q_t^+\bm\mu_t$, and \eqref{eqSRuncond} predicts its unconditional Sharpe ratio. The CMVE portfolio has weights $\bm\Sigma_t^+\bm\mu_t$. Both are long--short portfolios in all stocks, with about 2,300 long and 2,000 short positions in the average test month and a largest weight of about a tenth of the gross exposure.

\begin{table}
\begin{center}
\resizebox{\textwidth}{!}{\input{tables/rev11_tab2_ap.tex}}
\end{center}
\caption{\label{tab_ap}Portfolios on the 594 test months. Predicted unconditional Sharpe ratio $\sqrt{12\,\bar\theta/(1-\bar\theta)}$ of \eqref{eqSRuncond}, with the expectation in $\bar\theta$ replaced by the average over the test months. Realized annualized Sharpe ratios of the UMVE portfolio (SR UMVE) and of the CMVE portfolio (SR CMVE). Newey--West $t$-statistic of the mean UMVE return and $p$-values of the \citet{ledoitwolf08} test of its Sharpe ratio against $1/N_t$ and against IPCA with $K=60$. Intercept (annualized, in percent), its $t$-statistic and adjusted $R^2$ of the regression of the UMVE return, scaled to 10\% annual volatility, on the five factors of \citet{famafrench15}. Turnover and the net Sharpe ratios refer to the UMVE portfolio at unit gross exposure, $\sum_i|w_{t,i}|=1$: average monthly turnover $\sum_i|w_{t,i}-w_{t-1,i}|$, and Sharpe ratios before costs (``SR net 0'') and net of proportional costs of 10 to 100 basis points of turnover.}
\end{table}

In Table~\ref{tab_ap}, the realized Sharpe ratios of the UMVE portfolio are $4.0$--$4.9$ for the linear kernel and $5.0$--$5.1$ for the additive kernel, against $0.5$ for the equally weighted portfolio. The $t$-statistics of the mean returns are between 14 and 19. The CMVE portfolio has a lower Sharpe ratio than the UMVE portfolio in every specification, consistent with part \ref{propSDF3} of Proposition~\ref{propGLS}. The predicted Sharpe ratios exceed the realized ones, by a factor of $1.4$ at $m=20$, $1.5$ at $m=40$ and $1.7$ to $1.8$ for $m=95$ and the additive kernel. Part of the gap is the upward bias of a maximal squared Sharpe ratio computed from estimated moments \citep{jobsonkorkie80}. The five factors of \citet{famafrench15} explain 1--4\% of the variation of the UMVE returns and leave annualized intercepts of 38--50\% at 10\% volatility, with $t$-statistics above 13.

The weights $\bm Q_t^+\bm\mu_t$ have a gross exposure that varies over time. At unit gross exposure the Sharpe ratios before costs are $3.9$--$4.9$ for the linear kernel, $4.9$--$5.0$ for the additive kernel and $4.6$--$4.8$ for IPCA. The portfolios turn over 74\% of their book per month at $m=20$, 93\% at $m=95$, 85--89\% for the additive kernel and 84--93\% for IPCA. Net of 25 basis points the Sharpe ratios are $3.2$--$3.9$ for the linear kernel, $4.0$--$4.1$ for the additive kernel and $3.6$--$3.8$ for IPCA. From 25 basis points on, every additive specification has a higher Sharpe ratio than every linear and IPCA specification, and at 100 basis points only the additive specifications have a Sharpe ratio above one.\footnote{For the selected specification the turnover chains the weights across months, so the trades induced by the 65 switches between models in the 594 test months are counted. They add one percentage point of monthly turnover.} The costs do not include the borrowing costs of the short positions or the price impact in the smallest stocks. Without the smallest 20\% of the stocks the Sharpe ratio of the linear kernel with $m=95$ falls from $4.9$ to $4.4$. The Sharpe ratios decline over the sample, from $6.5$--$6.9$ in 1972--1989 to $2.1$--$2.9$ in 2010--2021, in parallel for all specifications including IPCA (Figure~\ref{fig_rolling} and Table~\ref{tab_periods}). The decline is in line with the decay of the returns to published characteristics \citep{mcleanpontiff16}. The Sharpe ratios stay positive in the crash months, in which the equally weighted portfolio loses.

\subsection{Robustness}\label{sec_robustness}

Table~\ref{tab_robust} in Online Appendix~\ref{app_emp} varies one ingredient at a time. In the linear feature space the floor carries the gain of the regularizers (Table~\ref{tab_decomp}), and the penalty alone adds nothing. In the additive feature space the penalty also helps. On top of the floor it removes a quarter of the factors and adds 18 points.

The identity-shaped floor scores 9 points below the floor $\eta\hat{\bm F}$ ($t=8.4$), with a Sharpe ratio of $4.6$ against $4.9$. In the additive feature space it is selected in 3 of the 594 windows and scores 22 points below ($t=4.4$), for the reason given in Section~\ref{sec_regularization}.

The weight matters mainly for the portfolios. The scalar weight $N_t^{-1/2}\bm I_{N_t+1}$, which lets the weight of a month grow with $N_t$ in the second-moment term, lowers the Sharpe ratio from $4.9$ to $3.4$ ($p<0.01$), at a score that is not significantly lower. Without regularizers and at $m=40$, the Frobenius weight lowers the Sharpe ratio from $3.9$ to $3.4$ and scores 34 points lower than the Gram-whitened weight (814 against 848, on every second window). At $m=95$ without regularizers, the Browne weight has the Sharpe ratio of the Gram-whitened weight and a score that is 10 points lower.

In the main specification the idiosyncratic variance is the same for every stock. Online Appendix~\ref{app_emp} lets it depend on size and past idiosyncratic volatility, in the linear feature space with $m=20$, the scalar weight and no floor. The CMK program estimates this specification, and it raises $R^{2,2}$ slightly. We have not combined it with the additive kernel, so part of the gain of the Gaussian features may come from the cross section of idiosyncratic risk.

\section{Conclusion}\label{sec_conclusion}

The CMK specifies the conditional mean and covariance matrix of returns as one object, with PSD covariance matrices for cross sections of any size and composition. Its estimation is a convex semidefinite program, the CMK program, and a two-stage closed form, the CMK estimator, approximates it. On US stocks the CMK estimator fits the two moments better than IPCA with the same characteristics. The gain is in the covariance and comes from dropping the rank restriction, from the regularization of the monthly regressions and from the floor, and nonlinear loadings add to it. The implied UMVE portfolios earn Sharpe ratios close to those of IPCA. A covariance estimate benefits from a floor at the scale of the noise, shaped like the noise covariance, in the directions where the systematic variance is below the noise level. The estimation error of a month is dominated by its systematic variation, so all months are about equally informative. In our data equal weights for the months mainly raise the Sharpe ratio of the UMVE portfolio. Above the floor, a convex penalty selects the number of factors.

We see five directions for further research. First, a theory in which the number of features grows with the sample would let the CMK estimator approach the full RKHS and would cover the empirical regime, in which the number of features is comparable to the number of months or larger. As for kernel ridge regression with Nystr\"om features, rates then require source and capacity conditions \citep{caponnettodevito07,rudicamorianorosasco15}. Second, the CMK program, which is needed when no closed form exists, has $O(m^2)$ variables. Without the floor, a diagonal $\bm V$ in the feature basis reduces its semidefinite constraint to $\sum_k b_k^2/V_{kk}\le1$, with $b_k=0$ where $V_{kk}=0$, and the trace penalty to a lasso on the features. This would scale it to thousands of features, at the cost of the basis invariance of Section~\ref{sec_regularization}. Third, idiosyncratic variances that depend on the characteristics, which the CMK program can estimate (Online Appendix~\ref{app_emp}), change the factor portfolios, and their implications for pricing are open. Fourth, the moments change with the characteristics and through the rolling windows, but not with the level of market volatility. A kernel that also depends on a volatility state would capture this. Fifth, the construction applies beyond returns, to any panel of outcomes with covariates in which conditional means and covariances are needed jointly.
\bibliographystyle{plainnat}
\bibliography{master}

\clearpage
\begin{appendix}

\section{Proofs}\label{app_proofs}

We prove the results of the main text and auxiliary results that it uses, among them the representer theorem, the quadratic structure of the loss and the conditions for injectivity. Theorem~\ref{thmUrepr} is proved before Theorem~\ref{thmreprNEW}, whose existence part uses it.

\subsection{Proof of Lemma~\ref{lemfeaturemap}}

\ref{lemfeaturemap1}: Let $\Ccal'$ be the RKHS of the kernel $q^\sy$ on $\Zcal_\Delta$, which exists by the Moore--Aronszajn theorem \citep[Chapter 2]{pau_rag_16}. Being separable by assumption, it embeds isometrically into $\Ccal$, which is infinite-dimensional and separable. Composing with a unitary of $\Ccal$, we choose the embedding to map the unit vector $q^\sy(\cdot,\Delta)$, with $\|q^\sy(\cdot,\Delta)\|_{\Ccal'}^2=q^\sy(\Delta,\Delta)=1$, to $p$. We identify $\Ccal'$ with its image, so that $p=q^\sy(\cdot,\Delta)$. Define $\hsy(z)\coloneqq q^\sy(\cdot,z)$ for $z\in\Zcal$ and $\hsy(\Delta)\coloneqq 0$. Then $\hsy(z)+p1_{z=\Delta}=q^\sy(\cdot,z)$ for all $z\in\Zcal_\Delta$, and the reproducing property gives $\langle \hsy(z)+p1_{z=\Delta},\hsy(z')+p1_{z'=\Delta}\rangle_\Ccal=q^\sy(z,z')$. For the idiosyncratic component, $q^\id(z,z)\ge0$ because $q^\id$ is a kernel. Pick any unit vector $e\in\Ccal$ and set $\hid(z)\coloneqq q^\id(z,z)^{1/2}e$ for $z\in\Zcal$, $\hid(\Delta)\coloneqq0$. Then $\|\hid(z)\|^2_\Ccal1_{z=z'}=q^\id(z,z')$, and $q=q_h$ follows.\footnote{The RKHS of $q^\sy$ is separable, for example, when $\Zcal$ is a separable metric space and $q^\sy$ is continuous. In the reparametrized problem of Section~\ref{sec_reparam}, only finite-dimensional subspaces of $\Ccal$ matter.}

\ref{lemfeaturemap2}: Since $O^\sy p=p$, we have $O^\sy\hsy(z)+p1_{z=\Delta}=O^\sy(\hsy(z)+p1_{z=\Delta})$, and linear isometries preserve inner products. Similarly $\|O^\id\hid(z)\|_\Ccal=\|\hid(z)\|_\Ccal$. Hence $q_{\tilde h}=q_h$.

\subsection{Proof of Theorem~\ref{thmUrepr}}

We write $\Rcal(h,\xi_t)\coloneqq\Lcal(q_h,\xi_t)+\lsy\|\hsy\|^2_{\Hsy}+\lid\|\hid\|^2_{\Hid}$, $\Rcal_T(h)$ for its sample average, the objective of \eqref{krr1old}, $\Hcal\coloneqq\Hsy\times\Hid$, and $\Hcal_\phi$ for the space of the feature maps $h_\phi(\bm\gamma)$. We first prove Theorem~\ref{thmUrepr}, which does not depend on Theorem~\ref{thmreprNEW}, and then use it, together with Proposition~\ref{coruniqueness}, to establish the existence part of the representer theorem. The following lemma is the basis for the surjectivity claim in Theorem~\ref{thmUrepr}\ref{thmUrepr2}. It holds for any Hilbert space $\Ccal$, finite-dimensional or not.

\begin{lemma}\label{lembasNEWW}
For the mappings $ {\bm U}^\sy: \Ccal^{m^\sy}\to \Dcal^\sy\coloneqq\{\bm U^\sy\in\Sy^{m^\sy+1}_+:\bm U^\sy_{00}=1\}$ and $ {\bm U}^\id: \Ccal^{m^\id}\to \Sy^{m^\id}_+$ defined in \eqref{eqUdef} the following hold:
    \begin{enumerate}
   \item\label{lembasNEW1} $ {\bm U}^\sy$ is surjective if and only if $\dim\Ccal\ge m^\sy+1$.
    \item\label{lembasNEW2} $ {\bm U}^\id$ is surjective if and only if $\dim\Ccal\ge m^\id$.
  \end{enumerate}
Hence the minimal dimensional requirement on $\Ccal$ for Theorem \ref{thmUrepr} is $\dim\Ccal \ge \max\{m^\sy+1,m^\id\}$, which holds for every infinite-dimensional $\Ccal$.
\end{lemma}

\begin{proof}
\ref{lembasNEW1}: If $\dim\Ccal=\infty$, it suffices to prove surjectivity from $\Ccal'^{m^\sy}$ for a subspace $\Ccal'\ni p$ of dimension $m^\sy+1$, so we may assume $\dim\Ccal<\infty$. Let $d\coloneqq \dim\Ccal - 1$, and let $p,e_1,\dots,e_d$ be an orthonormal basis of $\Ccal$. Every $\bm\gamma^\sy\in\Ccal^{m^\sy}$ can be expressed in unique coordinates as $\gamma^\sy_i = b_i p + \sum_{j=1}^d a_{ij} e_j $ for some vector $\muf=[b_i:1\le i\le m^\sy]\in\R^{m^\sy}$ and matrix $\bm A=[a_{ij}]\in\R^{m^\sy\times d}$, and vice versa. In these coordinates, ${\bm U}^\sy(\bm\gamma^\sy)= \begin{bmatrix}
        1 & \muf^\top \\ \muf & \muf\muf^\top + \bm A\bm A^\top
    \end{bmatrix}$, and $\bm A\bm A^\top $ is the Schur complement of the $(0,0)$ entry. Hence ${\bm U}^\sy$ is surjective onto $\Dcal^\sy$ if and only if every matrix $\bm C\in\Sy^{m^\sy}_+$ can be expressed as $\bm C=\bm A\bm A^\top $ for some $\bm A\in\R^{m^\sy\times d}$, which holds if and only if $d \ge m^\sy$.
\ref{lembasNEW2}: This follows in the same way, without the constraint $\bm U^\sy_{00}=1$.
\end{proof}

\ref{thmUrepr1}: For $h_\phi^\sy=\bm\phi^\sy\bm\gamma^\sy$ and $z\in\Zcal_\Delta$ we have, with $\bm\phi^\sy(\Delta)=\bm0$, $h_\phi^\sy(z)+p1_{z=\Delta}=\sum_{i=0}^{m^\sy}\aphi^\sy_i(z)\gamma_i$ where $\gamma_0\coloneqq p$ and $\aphi^\sy=[\aphi^\sy_0,\dots,\aphi^\sy_{m^\sy}]$. Hence
\[ \langle h_\phi^\sy(z)+p1_{z=\Delta},h_\phi^\sy(z')+p1_{z'=\Delta}\rangle_\Ccal=\sum_{i,j=0}^{m^\sy}\aphi^\sy_i(z)\langle\gamma_i,\gamma_j\rangle_\Ccal\aphi^\sy_j(z')=\aphi^\sy(z)\bm U^\sy(\bm\gamma^\sy)\aphi^\sy(z')^\top ,\]
because the Gram matrix of $(p,\gamma_1,\dots,\gamma_{m^\sy})$ is $\bm U^\sy(\bm\gamma^\sy)$, using $\|p\|_\Ccal=1$. Likewise $\|h_\phi^\id(z)\|^2_\Ccal1_{z=z'}=\bm\phi^\id(z)\bm U^\id(\bm\gamma^\id)\bm\phi^\id(z')^\top1_{z=z'}$ for all $z,z'\in\Zcal_\Delta$. This proves $q_{h_\phi}=q_{\bm U(\bm\gamma)}$. Evaluating on the augmented cross section $\Zt_t$ gives $q_{h_\phi}(\Zt_t,\Zt_t^\top)=q_{\bm U(\bm\gamma)}(\Zt_t,\Zt_t^\top)$ and thus $\Lcal(h_\phi,\xi_t)=\Lcal(\bm U(\bm\gamma),\xi_t)$. For the norms, the tensor product structure $\Hcal^\tau\cong\Gcal^\tau\otimes\Ccal$ and orthonormality \eqref{eqorthonormal} give
\[ \|h_\phi^\tau\|^2_{\Hcal^\tau}=\Big\|\sum_{i=1}^{m^\tau}\phi^\tau_i\otimes\gamma^\tau_i\Big\|^2=\sum_{i,j=1}^{m^\tau}\langle\phi^\tau_i,\phi^\tau_j\rangle_{\Gcal^\tau}\langle\gamma^\tau_i,\gamma^\tau_j\rangle_\Ccal=\sum_{i=1}^{m^\tau}\|\gamma^\tau_i\|_\Ccal^2 ,\]
which equals $\tr\bm U^\id(\bm\gamma^\id)$ for $\tau=\id$ and $\tr\bm U^\sy(\bm\gamma^\sy)-1$ for $\tau=\sy$. Hence $\Rcal(h_\phi,\xi_t)=\Rcal(\bm U(\bm\gamma),\xi_t)-\lsy$, and averaging gives the claim.

\ref{thmUrepr2}: This is Lemma~\ref{lembasNEWW}.

\ref{thmUrepr3}: By \ref{thmUrepr1}, $\Rcal_T(h_\phi)=\Rcal_T(\bm U(\bm\gamma))-\lsy$ for every $h_\phi=\bm\phi\bm\gamma\in\Hcal_\phi$, and by \ref{thmUrepr2} the coefficient map $\bm\gamma\mapsto\bm U(\bm\gamma)$ is onto $\Dcal$. Hence $\inf_{\Hcal_\phi}\Rcal_T=\inf_\Dcal\Rcal_T-\lsy$, and $h_\phi$ attains the left infimum if and only if $\bm U(\bm\gamma)$ attains the right one. The induced moment kernels coincide by $q_{h_\phi}=q_{\bm U(\bm\gamma)}$.

\subsection{Representer Theorem}

Let $N_{\rm tot}\coloneqq \sum _{t<T}N_t$ be the total sample size.

\begin{theorem}[Representer theorem]\label{thmreprNEW}
Let $\lsy,\lid>0$. Problem \eqref{krr1old} has a minimizer, and every minimizer $h=(h^\sy,h^\id)$ is of the form
\begin{equation}\label{eqn_representer}
   h^{\tau}(\cdot) = \sum_{t=0}^{T-1}\sum_{i=1}^{N_t} k^{\tau}(\cdot,z_{t,i}) \gamma^{\tau}_{t,i},\qquad \gamma^{\tau}_{t,i} \in\Ccal,\quad \tau\in\{\sy,\id\}.
\end{equation}
\end{theorem}

The span property is a special case of the representer theorem of \citet{diwalejones18} for Hilbert-space-valued functions, and it generalizes \citet[Theorem 4.1]{micchellipontil05}. The existence of a minimizer is specific to the moment kernel loss. The positive penalties make the problem coercive, which gives existence, and they confine the minimizers to the span of the kernel sections at the sample covariates. Once the features are fixed, positive penalties are no longer needed (Proposition~\ref{coruniqueness}).

For $\tau\in\{\sy,\id\}$ define the linear sample operator $\Tcal^\tau:\Hcal^\tau\to \Ccal^{N_{\rm tot}}$ by
$\Tcal^\tau h^\tau \coloneqq [h^\tau(z_{t,i}): i=1,\dots,N_t,\, t=0,\dots,T-1]$.
Its adjoint ${\Tcal^\tau}^\ast \bm \gamma^\tau$ is given by the right hand side of \eqref{eqn_representer}, for $\bm\gamma^\tau=[\gamma_{t,i}^\tau:i=1,\dots,N_t,\, t=0,\dots,T-1]$. Indeed, for $\phi\in\Gcal^\tau$ and $v\in\Ccal$,
\begin{align*}
  \langle {\Tcal^\tau}^\ast \bm \gamma^\tau , \phi\otimes v\rangle_{\Hcal^\tau} &= \langle  \bm \gamma^\tau , \Tcal^\tau(\phi\otimes v)\rangle_{\Ccal^{N_{\rm tot}}}  = \sum_{t=0}^{T-1}\sum_{i=1}^{N_t} \langle k^\tau(\cdot,z_{t,i}),\phi\rangle_{\Gcal^\tau} \langle \gamma^\tau_{t,i},v\rangle_\Ccal \\
  &= \Big\langle\sum_{t=0}^{T-1}\sum_{i=1}^{N_t} k^\tau(\cdot,z_{t,i})\otimes\gamma^\tau_{t,i} , \phi\otimes v \Big\rangle_{\Hcal^\tau}.
\end{align*}
Let $\Gcal_1^\tau$ be the subspace of $\Gcal^\tau$ spanned by $\{ k^\tau(\cdot,z_{t,i}): i=1,\dots,N_t,\, t=0,\dots,T-1\}$. It has finite dimension $d^\tau\coloneqq\dim(\Gcal_1^\tau)\le N_{\rm tot}$, and thus is closed in $\Gcal^\tau$, so that $\Ima({\Tcal^\tau}^\ast)=\Gcal_1^\tau\otimes\Ccal$ is a closed subspace of $\Hcal^\tau$, and $\Hcal^\tau = \ker(\Tcal^\tau)\oplus\Ima({\Tcal^\tau}^\ast)$. Decompose any $h\in\Hcal$ as $h^\tau = h^\tau_{\ker}+ h^\tau_1$ with $h^\tau_{\ker}\in\ker(\Tcal^\tau)$ and $h^\tau_1 \in \Ima({\Tcal^\tau}^\ast)$. The loss $\Lcal(q_h,\xi_{t})$ is a function of $(\Tcal^\sy h^\sy,\Tcal^\id h^\id) = (\Tcal^\sy h^\sy_1,\Tcal^\id h^\id_1)$ only, while $\| h^\tau \|_{\Hcal^\tau}^2=\| h^\tau_{\ker} \|_{\Hcal^\tau}^2+\| h^\tau_1 \|_{\Hcal^\tau}^2 \ge  \| h^\tau_1 \|_{\Hcal^\tau}^2$, with equality if and only if $h^\tau_{\ker}=0$. As $\lsy, \lid >  0$, we obtain $\Rcal_T(h_1)\le\Rcal_T(h)$ with equality if and only if $h=h_1$. This shows that every minimizer lies in $\Hcal_1\coloneqq(\Gcal_1^\sy\otimes\Ccal)\times(\Gcal_1^\id\otimes\Ccal)$, that is, is of the form \eqref{eqn_representer}, and that $\inf_\Hcal\Rcal_T=\inf_{\Hcal_1}\Rcal_T$.

For existence, let $\bm\chi^\tau$ be an orthonormal basis of $\Gcal_1^\tau$, so that $\Hcal_1$ is the space $\Hcal_\phi$ for the features $\bm\phi^\tau=\bm\chi^\tau$, $m^\tau=d^\tau$. By Proposition~\ref{coruniqueness}\ref{corun1}, the convex problem \eqref{krr1convex} for this choice has a solution $\bm U^\ast$, and since $\Ccal$ is infinite-dimensional, Theorem~\ref{thmUrepr}\ref{thmUrepr2} provides coefficients $\bm\gamma$ with $\bm U(\bm\gamma)=\bm U^\ast$. By Theorem~\ref{thmUrepr}\ref{thmUrepr3}, $h_1\coloneqq(\bm\chi^\sy\bm\gamma^\sy,\bm\chi^\id\bm\gamma^\id)$ minimizes $\Rcal_T$ over $\Hcal_1$, hence over $\Hcal$.

\subsection{Proof of Lemma~\ref{lemsuffstat}}

We prove the general form of the sufficient statistics, of which Lemma~\ref{lemsuffstat} is the isotropic case. For a matrix $\bm A$, $\bm A\odot\bm A$ is its entrywise square.

\begin{lemma}[Sufficient statistics, general case]\label{lemsuffstatgen}
Let $a_t$, $\bm s_t$ and $\bm M_t$ be as in \eqref{eqsuffstatdef}, and $\bm Y^\id_t\coloneqq \bm\Phi^{\id\top}_t\Diag(\Om_t\X_{t+1}\X_{t+1}^\top\Om_t)\bm\Phi^\id_t\in\Sy^{m^\id}_+$. Then
\begin{multline}\label{eqsuffstat}
   \Lcal(\bm U,\xi_t) = a_t^2 - 2\,\bm s_t^\top\bm U^\sy\bm s_t - 2\tr(\bm U^\id\bm Y^\id_t)
   + \tr(\bm U^\sy\bm M_t\bm U^\sy\bm M_t) \\ + 2\,\bm\varpi^\sy_t(\bm U^\sy)^\top\bm\varpi^\id_t(\bm U^\id) + \bm\varpi^\id_t(\bm U^\id)^\top(\Om_t\odot\Om_t)\bm\varpi^\id_t(\bm U^\id) ,
\end{multline}
where $\bm\varpi^\sy_t(\bm U^\sy)\coloneqq\diag(\Om_t\bm\Phi^\sy_t\bm U^\sy\bm\Phi^{\sy\top}_t\Om_t)=\bm\Theta^\sy_t\vect(\bm U^\sy)$ and $\bm\varpi^\id_t(\bm U^\id)\coloneqq\diag(\bm\Phi^\id_t\bm U^\id\bm\Phi^{\id\top}_t)=\bm\Theta^\id_t\vect(\bm U^\id)$, for the matrices $\bm\Theta^\sy_t$ and $\bm\Theta^\id_t$ whose $i$th rows, $i=0,\dots,N_t$, are the Kronecker products $(\Om_t\bm\Phi^\sy_t)_{i\cdot}\otimes(\Om_t\bm\Phi^\sy_t)_{i\cdot}$ and $\bm\phi^\id(Z_{t,i})\otimes\bm\phi^\id(Z_{t,i})$, with $Z_{t,i}$ the entries of $\Zt_t$, and $\vect$ stacks the columns of a matrix. The last two terms are thus quadratic forms in $(\vect\bm U^\sy,\vect\bm U^\id)$ with coefficient matrices $\bm\Theta^{\sy\top}_t\bm\Theta^\id_t\in\R^{(m^\sy+1)^2\times(m^\id)^2}$ and $\bm\Theta^{\id\top}_t(\Om_t\odot\Om_t)\bm\Theta^\id_t\in\R^{(m^\id)^2\times(m^\id)^2}$.
\end{lemma}

The objective is a quadratic in $O((m^\sy)^2+(m^\id)^2)$ variables. For the Frobenius weight the statistics are accumulated in one pass at a cost of order $N_{\rm tot}\big((m^\sy+1)^2(m^\id)^2+(m^\id)^4\big)$. For a general weight the last coefficient matrix involves $\Om_t\odot\Om_t$, which costs $O(N_t^2(m^\id)^2)$ per month when the weight is dense. In terms of these statistics, the gradient \eqref{eqABc} at $\xi=\xi_t$ has $\bm B(\xi_t)=(-2\bm s_t\bm s_t^\top+\lsy\bm I_{m^\sy+1},\ -2\bm Y^\id_t+\lid\bm I_{m^\id})$, and $\Acal(\xi_t)$ is determined by $\bm M_t$, $\bm\Theta^{\sy\top}_t\bm\Theta^\id_t$ and $\bm\Theta^{\id\top}_t(\Om_t\odot\Om_t)\bm\Theta^\id_t$.

\begin{proof}
Write $\bm Y_{\bm 0}\coloneqq\X_{t+1}\X_{t+1}^\top$, $\bm L\coloneqq\bm\Phi^\sy_t\bm U^\sy\bm\Phi^{\sy\top}_t$ and $\bm H\coloneqq\Diag(\bm\Phi^\id_t\bm U^\id\bm\Phi^{\id\top}_t)=\Diag(\bm\varpi^\id_t(\bm U^\id))$, write $\bm A_{\Om}\coloneqq\Om^{1/2}\bm A\Om^{1/2}$, and drop the index $t$. Expanding $\|(\bm Y_{\bm 0})_{\Om}-\bm L_{\Om}-\bm H_{\Om}\|_F^2$ gives six terms, namely $\|(\bm Y_{\bm 0})_{\Om}\|_F^2=(\X^\top\Om\X)^2=a^2$, $\langle(\bm Y_{\bm 0})_{\Om},\bm L_{\Om}\rangle_F=\X^\top\Om\bm\Phi^\sy\bm U^\sy\bm\Phi^{\sy\top}\Om\X=\bm s^\top\bm U^\sy\bm s$, $\langle(\bm Y_{\bm 0})_{\Om},\bm H_{\Om}\rangle_F=\tr(\Om\bm Y_{\bm 0}\Om\bm H)=\sum_i(\Om\bm Y_{\bm 0}\Om)_{ii}(\bm\Phi^\id\bm U^\id\bm\Phi^{\id\top})_{ii}=\tr(\bm U^\id\bm Y^\id)$, $\|\bm L_{\Om}\|_F^2=\tr(\Om\bm\Phi^\sy\bm U^\sy\bm\Phi^{\sy\top}\Om\bm\Phi^\sy\bm U^\sy\bm\Phi^{\sy\top})=\tr(\bm U^\sy\bm M\bm U^\sy\bm M)$, $\langle\bm L_{\Om},\bm H_{\Om}\rangle_F=\tr(\Om\bm L\Om\bm H)=\diag(\Om\bm L\Om)^\top\diag(\bm H)=\bm\varpi^\sy(\bm U^\sy)^\top\bm\varpi^\id(\bm U^\id)$, and $\|\bm H_{\Om}\|_F^2=\tr(\Om\bm H\Om\bm H)=\sum_{i,j}\Omega_{ij}^2H_{ii}H_{jj}=\bm\varpi^\id(\bm U^\id)^\top(\Om\odot\Om)\bm\varpi^\id(\bm U^\id)$. Collecting terms yields \eqref{eqsuffstat}. The identity $\diag(\bm A\bm U\bm A^\top)_i=\bm A_{i\cdot}\bm U\bm A_{i\cdot}^\top=(\bm A_{i\cdot}\otimes\bm A_{i\cdot})\vect(\bm U)$, applied to $\bm A=\Om\bm\Phi^\sy$ and to $\bm A=\bm\Phi^\id$, gives $\bm\varpi^\sy=\bm\Theta^\sy\vect(\bm U^\sy)$ and $\bm\varpi^\id=\bm\Theta^\id\vect(\bm U^\id)$.
\end{proof}

Lemma~\ref{lemsuffstat} follows: in the isotropic case, $\bm\phi^\id(Z_{t,i})=1_{i\ge1}$, so $\bm\varpi^\id_t(u^\id)=u^\id\diag(\bm J_t)$, $\tr(u^\id\bm Y^\id_t)=u^\id\X_{t+1}^\top\Om_t\bm J_t\Om_t\X_{t+1}=u^\id r_t$, $\bm\varpi^\id_t(u^\id)^\top(\Om_t\odot\Om_t)\bm\varpi^\id_t(u^\id)=(u^\id)^2\tr(\Om_t\bm J_t\Om_t\bm J_t)=n_t(u^\id)^2$ and $\bm\varpi^\sy_t(\bm U^\sy)^\top\bm\varpi^\id_t(u^\id)=u^\id\tr(\Om_t\bm\Phi^\sy_t\bm U^\sy\bm\Phi^{\sy\top}_t\Om_t\bm J_t)=u^\id\tr(\bm U^\sy\bm D_t)$. For the Frobenius weight, $\bm\Phi^{\sy\top}_t\bm J_t\bm\Phi^\sy_t=\bm\Phi^{\sy\top}_t\bm\Phi^\sy_t-\bm E_{00}$ gives $\bm D_t=\bm M_t-\bm E_{00}$, $r_t=\|\bm x_{t+1}\|_2^2$ and $n_t=N_t$.

\subsection{Quadratic Structure of the Loss}

For a data point $\xi=(\bm x,\bm z)$ with weight $\Om$ write $\bm\Phi^\sy$ and $\bm\Phi^\id$ for the feature matrices of its augmented covariate array $\Zt$, and $\bm Y_{\bm U}\coloneqq\X\X^\top-q_{\bm U}(\Zt,\Zt^\top)$ for the residual matrix.

\begin{lemma}[Quadratic structure]\label{lemquadratic}
The penalized loss \eqref{eqRcalU} is the convex quadratic \eqref{eqRquadratic}, where $\Acal(\xi)$ is the self-adjoint PSD operator on $\Sm$ with $\langle\bm U,\Acal(\xi)\bm U'\rangle=2\tr\big(\Om\,q_{\bm U}(\Zt,\Zt^\top)\,\Om\,q_{\bm U'}(\Zt,\Zt^\top)\big)$, $\Lcal(\bm 0,\xi)=(\X^\top\Om\X)^2$, and, with $\bm\Lambda$ as in Section~\ref{sec_uniqueness},
\begin{equation}\label{eqABc}
   \nabla_{\bm U}\Rcal(\bm U,\xi)=\Acal(\xi)\bm U+\bm B(\xi)=-2\big(\bm\Phi^{\sy\top}\Om\bm Y_{\bm U}\Om\bm\Phi^\sy,\ \bm\Phi^{\id\top}\Diag(\Om\bm Y_{\bm U}\Om)\bm\Phi^\id\big)+\bm\Lambda .
\end{equation}
Moreover, $\|\Acal(\xi)\|\le 2\|\Om\|_2^2\big(\|\bm\Phi^\sy\|_F^2+\|\bm\Phi^\id\|_F^2\big)^2$ and $\|\bm B(\xi)\|\le 2\|\Om\|_2^2\big(\|\bm\Phi^\sy\|_F^2+\|\bm\Phi^\id\|_F^2\big)\|\X\|_2^2+\|\bm\Lambda\|$.
\end{lemma}

\begin{proof}

Since $\bm U\mapsto q_{\bm U}(\Zt,\Zt^\top)$ is linear, $\Lcal(\bm U,\xi)=\|\Om^{1/2}(\X\X^\top-q_{\bm U}(\Zt,\Zt^\top))\Om^{1/2}\|_F^2$ is a convex quadratic with Hessian $\Acal(\xi)$ as stated, the penalty is linear, and $\Lcal(\bm 0,\xi)=\|\Om^{1/2}\X\X^\top\Om^{1/2}\|_F^2=(\X^\top\Om\X)^2$. This gives \eqref{eqRquadratic}. For the gradient, $\langle\bm Y,\bm\Phi^\sy\bm U^\sy\bm\Phi^{\sy\top}\rangle_F=\tr(\bm\Phi^{\sy\top}\bm Y\bm\Phi^\sy\bm U^\sy)$ and $\langle\bm Y,\Diag(\bm\Phi^\id\bm U^\id\bm\Phi^{\id\top})\rangle_F=\tr(\bm\Phi^{\id\top}\Diag(\bm Y)\bm\Phi^\id\bm U^\id)$ for symmetric $\bm Y$, so that the derivative of $\|\Om^{1/2}\bm Y_{\bm U}\Om^{1/2}\|_F^2$ in the direction $\bm U'$ is $-2\langle\Om\bm Y_{\bm U}\Om,q_{\bm U'}(\Zt,\Zt^\top)\rangle_F=\langle(-2\bm\Phi^{\sy\top}\Om\bm Y_{\bm U}\Om\bm\Phi^\sy,-2\bm\Phi^{\id\top}\Diag(\Om\bm Y_{\bm U}\Om)\bm\Phi^\id),\bm U'\rangle$, which is \eqref{eqABc}. For the bounds, $\|\Om^{1/2}\bm A\Om^{1/2}\|_F\le\|\Om\|_2\|\bm A\|_F$, $\|\bm A\bm U\bm A^\top\|_F\le\|\bm A\|_F^2\|\bm U\|_F$ and $\|\Diag(\bm A)\|_F\le\|\bm A\|_F$ give $\|q_{\bm U}(\Zt,\Zt^\top)\|_F\le(\|\bm\Phi^\sy\|_F^2+\|\bm\Phi^\id\|_F^2)\|\bm U\|$, hence $\|\Acal(\xi)\|\le2\|\Om\|_2^2(\|\bm\Phi^\sy\|_F^2+\|\bm\Phi^\id\|_F^2)^2$, and $\|\bm B(\xi)\|=\|\nabla\Rcal(\bm 0,\xi)\|\le2\|\Om\|_2^2(\|\bm\Phi^\sy\|_F^2+\|\bm\Phi^\id\|_F^2)\|\X\X^\top\|_F+\|\bm\Lambda\|$ with $\|\X\X^\top\|_F=\|\X\|_2^2$, by the same estimate applied to the adjoint.\end{proof}

\subsection{Injectivity and Strong Convexity}

For a data point $\xi$, consider the single-month map
\begin{equation}\label{eqcsmap}
   \bm U\mapsto q_{\bm U}(\Zt,\Zt^\top) .
\end{equation}
Let $\sigma_{\min}(\Zt)\coloneqq\min\{\|q_{\bm U}(\Zt,\Zt^\top)\|_F:\bm U\in\Sm^\circ,\ \|\bm U\|=1\}$ be its smallest singular value on $\Sm^\circ$, $\sigma_{\min}(\bm A)\coloneqq\min_{\|\bm v\|_2=1}\|\bm A\bm v\|_2$ for a matrix $\bm A$, and $\lambda_{\min}(\Om)$ the smallest eigenvalue of the weight.

\begin{lemma}[Injectivity and strong convexity]\label{lemstrongconvex}
Fix a data point $\xi=(\bm x,\bm z)$ with $N$ assets.
\begin{enumerate}
\item\label{lemsc1} $\Rcal(\cdot,\xi)$ is $\alpha$-strongly convex with $\alpha=2\lambda_{\min}(\Om)^2\sigma_{\min}(\Zt)^2$, and $\alpha>0$ if and only if the map \eqref{eqcsmap} is injective. This holds if and only if (a) $\rank\bm\phi^\sy(\bm z)=m^\sy$, (b) the $N$-vectors $\big[\phi^\id_i(z_k)\phi^\id_j(z_k)\big]_{k=1}^N$, $1\le i\le j\le m^\id$, are linearly independent, and (c) no nonzero matrix $\bm\phi^\sy(\bm z)\bm V\bm\phi^\sy(\bm z)^\top$, $\bm V\in\Sy^{m^\sy}$, is of the form $\Diag(\bm\phi^\id(\bm z)\bm U^\id\bm\phi^\id(\bm z)^\top)$ for some $\bm U^\id\in\Sy^{m^\id}$.
\item\label{lemsc3} In the isotropic specification the map \eqref{eqcsmap} is injective if and only if $\rank\bm\phi^\sy(\bm z)=m^\sy$ and $N>m^\sy$, and
\begin{equation}\label{eqsciso}
\begin{gathered}
   \sigma_{\min}(\Zt)^2\ \ge\ \omega_N\,\min\{1,\sigma_{\min}(\bm\phi^\sy(\bm z))\}^4,\\ \omega_N\coloneqq\tfrac{1}{2}\Big(N+1-\sqrt{(N-1)^2+4m^\sy}\Big)\ge 1-\tfrac{m^\sy}{N-1}\quad (N\ge2).
\end{gathered}
\end{equation}
\end{enumerate}
\end{lemma}

Condition (c) holds whenever every row of $\bm\phi^\sy(\bm z)$ lies in the span of the other rows, and (a)--(c) hold for generic feature values if $N>m^\sy$ and $N\ge m^\id(m^\id+1)/2$. For the Frobenius weight, part \ref{lemsc3} gives $\alpha\ge 2\omega_N\min\{1,\sigma_{\min}(\bm\phi^\sy(\bm z))\}^4$, of order two when the smallest singular value of the feature matrix exceeds one and $m^\sy/N$ stays away from one.

\begin{proof}

\ref{lemsc1}: By Lemma~\ref{lemquadratic}, $\Rcal(\cdot,\xi)$ is $\alpha$-strongly convex if and only if $2\|\Om^{1/2}q_{\bm U}(\Zt,\Zt^\top)\Om^{1/2}\|_F^2\ge\alpha\|\bm U\|^2$ for all $\bm U\in\Sm^\circ$. Since $\|\Om^{1/2}\bm A\Om^{1/2}\|_F\ge\lambda_{\min}(\Om)\|\bm A\|_F$, this holds with $\alpha=2\lambda_{\min}(\Om)^2\sigma_{\min}(\Zt)^2$, where the minimum defining $\sigma_{\min}(\Zt)$ is attained by compactness of the unit sphere in $\Sm^\circ$. As $\Om\succ\bm0$, $\alpha$ is positive if and only if $\sigma_{\min}(\Zt)>0$, that is, if and only if the map \eqref{eqcsmap} is injective on $\Sm^\circ$, equivalently on $\Sm$. Write $q_{\bm U}(\Zt,\Zt^\top)=\bm L(\bm U^\sy)+\bm H(\bm U^\id)$ with $\bm L(\bm U^\sy)\coloneqq\bm\Phi^\sy(\Zt)\bm U^\sy\bm\Phi^\sy(\Zt)^\top$ and $\bm H(\bm U^\id)\coloneqq\Diag(\bm\Phi^\id(\Zt)\bm U^\id\bm\Phi^\id(\Zt)^\top)$. If the sum is injective, then so are $\bm L$ and $\bm H$, and their images intersect only in $\bm 0$. Conversely, if $\bm L(\bm U^\sy)+\bm H(\bm U^\id)=\bm0$ then $\bm L(\bm U^\sy)=-\bm H(\bm U^\id)$ lies in both images, so (c) gives $\bm L(\bm U^\sy)=\bm H(\bm U^\id)=\bm 0$ and (a), (b) give $\bm U=\bm 0$. Now $\bm L$ is injective if and only if $\rank\bm\Phi^\sy(\Zt)=m^\sy+1$, that is, (a), since $\bm\Phi^\sy(\Zt)\bm v=\bm 0$ for $\bm v\ne\bm0$ gives $\bm L(\bm v\bm v^\top)=\bm 0$, while under (a) $\bm L(\bm U^\sy)=\bm0$ implies $\bm U^\sy=\bm\Phi^\sy(\Zt)^+\bm L(\bm U^\sy)\bm\Phi^\sy(\Zt)^{+\top}=\bm 0$. And $\bm H$ is injective if and only if the vectors $[\phi^\id_i(z_k)\phi^\id_j(z_k)]_k$, $i\le j$, are linearly independent, that is, (b), because $\bm H(\bm U^\id)$ has diagonal $\sum_{i\le j}(2-1_{i=j})U^\id_{ij}[\phi^\id_i(z_k)\phi^\id_j(z_k)]_k$. Every matrix in the image of $\bm H$ has a zero first row and column. Since $\bm\Phi^\sy(\Zt)=\Diag(1,\bm\phi^\sy(\bm z))$ in block form, under (a) a matrix $\bm L(\bm U^\sy)$ with zero first row and column has $\bm U^\sy=\Diag(0,\bm V)$ and equals $\Diag(0,\bm\phi^\sy(\bm z)\bm V\bm\phi^\sy(\bm z)^\top)$, so the images of $\bm L$ and $\bm H$ intersect only in $\bm 0$ if and only if no nonzero $\bm\phi^\sy(\bm z)\bm V\bm\phi^\sy(\bm z)^\top$ is of the form $\Diag(\bm\phi^\id(\bm z)\bm U^\id\bm\phi^\id(\bm z)^\top)$, which is (c). For the sufficient condition, let $\bm\phi\coloneqq\bm\phi^\sy(\bm z)$ and suppose $\bm\phi\bm V\bm\phi^\top=\diag(\bm d)$ with $d_k\ne0$. Then $\bm e_k\in\Ima(\bm\phi\bm V\bm\phi^\top)\subseteq\Ima\bm\phi$. Hence if no $\bm e_k$ lies in $\Ima\bm\phi$, every diagonal matrix in the image of $\bm L$ is zero. Finally, $\bm e_k\in\Ima\bm\phi$ means that some $\bm v$ satisfies $\bm\phi_{k\cdot}\bm v=1$ and $\bm\phi_{j\cdot}\bm v=0$ for $j\ne k$, that is, that row $k$ is not in the span of the other rows. For $N>m^\sy$ and rows in general position this never happens, while (a) and (b) hold generically under the stated dimension counts.

An auxiliary bound. For a matrix $\bm A$ with full column rank, $\bm A^+\bm A=\bm I$ and $\|\bm A^+\|_2=1/\sigma_{\min}(\bm A)$, so $\|\bm U\|_F=\|\bm A^+(\bm A\bm U\bm A^\top)\bm A^{+\top}\|_F\le\|\bm A\bm U\bm A^\top\|_F/\sigma_{\min}(\bm A)^2$. Applied to $\bm A=\bm\Phi^\sy(\Zt)$, whose singular values are $1$ and those of $\bm\phi^\sy(\bm z)$, this gives $\|\bm\Phi^\sy(\Zt)\bm U^\sy\bm\Phi^\sy(\Zt)^\top\|_F\ge\min\{1,\sigma_{\min}(\bm\phi^\sy(\bm z))\}^2\|\bm U^\sy\|_F$.

\ref{lemsc3}: Let $\bm U=(\bm U^\sy,u)\in\Sy^{m^\sy+1}\times\R$, $\bm L\coloneqq\bm\Phi^\sy(\Zt)\bm U^\sy\bm\Phi^\sy(\Zt)^\top$, $\bm J\coloneqq\Diag(0,1,\dots,1)\in\Sy^{N+1}$, so that $q_{\bm U}(\Zt,\Zt^\top)=\bm L+u\bm J$ and $\|q_{\bm U}(\Zt,\Zt^\top)\|_F^2=\|\bm L\|_F^2+2u\tr(\bm J\bm L)+Nu^2$. If $\bm\phi\coloneqq\bm\phi^\sy(\bm z)$ is rank deficient, or $N\le m^\sy$, the right-hand side of \eqref{eqsciso} is at most zero and there is nothing to prove. So assume that $\rank\bm\phi=m^\sy$ and $N\ge2$, so that $\bm M\coloneqq\bm\Phi^\sy(\Zt)^\top\bm\Phi^\sy(\Zt)=\Diag(1,\bm\phi^\top\bm\phi)\succ0$, and let $\bm D\coloneqq\bm\Phi^\sy(\Zt)^\top\bm J\bm\Phi^\sy(\Zt)=\Diag(0,\bm\phi^\top\bm\phi)$. Then $\tr(\bm J\bm L)=\tr(\bm U^\sy\bm D)=\langle\bm M^{1/2}\bm U^\sy\bm M^{1/2},\bm M^{-1/2}\bm D\bm M^{-1/2}\rangle_F$, and since $\|\bm M^{1/2}\bm U^\sy\bm M^{1/2}\|_F^2=\tr(\bm U^\sy\bm M\bm U^\sy\bm M)=\|\bm L\|_F^2$ and $\bm M^{-1/2}\bm D\bm M^{-1/2}=\Diag(0,\bm I_{m^\sy})$ has Frobenius norm $\sqrt{m^\sy}$, the Cauchy--Schwarz inequality gives $|\tr(\bm J\bm L)|\le\sqrt{m^\sy}\|\bm L\|_F$. Hence
\[ \|q_{\bm U}(\Zt,\Zt^\top)\|_F^2\ge\|\bm L\|_F^2-2\sqrt{m^\sy}\,|u|\,\|\bm L\|_F+Nu^2\ge \omega_N\big(\|\bm L\|_F^2+u^2\big),\]
where $\omega_N$ is the smallest eigenvalue of $\begin{bmatrix}1&-\sqrt{m^\sy}\\-\sqrt{m^\sy}&N\end{bmatrix}$, namely $\omega_N=\frac{N+1-\sqrt{(N-1)^2+4m^\sy}}2$. By the auxiliary bound, $\|\bm L\|_F^2+u^2\ge\min\{1,\sigma_{\min}(\bm\Phi^\sy(\Zt))^4\}(\|\bm U^\sy\|_F^2+u^2)=\min\{1,\sigma_{\min}(\bm\phi)\}^4\|\bm U\|^2$, which proves \eqref{eqsciso}. Since $\sqrt{(N-1)^2+4m^\sy}\le (N-1)+2m^\sy/(N-1)$ for $N\ge2$, we get $\omega_N\ge1-m^\sy/(N-1)$, and $\omega_N>0$ if and only if $(N+1)^2>(N-1)^2+4m^\sy$, that is, $N>m^\sy$. Hence the map is injective if $\rank\bm\phi=m^\sy$ and $N>m^\sy$. Conversely, if $\bm\phi$ is rank deficient, the map is not injective by \ref{lemsc1}. If $N=m^\sy$ with $\bm\phi$ invertible, then $\bm\Phi^\sy(\Zt)$ is invertible and $\bm U\coloneqq(-\bm\Phi^\sy(\Zt)^{-1}\bm J\bm\Phi^\sy(\Zt)^{-\top},1)\ne\bm0$ satisfies $q_{\bm U}(\Zt,\Zt^\top)=\bm0$.\end{proof}

\begin{remark}[Identification of $\bm U$ by the kernel $q_{\bm U}$]\label{remidentU}
By its definition in Section~\ref{sec_reparam}, $\bm U$ determines the function $q_{\bm U}$ on $\Zcal_\Delta\times\Zcal_\Delta$. The converse, injectivity of the linear map $\bm U\mapsto q_{\bm U}$ on $\Sm$, holds if the products $\phi^\id_i\phi^\id_j$, $i\le j$, are linearly independent functions on $\Zcal$, and the functions $\bm\phi^\sy$ are linearly independent on $\Zcal\setminus\{z_0\}$ for every $z_0\in\Zcal$. The second condition is mild. The first is trivially satisfied for $m^\id=1$ with $\phi^\id_1\not\equiv0$, but it restricts richer idiosyncratic specifications, since only the quadratic form $z\mapsto\bm\phi^\id(z)\bm U^\id\bm\phi^\id(z)^\top$ is observable. Conditions (a)--(c) of Lemma~\ref{lemstrongconvex} are the sample counterparts of these conditions.
\end{remark}

\subsection{Proof of Proposition~\ref{coruniqueness}}

\ref{corun1}: For fixed $\xi_t$, $\bm U\mapsto\X_{t+1}\X_{t+1}^\top-q_{\bm U}(\Zt_t,\Zt_t^\top)$ is affine, so $\Lcal(\cdot,\xi_t)$ is a convex quadratic, the trace penalty is linear, and $\Rcal_T$ is a convex quadratic on $\Sm$. If $\lsy,\lid>0$, then on $\Dcal$ we have $\Rcal_T(\bm U)\ge\lsy\tr\bm U^\sy+\lid\tr\bm U^\id\ge\min\{\lsy,\lid\}(\|\bm U^\sy\|_F+\|\bm U^\id\|_F)$, using $\tr\bm A\ge\|\bm A\|_F$ for $\bm A\succeq0$. Thus $\Rcal_T$ is coercive on the closed set $\Dcal$, its sublevel sets are compact, and a minimizer exists. If the map \eqref{eqsamplemap} is injective, existence follows from part \ref{corun3}.

\ref{corun2}: Let $\bm U_1,\bm U_2$ be two solutions of \eqref{krr1convex} and $\bar{\bm U}\coloneqq(\bm U_1+\bm U_2)/2\in\Dcal$. By linearity of $\bm U\mapsto q_{\bm U}(\Zt_t,\Zt_t^\top)$ and of the trace penalty, and the parallelogram identity $\|\frac{a+b}2\|_F^2=\frac12\|a\|_F^2+\frac12\|b\|_F^2-\frac14\|a-b\|_F^2$ applied to $a=\Om_t^{1/2}(\X_{t+1}\X_{t+1}^\top-q_{\bm U_1}(\Zt_t,\Zt_t^\top))\Om_t^{1/2}$ and $b=\Om_t^{1/2}(\X_{t+1}\X_{t+1}^\top-q_{\bm U_2}(\Zt_t,\Zt_t^\top))\Om_t^{1/2}$,
\[ \Rcal_T(\bar{\bm U})=\tfrac12\Rcal_T(\bm U_1)+\tfrac12\Rcal_T(\bm U_2)-\frac1T\sum_{t=0}^{T-1}\frac{1}{4}\big\|\Om_t^{1/2}\big(q_{\bm U_1}(\Zt_t,\Zt_t^\top)-q_{\bm U_2}(\Zt_t,\Zt_t^\top)\big)\Om_t^{1/2}\big\|_F^2 .\] Since $\bm U_1,\bm U_2$ are minimizers, the last sum must vanish, so $q_{\bm U_1}(\Zt_t,\Zt_t^\top)=q_{\bm U_2}(\Zt_t,\Zt_t^\top)$ for all $t$ (as $\Om_t\succ\bm0$).

\ref{corun3}: If the map \eqref{eqsamplemap} is injective, then the quadratic part $\frac1T\sum_t\|\Om_t^{1/2}q_{\bm U}(\Zt_t,\Zt_t^\top)\Om_t^{1/2}\|_F^2$ of $\Rcal_T$ is positive for $\bm U\ne\bm 0$, so $\Rcal_T$ is strictly convex and coercive on $\Sm$, and it has a unique minimizer on the closed convex set $\Dcal$. Uniqueness of $q_{\bm U^\ast_T}$ follows. The kernel of \eqref{eqsamplemap} is the intersection of the kernels of the maps \eqref{eqcsmap} over $t<T$, so injectivity of \eqref{eqcsmap} for one $t$ implies injectivity of \eqref{eqsamplemap}. In that case Lemma~\ref{lemstrongconvex}\ref{lemsc1} gives $\langle\bm U,\Acal(\xi_t)\bm U\rangle\ge2\lambda_{\min}(\Om_t)^2\sigma_{\min}(\Zt_t)^2\|\bm U\|^2$ on $\Sm^\circ$, and since all $\Acal(\xi_j)$ are PSD, $\Rcal_T$ is strongly convex with $\alpha=\frac2T\lambda_{\min}(\Om_t)^2\sigma_{\min}(\Zt_t)^2$. In general, injectivity of \eqref{eqsamplemap} makes $\Acal_T$ positive definite on $\Sm^\circ$, which is strong convexity.

\subsection{KKT Conditions and Number of Factors}\label{app_proofproprank}

The KKT conditions characterize the solutions of the CMK program on $\Dcal_\eta$. Let $\nabla^\sy\Lcal_T(\bm U)$ and $\nabla^\id\Lcal_T(\bm U)$ be the blocks of the gradient of the sample loss $\Lcal_T(\bm U)\coloneqq\frac1T\sum_{t<T}\Lcal(\bm U,\xi_t)$.

\begin{proposition}[KKT conditions and number of factors]\label{proprank}
Let $\lsy,\lid,\eta\ge0$, and let $\Lcal$ be a nonnegative convex quadratic loss, such as \eqref{lossdef} or \eqref{eqlossrho}. A point $\bm U^\ast\in\Dcal_\eta$ minimizes $\Lcal_T(\bm U)+\lsy\tr\bm U^\sy+\lid\tr\bm U^\id$ over $\Dcal_\eta$ if and only if there is $\nu\in\R$ such that
\begin{equation}\label{eqKKTrank}
\begin{gathered}
   \bm\Xi^\sy\coloneqq\nabla^\sy\Lcal_T(\bm U^\ast)+\lsy\bm I+\nu\bm E_{00}\succeq0,\qquad \bm\Xi^\sy(\bm U^{\sy\ast}-\bm F_\eta)=\bm 0,\\
   \bm\Xi^\id\coloneqq\nabla^\id\Lcal_T(\bm U^\ast)+\lid\bm I\succeq0,\qquad \bm\Xi^\id\bm U^{\id\ast}=\bm 0 .
\end{gathered}
\end{equation}
\begin{enumerate}
\item\label{proprank1} The columns of $\bm U^{\sy\ast}-\bm F_\eta$ lie in the eigenspace of $-\nabla^\sy\Lcal_T(\bm U^\ast)-\nu\bm E_{00}$ for its largest eigenvalue $\lsy$. The number of factors, $\rank\tilde{\bm C}^\ast=\rank(\bm U^{\sy\ast}-\bm F_\eta)-1$, is at most the dimension of this eigenspace minus one.
\item\label{proprank3} If $\lsy>0$, then $\|\muf^\ast\|_2^2+\tr\tilde{\bm C}^\ast\le\Lcal_T(\bm E_{00}+\bm F_\eta,\bm 0)/\lsy$.
\end{enumerate}
\end{proposition}

Part \ref{proprank1} is the semidefinite analog of the active-set condition of the lasso. For large $\lsy$ the excess is zero and the mean is of order $1/\lsy$.

\begin{proof}[Proof of Proposition~\ref{proprank}]
The problem minimizes a nonnegative convex quadratic subject to $\bm U^\sy-\bm F_\eta\succeq0$, $\bm U^\sy_{00}=1$ and $\bm U^\id\succeq0$. The first constraint implies $\bm U^\sy\succeq0$, since $\bm F_\eta\succeq0$, so the feasible set is $\Dcal_\eta$. Slater's condition holds: $\bm U^\sy=\bm F_\eta+\bm I_{m^\sy+1}$, $\bm U^\id=\bm I_{m^\id}$ is strictly feasible, since $(\bm F_\eta)_{00}=0$. The KKT conditions are therefore necessary and sufficient \citep[Section~5.9]{boyd04}: there are $\bm\Xi^\sy,\bm\Xi^\id\succeq0$ and $\nu\in\R$ with $\nabla\Lcal_T(\bm U^\ast)+(\lsy\bm I,\lid\bm I)+\nu(\bm E_{00},\bm 0)=(\bm\Xi^\sy,\bm\Xi^\id)$, $\langle\bm\Xi^\sy,\bm U^{\sy\ast}-\bm F_\eta\rangle=0$ and $\langle\bm\Xi^\id,\bm U^{\id\ast}\rangle=0$. Since $\langle\bm A,\bm B\rangle=0$ for PSD $\bm A,\bm B$ forces $\bm A\bm B=\bm 0$, this is \eqref{eqKKTrank}.

\ref{proprank1}: From $\bm\Xi^\sy(\bm U^{\sy\ast}-\bm F_\eta)=\bm0$ we get $\Ima(\bm U^{\sy\ast}-\bm F_\eta)\subseteq\ker\bm\Xi^\sy$, and $-(\nabla^\sy\Lcal_T(\bm U^\ast)+\nu\bm E_{00})=\lsy\bm I-\bm\Xi^\sy$ with $\bm\Xi^\sy\succeq0$ has all eigenvalues at most $\lsy$, with eigenvalue $\lsy$ exactly on $\ker\bm\Xi^\sy$. Hence $\rank(\bm U^{\sy\ast}-\bm F_\eta)\le\dim\ker\bm\Xi^\sy$, the multiplicity of the largest eigenvalue $\lsy$. In block form, $\bm U^{\sy\ast}-\bm F_\eta$ has entry $1$ in position $(0,0)$, first column $\muf^\ast$ and Schur complement $\bm V^\ast-\muf^\ast\muf^{\ast\top}-\eta\hat{\bm F}=\tilde{\bm C}^\ast$, so $\rank(\bm U^{\sy\ast}-\bm F_\eta)=1+\rank\tilde{\bm C}^\ast$. 

\ref{proprank3}: The point $(\bm E_{00}+\bm F_\eta,\bm 0)$ is feasible. Optimality against it and $\Lcal_T\ge0$ give $\lsy\tr\bm U^{\sy\ast}\le\Lcal_T(\bm U^\ast)+\lsy\tr\bm U^{\sy\ast}\le\Lcal_T(\bm E_{00}+\bm F_\eta,\bm 0)+\lsy(1+\eta\tr\hat{\bm F})$, and $\tr\bm U^{\sy\ast}-1-\eta\tr\hat{\bm F}=\|\muf^\ast\|_2^2+\tr\tilde{\bm C}^\ast$ by \eqref{eqpenaltydecomp}.

Large $\lsy$: Part \ref{proprank3} bounds $\bm U^{\sy\ast}$. Optimality also gives $\Lcal_T(\bm U^\ast)\le\Lcal_T(\bm E_{00}+\bm F_\eta,\bm 0)$, since $\tr\bm U^{\sy\ast}\ge1+\eta\tr\hat{\bm F}$, so that in the isotropic specification $u^{\id\ast}$ is bounded, since $n_t>0$, and the gradient $\nabla^\sy\Lcal_T(\bm U^\ast)$ stays bounded as $\lsy$ grows. For large $\lsy$ the lower right block of $\bm\Xi^\sy$ is therefore positive definite, the null space of $\bm\Xi^\sy$ has dimension one, and $\bm U^{\sy\ast}-\bm F_\eta$ has rank one. Hence $\tilde{\bm C}^\ast=\bm 0$, and $\muf^\ast=-(\nabla^\sy_{11}+\lsy\bm I)^{-1}\nabla^\sy_{10}$, with $\nabla^\sy_{10}$ and $\nabla^\sy_{11}$ the lower left and lower right blocks of $\nabla^\sy\Lcal_T(\bm U^\ast)$, is of order $1/\lsy$. So for large $\lsy$ the solution is the floor with a mean of order $1/\lsy$.
\end{proof}

\begin{corollary}[Isotropic idiosyncratic variance]\label{corisou}
In the setting of Proposition~\ref{proprank} with a loss of the form \eqref{eqsuffstatiso}, such as \eqref{lossdef} in the isotropic specification or \eqref{eqlossrho}, the idiosyncratic variance is a soft-thresholded weighted residual variance,
\begin{equation*} u^{\id\ast}=\max\Big\{0,\ \frac{\bar r-\lid/2}{\bar n}\Big\},\qquad \bar r\coloneqq\frac1T\sum_{t=0}^{T-1}\big(r_t-\tr(\bm U^{\sy\ast}\bm D_t)\big),\qquad \bar n\coloneqq\frac1T\sum_{t=0}^{T-1}n_t ,\end{equation*}
with the statistics $r_t$, $\bm D_t$ and $n_t$ of the loss, and $\bar n>0$. For the Frobenius weight, $\bar r$ is the average over the sample of $\sum_{i=1}^{N_t}\big(x_{t+1,i}^2-\aphi^\sy(z_{t,i})\bm U^{\sy\ast}\aphi^\sy(z_{t,i})^\top\big)$ and $\bar n$ the average cross-sectional size.
\end{corollary}

\begin{proof}
Here $\bm U^\id=u^\id\ge0$ is a scalar, and by \eqref{eqsuffstatiso} the derivative of the sample loss is $\partial_{u}\Lcal_T=\frac1T\sum_{t=0}^{T-1}\big[-2r_t+2\tr(\bm U^{\sy}\bm D_t)+2n_tu^\id\big]=-2\bar r+2\bar nu^\id$. The conditions $\Xi^\id=-2\bar r+2\bar nu^{\id\ast}+\lid\ge0$ and $\Xi^\id u^{\id\ast}=0$ in \eqref{eqKKTrank} give the claim. For the Frobenius weight, $\tr(\bm U^\sy(\bm M_t-\bm E_{00}))=\sum_{i\ge1}\aphi^\sy(z_{t,i})\bm U^\sy\aphi^\sy(z_{t,i})^\top$.
\end{proof}

\subsection{Proof of Lemma~\ref{lemblocks}}\label{app_prooflemblocks}

We prove the identity \eqref{eqUhat} stated before the lemma and the lemma. Drop the index $t$ and write $\bm\phi\coloneqq\bm\phi^\sy_t$, $\bm G\coloneqq\bm\phi^\top\bm\phi$, $\bm U\coloneqq\bm U^\sy$ and $u\coloneqq u^\id$. Then $\bm\phi^+=\bm G^+\bm\phi^\top$, $\bm P=\bm\phi\bm G^+\bm\phi^\top$, $\bm\Pi=\bm G\bm G^+=\bm G^+\bm G$, $\bm\Pi\bm f=\bm f$ and $\bm\Pi\bm G^+\bm\Pi=\bm G^+$.

Identity \eqref{eqUhat}: for any weight, $\bm s=\bm\Phi^{\sy\top}\Om^{1/2}(\Om^{1/2}\X)$ and $\bm D=(\bm\Phi^{\sy\top}\Om^{1/2})(\Om^{1/2}\bm J)(\bm J\Om^{1/2})(\Om^{1/2}\bm\Phi^\sy)$ have their columns in $\Ima(\bm\Phi^{\sy\top}\Om^{1/2})=\Ima\bm M$, so $\bm M\bm M^+(\bm s\bm s^\top-u\bm D)\bm M^+\bm M=\bm s\bm s^\top-u\bm D$. The terms of \eqref{eqsuffstatiso} that involve $\bm U$ are $\tr(\bm U\bm M\bm U\bm M)-2\tr\big(\bm U(\bm s\bm s^\top-u\bm D)\big)=\tr(\bm U\bm M\bm U\bm M)-2\tr(\bm U\bm M\hat{\bm U}\bm M)=\|\bm M^{1/2}(\bm U-\hat{\bm U})\bm M^{1/2}\|_F^2-\|\bm M^{1/2}\hat{\bm U}\bm M^{1/2}\|_F^2$, and the last term does not involve $\bm U$.

Gram-whitened weight, $\Om=\Diag(\kappa,\tilde{\bm\Omega})$ with $\tilde{\bm\Omega}\coloneqq\bm\phi(\bm G^+)^2\bm\phi^\top+\bm I_N-\bm P$: since $(\bm I_N-\bm P)\bm\phi=\bm0$, the two summands of $\tilde{\bm\Omega}$ act on orthogonal subspaces, and $\tilde{\bm\Omega}$ has the eigenvalues $1/\lambda_k(\bm G)$ on $\Ima\bm\phi$ and $1$ on its complement, so $\Om^{\rm gw}\succ\bm0$. With $\tilde{\bm\Omega}\bm\phi=\bm\phi(\bm G^+)^2\bm G$, $\tilde{\bm\Omega}^2=\bm\phi(\bm G^+)^3\bm\phi^\top+\bm I_N-\bm P$ and $\bm\epsilon\coloneqq(\bm I_N-\bm P)\bm x$, the statistics \eqref{eqsuffstatdef} are
\[ \bm M=\Diag(\kappa,\bm G(\bm G^+)^2\bm G)=\Diag(\kappa,\bm\Pi),\qquad \bm s=(\kappa,\bm G(\bm G^+)^2\bm\phi^\top\bm x)=(\kappa,\bm f),\]
$\bm D=\Diag(0,\bm G(\bm G^+)^3\bm G)=\Diag(0,\bm G^+)$, $a=\kappa+\|\bm f\|_2^2+\|\bm\epsilon\|_2^2$, $r_t=\bm f^\top\bm G^+\bm f+\|\bm\epsilon\|_2^2$ and $n_t=\tr((\bm G^+)^2)+N-\rank\bm\phi$. Hence $\bm M^+\bm s=(1,\bm f)$, $\bm M^+\bm D\bm M^+=\Diag(0,\bm G^+)$, and $\hat{\bm U}$ is the matrix \eqref{eqUhatgw}. By the identity, the terms of the loss without $\bm U$ are $a^2-2ur_t+n_t u^2-\|\bm M^{1/2}\hat{\bm U}\bm M^{1/2}\|_F^2$, and $\|\bm M^{1/2}\hat{\bm U}\bm M^{1/2}\|_F^2=\kappa^2+2\kappa\|\bm f\|_2^2+\|\bm f\bm f^\top\|_F^2-2u\,\bm f^\top\bm G^+\bm f+u^2\tr((\bm G^+)^2)$. Collecting the terms in $u$ gives $-2u\|\bm\epsilon\|_2^2+(N-\rank\bm\phi)u^2=(N-\rank\bm\phi)(u-\hat u_t)^2-(N-\rank\bm\phi)\hat u_t^2$, using $N>\rank\bm\phi$. This is the last term of \eqref{eqlossgram}. For $\bm U\in\Dcal$, $\bm M^{1/2}(\bm U-\hat{\bm U})\bm M^{1/2}$ has the $(0,0)$ entry $0$, the first column $\sqrt\kappa\,\bm\Pi(\muf-\bm f)$ and the lower block $\bm\Pi(\bm V-\bm f\bm f^\top)\bm\Pi+u\bm G^+$, whose squared Frobenius norms add up to the first two terms of \eqref{eqlossgram}.

\subsection{Proof of Proposition~\ref{propclosedform}}

For $\bm U\in\Dcal$ and fixed $u^\id$ the loss \eqref{eqlossrho} is $\|\bm M^{1/2}(\bm U^\sy-\hat{\bm U}_{t})\bm M^{1/2}\|_F^2$ with $\bm M\coloneqq\Diag(\kappa,\bm I_{m^\sy})$ and the single-month estimate $\hat{\bm U}_{t}$ of \eqref{eqUhatgw} with $\bm f^\varrho_{t+1}$ and $\bm G^+_{t,\varrho}$ in place of $\bm f_{t+1}$ and $\bm G^+_t$, plus terms without $\bm U^\sy$. Its average over the window is $\|\bm M^{1/2}(\bm U^\sy-\bar{\bm U})\bm M^{1/2}\|_F^2$ plus a constant, where $\bar{\bm U}$ has the entry $1$, the first column $\bar{\bm f}$ and the lower right block $\bar{\bm A}(u^\id)$. This is $2\kappa\|\muf-\bar{\bm f}\|_2^2+\|\bm V-\bar{\bm A}(u^\id)\|_F^2$. With $u^\id=\hat u$ and $\muf=\bar{\bm f}$, the decomposition \eqref{eqpenaltydecomp} and $\bm V=\bar{\bm f}\bar{\bm f}^\top+\eta\hat{\bm F}+\tilde{\bm C}$ give the objective $\|\tilde{\bm C}-(\bar{\bm A}-\bar{\bm f}\bar{\bm f}^\top-\eta\hat{\bm F})\|_F^2+\lsy\tr\tilde{\bm C}$ plus a constant, which equals $\|\tilde{\bm C}-(\bar{\bm A}-\bar{\bm f}\bar{\bm f}^\top-\eta\hat{\bm F}-\frac{\lsy}2\bm I)\|_F^2$ plus a constant. Its unique minimizer over the PSD cone is the projection \eqref{eqclosedform}.

Write $(\muf,\bm V,u^\id)$ for the one-step estimator $(\check{\muf},\check{\bm V},\check u)$. For the solution of \eqref{krr1floor}, the gradient of the sample loss in $\bm U^\sy$ is $2\bm M(\bm U^\sy-\bar{\bm U})\bm M$, with first column $(0,2\kappa(\muf-\bar{\bm f}))$ and lower right block $2(\bm V-\bar{\bm A}(u^\id))$, at the idiosyncratic variance $u^\id$ of the solution. By Proposition~\ref{proprank} with $\lid=0$ and $\lsy\ge0$, the KKT conditions \eqref{eqKKTrank} hold. The lower right block of $\bm\Xi^\sy$ is $2(\bm V-\bar{\bm A}(u^\id))+\lsy\bm I$, which is PSD. The first column of $\bm U^{\sy}-\bm F_\eta$ is $(1,\muf)$, so the last $m^\sy$ entries of the first column of $\bm\Xi^\sy(\bm U^{\sy}-\bm F_\eta)=\bm 0$ read $2\kappa(\muf-\bar{\bm f})+(2(\bm V-\bar{\bm A}(u^\id))+\lsy\bm I)\muf=\bm0$, which is the stated formula for $\muf$. If the excess has full rank, then $\bm U^\sy-\bm F_\eta$ is nonsingular, so $\bm\Xi^\sy=\bm 0$, $\muf=\bar{\bm f}$ and $\bm V=\bar{\bm A}(u^\id)-\frac{\lsy}2\bm I$. If moreover the idiosyncratic variance is $\hat u$, then $\bar{\bm A}-\bar{\bm f}\bar{\bm f}^\top-\eta\hat{\bm F}-\frac{\lsy}{2}\bm I$ equals the excess $\bm V-\muf\muf^\top-\eta\hat{\bm F}$, which is positive definite and therefore equals its projection \eqref{eqclosedform}.

\subsection{Proof of Proposition~\ref{propident}}

We first record that the population loss is well defined.

\begin{lemma}[Population loss]\label{lempoploss}
Under Assumption~\ref{assdata}, $\|\Acal(\xi)\|$, $\|\bm B(\xi)\|$ and $\Lcal(\bm 0,\xi)$ have finite expectation, so that the population loss \eqref{eqpoploss} and its derivatives $\nabla\Ecal(\bm U)=\Acal\bm U+\bm B$, $\nabla^2\Ecal=\Acal$ are well defined. Moreover, the sample averages $\Acal_T$, $\bm B_T$ and $\Lcal_T(\bm 0)$ converge to $\Acal$, $\bm B$ and $\E[\Lcal(\bm 0,\xi)]$ with probability one, and consequently $\Rcal_T\to\Ecal$ and $\nabla\Rcal_T\to\nabla\Ecal$ uniformly on compact subsets of $\Sm$ with probability one.
\end{lemma}

\begin{proof}
Since $\|\bm\Phi^\sy(\Zt)\|_F^2=1+\|\bm\phi^\sy(\bm z)\|_F^2$ and $\|\bm\Phi^\id(\Zt)\|_F=\|\bm\phi^\id(\bm z)\|_F$, the bounds of Lemma~\ref{lemquadratic} and $(a+b)^2\le2a^2+2b^2$ give $\|\Acal(\xi)\|\le 8\|\Om\|_2^2(1+\|\bm\phi^\sy(\bm z)\|_F^4+\|\bm\phi^\id(\bm z)\|_F^4)$, which is integrable by \eqref{cond_moments}. For $\bm B(\xi)$, the Cauchy--Schwarz inequality gives $\E[\|\Om\|_2^2ab]\le\E[\|\Om\|_2^2a^2]^{1/2}\E[\|\Om\|_2^2b^2]^{1/2}$ with $a\coloneqq\|\bm\Phi^\sy\|_F^2+\|\bm\Phi^\id\|_F^2$ and $b\coloneqq\|\X\|_2^2=1+\|\bm x\|_2^2$, both of which are finite by \eqref{cond_moments}. Finally, $\E[\Lcal(\bm 0,\xi)]\le\E[\|\Om\|_2^2(1+\|\bm x\|_2^2)^2]<\infty$, so \eqref{eqpoploss} is well defined. By Assumption~\ref{assdata}, the sequences $\Acal(\xi_t)$, $\bm B(\xi_t)$, $\Lcal(\bm 0,\xi_t)$ are stationary and ergodic with integrable entries (in any basis of the finite-dimensional spaces involved), so Birkhoff's ergodic theorem \citep[see, e.g.,][Section 24]{Billingsley1986} yields the almost sure convergence of the three averages. Uniform convergence on compacts follows from $|\Rcal_T(\bm U)-\Ecal(\bm U)|\le\frac12\|\Acal_T-\Acal\|\|\bm U\|^2+\|\bm B_T-\bm B\|\|\bm U\|+|\Lcal_T(\bm 0)-\E[\Lcal(\bm 0,\xi)]|$ and $\|\nabla\Rcal_T(\bm U)-\nabla\Ecal(\bm U)\|\le\|\Acal_T-\Acal\|\|\bm U\|+\|\bm B_T-\bm B\|$.
\end{proof}

\ref{propident1}: By conditional Jensen, $\E[\|\Om\|_2^2\|\Mcal(\xi)\|_F^2]\le\E[\|\Om\|_2^2\|\X\|_2^4]<\infty$, so all terms below are integrable. Write $\bm A_{\Om}\coloneqq\Om^{1/2}\bm A\Om^{1/2}$. Expanding
\begin{multline*} \|(\X\X^\top-q_{\bm U}(\Zt,\Zt^\top))_{\Om}\|_F^2=\|(\X\X^\top-\Mcal(\xi))_{\Om}\|_F^2\\+2\langle(\X\X^\top-\Mcal(\xi))_{\Om},(\Mcal(\xi)-q_{\bm U}(\Zt,\Zt^\top))_{\Om}\rangle_F+\|(\Mcal(\xi)-q_{\bm U}(\Zt,\Zt^\top))_{\Om}\|_F^2 ,\end{multline*}
and noting that $\Om$ and $\Mcal(\xi)-q_{\bm U}(\Zt,\Zt^\top)$ are functions of $\bm z$ while $\E[\X\X^\top-\Mcal(\xi)\mid \bm z]=\bm 0$, the conditional expectation of the cross term vanishes, and \eqref{eqpoplossdecomp} follows by taking expectations. The first term in \eqref{eqpoplossdecomp} does not depend on $\bm U$, which proves the projection statement for $\bm\Lambda=\bm 0$.

\ref{propident2}: Inserting $\Mcal(\xi)=q_{\bm U_0}(\Zt,\Zt^\top)$ into \eqref{eqpoplossdecomp} and using $\E[\|(q_{\bm U-\bm U_0}(\Zt,\Zt^\top))_{\Om}\|_F^2]=\frac12\langle\bm U-\bm U_0,\Acal(\bm U-\bm U_0)\rangle$ gives the displayed identity, because $\Ecal(\bm U_0)=\E[\|(\X\X^\top-q_{\bm U_0}(\Zt,\Zt^\top))_{\Om}\|_F^2]+\langle\bm\Lambda,\bm U_0\rangle$. For $\bm\Lambda=\bm0$, $\Ecal(\bm U)-\Ecal(\bm U_0)=\E[\|(q_{\bm U-\bm U_0}(\Zt,\Zt^\top))_{\Om}\|_F^2]\ge0$, so $\bm U_0$ is a minimizer, and any minimizer $\bm U^\ast\in\Dcal$ has $\E[\|(q_{\bm U^\ast-\bm U_0}(\Zt,\Zt^\top))_{\Om}\|_F^2]=0$, that is, $q_{\bm U^\ast}(\Zt,\Zt^\top)=q_{\bm U_0}(\Zt,\Zt^\top)=\Mcal(\xi)$ $\Pa$-a.s., since $\Om\succ\bm0$. If $\Ecal$ is strongly convex, the minimizer is unique.

\ref{propident3}: The difference $\bm U^\ast_{\bm\Lambda}-\bm U_0$ lies in $\Sm^\circ$. Strong convexity makes $\Ecal$ coercive on the closed set $\Dcal$, so $\bm U^\ast_{\bm\Lambda}$ exists and is unique, and $\Ecal(\bm U^\ast_{\bm\Lambda})\le\Ecal(\bm U_0)$ together with \ref{propident2} yields $\frac\alpha2\|\bm U^\ast_{\bm\Lambda}-\bm U_0\|^2\le\frac12\langle\bm U^\ast_{\bm\Lambda}-\bm U_0,\Acal(\bm U^\ast_{\bm\Lambda}-\bm U_0)\rangle\le-\langle\bm\Lambda,\bm U^\ast_{\bm\Lambda}-\bm U_0\rangle\le\|\bm\Lambda\|\|\bm U^\ast_{\bm\Lambda}-\bm U_0\|$.

\ref{propident4}: As in \ref{propident3}, $\bm U^\ast_{\rm fix}$ exists and is unique. Let $\bm S\coloneqq(\Diag(0,\bm F_{\rm fix}),\bm 0)$, so that $\|\bm S\|=\|\bm F_{\rm fix}\|_F$. $\bm U_0+\bm S$ lies in $\Dcal_{\rm fix}$, so $\Ecal(\bm U^\ast_{\rm fix})\le\Ecal(\bm U_0+\bm S)$, and \ref{propident2} with strong convexity gives $\frac\alpha2\|\bm U^\ast_{\rm fix}-\bm U_0\|^2\le\frac12\langle\bm S,\Acal\bm S\rangle\le\frac12\|\Acal\|\|\bm S\|^2$.

\subsection{Proof of Theorem~\ref{thmasy}}

The following elementary inequality is the backbone of the proof of Theorem~\ref{thmasy}.

\begin{lemma}[Basic inequality]\label{lembasicineq}
Let $\Kcal\subset\Dcal$ be closed and convex, and let $\Rcal_j(\bm U)=\frac12\langle\bm U,\Acal_j\bm U\rangle+\langle\bm B_j,\bm U\rangle$ plus a constant, $j=1,2$, be convex quadratics on $\Sm$ with self-adjoint PSD $\Acal_j$, and assume $\langle\bm U,\Acal_1\bm U\rangle\ge\alpha_1\|\bm U\|^2$ on $\Sm^\circ$ for some $\alpha_1\ge0$. If $\bm U_j\in\argmin_{\Kcal}\Rcal_j$, $j=1,2$, then
\[ \alpha_1\|\bm U_1-\bm U_2\|\le\big\|(\Acal_2-\Acal_1)\bm U_2+\bm B_2-\bm B_1\big\| .\]
\end{lemma}

\begin{proof}
Write $\bm U_{12}\coloneqq\bm U_1-\bm U_2$, which lies in $\Sm^\circ$. The first-order optimality conditions on the convex set $\Kcal$ read $\langle\nabla \Rcal_1(\bm U_1),\bm U_2-\bm U_1\rangle\ge0$ and $\langle\nabla \Rcal_2(\bm U_2),\bm U_1-\bm U_2\rangle\ge0$. Adding them gives $\langle\nabla \Rcal_1(\bm U_1)-\nabla \Rcal_2(\bm U_2),\bm U_{12}\rangle\le0$. Since $\nabla \Rcal_1(\bm U_1)=\Acal_1\bm U_{12}+\Acal_1\bm U_2+\bm B_1$ and $\nabla \Rcal_2(\bm U_2)=\Acal_2\bm U_2+\bm B_2$, this reads
\[ \alpha_1\|\bm U_{12}\|^2\le\langle\bm U_{12},\Acal_1\bm U_{12}\rangle\le\big\langle(\Acal_2-\Acal_1)\bm U_2+\bm B_2-\bm B_1,\bm U_{12}\big\rangle\le\big\|(\Acal_2-\Acal_1)\bm U_2+\bm B_2-\bm B_1\big\|\,\|\bm U_{12}\| ,\]
and the claim follows.
\end{proof}

Write $\alpha_\Ecal\coloneqq\min\{\langle\bm U,\Acal\bm U\rangle:\bm U\in\Sm^\circ,\|\bm U\|=1\}$, and $\bm\zeta_T\coloneqq\frac1T\sum_{t<T}\bm\zeta(\xi_t)=(\Acal_T-\Acal)\bm U^\ast+\bm B_T-\bm B$, where $\bm U^\ast$ is the minimizer of $\Ecal$ over $\Dcal$ (whose existence and uniqueness is established first in each part).

\ref{thmasy1}: As $\Ecal$ is strongly convex, $\alpha_\Ecal>0$, so $\Ecal$ is $\alpha_\Ecal$-strongly convex on the affine set $\{\bm U^\sy_{00}=1\}$, hence coercive there, and has a unique minimizer $\bm U^\ast$ on the closed convex set $\Dcal$. By Lemma~\ref{lempoploss}, $\Acal_T\to\Acal$ with probability one, and the same minimum for $\Acal_T$ converges to $\alpha_\Ecal$. With probability one there is therefore $T_0$ such that it is at least $\alpha_\Ecal/2$ for all $T\ge T_0$, so that $\Rcal_T$ is strongly convex and has a unique minimizer $\bm U^\ast_T$ on $\Dcal$. Lemma~\ref{lembasicineq} with $\Rcal_1=\Rcal_T$, $\Rcal_2=\Ecal$ and $\Kcal=\Dcal$ gives $\frac{\alpha_\Ecal}{2}\|\bm U^\ast_T-\bm U^\ast\|\le\|\bm\zeta_T\|$ for $T\ge T_0$, and $\bm\zeta_T\to\bm 0$ with probability one by Lemma~\ref{lempoploss}.

\ref{thmasy2}: If $\Rcal(\cdot,\xi)$ is $\alpha$-strongly convex for $\Pa$-a.e.\ $\xi$, then so are $\Rcal_T$ for every $T$ and $\Ecal$, so all minimizers exist and are unique, and Lemma~\ref{lembasicineq} gives $\alpha\|\bm U^\ast_T-\bm U^\ast\|\le\|\bm\zeta_T\|$ for every $T\ge1$. The random elements $\bm\zeta(\xi_t)$ are i.i.d.\ with $\E[\bm\zeta(\xi)]=\bm0$ and $\E\|\bm\zeta(\xi)\|^2\le\sigma^2$, so that $\E\|\bm\zeta_T\|^2=\frac1{T^2}\sum_{t<T}\E\|\bm\zeta(\xi_t)\|^2\le\sigma^2/T$, which proves the bound.

\ref{thmasy3}: This is \citet[Theorem 3]{mil_23}, applied in the Hilbert space $\Sm^\circ$ to the integrand $\bm U'\mapsto\Rcal((\bm E_{00},\bm 0)+\bm U',\xi)$, which is $\alpha$-strongly convex for $\Pa$-a.e.\ $\xi$ (and can be modified on a null set to be so for every $\xi$), and to the proper convex lower semicontinuous function ``$\Psi$'' of that reference equal to the indicator function of $\Dcal-(\bm E_{00},\bm 0)$. The sub-Gaussian assumption of that theorem on the gradient noise $\nabla\Rcal(\bm U^\ast,\xi)-\nabla\Ecal(\bm U^\ast)=\bm\zeta(\xi)$ is \eqref{assEM}. The second formulation follows by solving $2\exp(-\varsigma^{-2}T\varepsilon^2\alpha^2/3)=\delta$ for $\varepsilon$.

\subsection{Scope of Theorem~\ref{thmasy}}\label{app_scope}

We collect three statements on the scope of Theorem~\ref{thmasy} (Remark~\ref{remthmasy}) and prove the first two. (a) The bounds transfer to the conditional moments. For all $\bm U,\bm U'\in\Sm$ and data points $\xi$, $\|q_{\bm U}(\Zt,\Zt^\top)-q_{\bm U'}(\Zt,\Zt^\top)\|_F\le\big(\|\bm\Phi^\sy\|_2^2+\big(\sum_{i=1}^N\|\bm\phi^\id(z_i)\|_2^4\big)^{1/2}\big)\|\bm U-\bm U'\|$, so that, for a fixed cross section, parts \ref{thmasy1}--\ref{thmasy3} of Theorem~\ref{thmasy} carry over to the estimated conditional moment matrices \eqref{eqFop}, and thus to $\bm\mu_t$ and, with constants depending on $\|\bm U^\ast\|$, to $\bm\Sigma_t$ in \eqref{eqmutSigmatest}. For a random evaluation cross section the multiplier is random, and the transfer of parts \ref{thmasy2} and \ref{thmasy3} needs its integrability or tail control in addition. (b) Condition \eqref{assEM} implies \eqref{assSM} with $\sigma=\varsigma$, and it holds when $\bm\phi^\sy$, $\bm\phi^\id$, the individual returns and $N^2\|\Om\|_2^2$ are bounded, which for the Frobenius weight means a bounded cross-sectional size, and for the Gram-whitened weight in addition nonzero eigenvalues of the Gram matrices bounded away from zero. The theorem is thus a large-$T$ result for cross sections of bounded, but arbitrary and varying, size, and the constants are not dimension-free: $\sigma^2$ and $\alpha$ depend on the cross-sectional size through the weight and the feature matrices (Lemma~\ref{lemstrongconvex}). We do not pursue joint large-$N$, large-$T$ asymptotics. (c) Parts \ref{thmasy2} and \ref{thmasy3} assume independence over time for clarity. Under weak dependence with summable autocovariances of the gradient noise $\bm\zeta(\xi_t)$, the bound of part \ref{thmasy2} holds with $\sigma^2$ replaced by $\sum_{j\in\mathbb Z}|\E\langle\bm\zeta(\xi_0),\bm\zeta(\xi_j)\rangle|$. An exponential bound as in part \ref{thmasy3} needs quantitative dependence and tail conditions and is not pursued here.

(a): $\|\bm\Phi^\sy(\bm U^\sy-\bm U'^\sy)\bm\Phi^{\sy\top}\|_F\le\|\bm\Phi^\sy\|_2^2\|\bm U^\sy-\bm U'^\sy\|_F$, and $\|\Diag(\bm\Phi^\id(\bm U^\id-\bm U'^\id)\bm\Phi^{\id\top})\|_F^2=\sum_{i=1}^N(\bm\phi^\id(z_i)(\bm U^\id-\bm U'^\id)\bm\phi^\id(z_i)^\top)^2\le\sum_{i=1}^N\|\bm\phi^\id(z_i)\|_2^4\|\bm U^\id-\bm U'^\id\|_F^2$. Since the fitted matrices \eqref{eqFop} contain $\bm\mu_t$ as first column and $\bm\Sigma_t+\bm\mu_t\bm\mu_t^\top$ as lower block, the statements carry over to the estimated conditional moments, and for $\bm\Sigma_t$ via the map $(\bm\mu,\bm Q)\mapsto\bm Q-\bm\mu\bm\mu^\top$, which is Lipschitz on bounded sets, so that the constants depend on $\|\bm U^\ast\|$.

(b): Jensen's inequality gives $\exp(\varsigma^{-2}\E\|\bm\zeta(\xi)\|^2)\le\E\exp(\varsigma^{-2}\|\bm\zeta(\xi)\|^2)\le\exp(1)$, that is, \eqref{assSM} with $\sigma=\varsigma$. Under the stated boundedness conditions, $\|\bm\Phi^\sy\|_F^2+\|\bm\Phi^\id\|_F^2$ and $\|\X\|_2^2$ are bounded by constants times $N$, so by the bounds of Lemma~\ref{lemquadratic} $\|\Acal(\xi)\|$ and $\|\bm B(\xi)\|$ are bounded by constants times $N^2\|\Om\|_2^2$, which is bounded (using $\|\Om\|_2^2\le N^2\|\Om\|_2^2$ and $N\|\Om\|_2^2\le N^2\|\Om\|_2^2$ for $N\ge1$). Hence $\|\bm\zeta(\xi)\|\le(\|\Acal(\xi)\|+\|\Acal\|)\|\bm U^\ast\|+\|\bm B(\xi)\|+\|\bm B\|$ is bounded by a constant $\varsigma$, $\Pa$-a.s., and \eqref{assEM} holds for this $\varsigma$.

\subsection{Proof of Proposition~\ref{proptwostage}}

\ref{proptwostage1}: The ratio $\hat u=\bar e/\bar d$ is continuous where $\bar d>0$, which holds since $\mathrm{df}_t>0$. $\hat{\bm F}$, $\bar{\bm A}$ and $\bar{\bm f}\bar{\bm f}^\top$ are continuous, and $\Pcal_+$ is the projection onto a closed convex set, hence 1-Lipschitz in the Frobenius norm. The ergodic theorem gives the almost sure convergence of the averages, and continuity carries it over. \ref{proptwostage2}: Under the model, $\bm f_{t+1}=\bm G_t^{-1}\bm\phi^{\sy\top}_t\bm x_{t+1}$ has $\E[\bm f_{t+1}\mid\bm z_t]=\muf_0$ and $\E[\bm f_{t+1}\bm f_{t+1}^\top\mid\bm z_t]=\bm V_0+u^\id_0\bm G_t^{-1}$, and the residual sum of squares has conditional expectation $u^\id_0\mathrm{df}_t$. The limits of $\bar{\bm f}$, $\hat u$ and $\hat{\bm F}$ are therefore $\muf_0$, $u^\id_0$ and $\bm F_\infty$, and that of $\bar{\bm A}-\bar{\bm f}\bar{\bm f}^\top-\eta\hat{\bm F}$ is $\bm C_0-\eta\bm F_\infty$. If $\lsy=0$ and $\bm C_0\succeq\eta\bm F_\infty$, then $\Pcal_+(\bm C_0-\eta\bm F_\infty)=\bm C_0-\eta\bm F_\infty$. \ref{proptwostage3}: Products and the ratio that defines $\hat u$ are locally Lipschitz, and $\Pcal_+$ is 1-Lipschitz, so the CMK estimator is a locally Lipschitz function of the averages at their expectations. By the central limit theorem the averages deviate from their expectations by $O_P(T^{-1/2})$.

\subsection{Proof of Proposition~\ref{propGLS}}

Throughout, \eqref{eqassid} holds, and we drop the time index and write $\bm\phi\coloneqq\bm\phi^\sy_t$, $\bm R\coloneqq\bm\phi^\top\bm\Sigma^{\id+}\bm\phi$, $\bm W\coloneqq\bm W_t$, $\bm\mu,\bm\Sigma,\bm Q$, and $\bm\Pi\coloneqq\bm\phi^+\bm\phi$ for the orthogonal projection onto $\Ima\bm\phi^\top$, and $\bm\Sigma^\id\coloneqq\bm\Sigma^\id_t$. The following lemma collects the explicit formulas behind Proposition~\ref{propGLS}.

\begin{lemma}[Moments of the factor representation]\label{lemfactordetails}
\begin{enumerate}
\item\label{lfd1} $\bm R^+\bm R=\bm\Pi$, and weighting by $\bm\Sigma^\id$, $\bm\Sigma$ or $\bm Q$ gives the same portfolios,
$\bm W^\top=(\bm\phi^\top\bm\Sigma^{+}\bm\phi)^+\bm\phi^\top\bm\Sigma^{+}=(\bm\phi^\top\bm Q^{+}\bm\phi)^+\bm\phi^\top\bm Q^{+}$.
\item\label{lfd2} $\E_t[\bm f_{t+1}]=\bm\Pi\muf$, $\cov_t[\bm f_{t+1}]=\bm\Pi\bm C\bm\Pi+\bm R^+$, $\E_t[\bm\epsilon_{t+1}]=\bm 0$, $\cov_t[\bm f_{t+1},\bm\epsilon_{t+1}]=\bm 0$, and $\cov_t[\bm\epsilon_{t+1}]=\bm\Sigma^\id-\bm\phi\bm R^+\bm\phi^\top$ has rank $\rank\bm\Sigma^\id-\rank\bm\phi$. In the isotropic case $\bm\Sigma^\id=u^\id\bm I_{N_t}$, $u^\id>0$: $\bm f_{t+1}=\bm\phi^+\bm x_{t+1}$, $\cov_t[\bm f_{t+1}]=\bm\Pi\bm C\bm\Pi+u^\id(\bm\phi^\top\bm\phi)^+$ and $\cov_t[\bm\epsilon_{t+1}]=u^\id(\bm I_{N_t}-\bm\phi\bm\phi^+)$.
\end{enumerate}
\end{lemma}

We prove Proposition~\ref{propGLS} together with the lemma.

Law of one price. From \eqref{eqmutSigmatest} and $\Ima(\bm A+\bm B)=\Ima\bm A+\Ima\bm B$ for $\bm A,\bm B\in\Sy^n_+$ we obtain $\Ima\bm\Sigma^\id\subseteq\Ima\bm\Sigma\subseteq\Ima\bm Q$. Together with \eqref{eqassid} and $\bm\mu=\bm\phi\muf\in\Ima\bm\phi$ this gives $\spn\{\bm\mu\}\subseteq\Ima\bm\phi\subseteq\Ima\bm\Sigma\subseteq\Ima\bm Q$, and in particular the law of one price $\bm\mu\in\Ima\bm\Sigma$.

Regression loadings and Lemma~\ref{lemfactordetails}\ref{lfd1}. By \eqref{eqmutSigmatest} and \eqref{eqassid}, the pair $(\bm V,\bm\Sigma^\id)$ is a regular decomposition of $\bm Q$ with respect to $\bm\phi$ in the sense of \citet[Section~4.3]{filipovicschneider26fund}, that is, $\bm Q=\bm\phi\bm V\bm\phi^\top+\bm\Sigma^\id$ with $\bm V,\bm\Sigma^\id\succeq0$ and $\Ima\bm\phi\subseteq\Ima\bm\Sigma^\id$. Their Theorem~3.1 states that, under the support condition $\Ima\bm\phi\subseteq\Ima\bm Q$ established above, the traded factors with $\Ima\bm W\subseteq\Ima\bm Q$ whose regression loadings are $\bm\phi$ are exactly those with weight matrix $\bm W^\top=(\bm\phi^\top\bm Q^+\bm\phi)^+\bm\phi^\top\bm Q^+$, and their Proposition~4.3 identifies this weight matrix with $\bm R^+\bm\phi^\top\bm\Sigma^{\id+}$, which is \eqref{GLSfac}. Our weight matrix indeed satisfies $\Ima\bm W\subseteq\Ima(\bm\Sigma^{\id+}\bm\phi)\subseteq\Ima\bm\Sigma^\id\subseteq\Ima\bm Q$. The remaining identity in Lemma~\ref{lemfactordetails}\ref{lfd1} is their Corollary~4.4 (with Lemma~E.8 applied to $\bm\Sigma=\bm\phi\bm C\bm\phi^\top+\bm\Sigma^\id$ for the $\bm\Sigma^\id$-weighting), whose hypotheses, the law of one price and $\spn\{\bm\mu\}\subseteq\Ima\bm\phi\subseteq\Ima\bm Q$, hold by the chain of inclusions above. Since $\bm R=\bm A^\top\bm A$ for $\bm A\coloneqq(\bm\Sigma^{\id+})^{1/2}\bm\phi$, and since $\bm\phi^\top=\bm A^\top(\bm\Sigma^\id)^{1/2}$ and $\bm A^\top=\bm\phi^\top(\bm\Sigma^{\id+})^{1/2}$ by \eqref{eqassid}, we have $\Ima\bm R=\Ima\bm A^\top=\Ima\bm\phi^\top$ and hence $\bm R^+\bm R=\bm\phi^+\bm\phi=\bm\Pi$.

Lemma~\ref{lemfactordetails}\ref{lfd2}. Lemma~E.10 of \citet{filipovicschneider26fund} gives
\[ \E_t[\bm f_{t+1}\bm f_{t+1}^\top]=\bm\Pi\,\bm V\,\bm\Pi+\bm R^+,\qquad
   \E_t[\bm\epsilon_{t+1}\bm\epsilon_{t+1}^\top]=\bm\Sigma^\id-\bm\phi\bm R^+\bm\phi^\top ,\]
and $\E_t[\bm\epsilon_{t+1}\bm f_{t+1}^\top]=(\bm I_{N_t}-\bm\phi\bm W^\top)\bm Q\bm W=(\bm I_{N_t}-\bm\phi\bm W^\top)\bm\Sigma^\id\bm W=\bm 0$, because $\bm W^\top\bm\phi=\bm R^+\bm R=\bm\Pi$ and $\bm\Sigma^\id\bm\Sigma^{\id+}\bm\phi=\bm\phi$ by \eqref{eqassid}.
Therefore $\E_t[\bm f_{t+1}]=\bm W^\top\bm\phi\muf=\bm R^+\bm R\muf=\bm\Pi\muf$ and $\cov_t[\bm f_{t+1}]=\bm\Pi(\bm V-\muf\muf^\top)\bm\Pi+\bm R^+=\bm\Pi\bm C\bm\Pi+\bm R^+$. The factor component operator $\bm\phi\bm W^\top$ is a projection onto $\Ima\bm\phi$, so $\bm\mu\in\Ima\bm\phi$ gives $\E_t[\bm\epsilon_{t+1}]=(\bm I_{N_t}-\bm\phi\bm W^\top)\bm\mu=\bm 0$. Consequently $\cov_t[\bm f_{t+1},\bm\epsilon_{t+1}]=\E_t[\bm f_{t+1}\bm\epsilon_{t+1}^\top]=\bm 0$ and $\cov_t[\bm\epsilon_{t+1}]=\E_t[\bm\epsilon_{t+1}\bm\epsilon_{t+1}^\top]$. For the rank, $\bm\phi=(\bm\Sigma^\id)^{1/2}\bm A$ by \eqref{eqassid}, so $\bm\phi\bm R^+\bm\phi^\top=(\bm\Sigma^\id)^{1/2}\bm A(\bm A^\top\bm A)^+\bm A^\top(\bm\Sigma^\id)^{1/2}=(\bm\Sigma^\id)^{1/2}\bm A\bm A^+(\bm\Sigma^\id)^{1/2}$ and $\cov_t[\bm\epsilon_{t+1}]=(\bm\Sigma^\id)^{1/2}(\bm P_{\bm\Sigma^\id}-\bm A\bm A^+)(\bm\Sigma^\id)^{1/2}$, with $\bm P_{\bm\Sigma^\id}$ the orthogonal projection onto $\Ima\bm\Sigma^\id$. As $\Ima\bm A\subseteq\Ima\bm\Sigma^\id$, the matrix $\bm P_{\bm\Sigma^\id}-\bm A\bm A^+$ is the orthogonal projection onto the orthogonal complement of $\Ima\bm A$ in $\Ima\bm\Sigma^\id$, of rank $\rank\bm\Sigma^\id-\rank\bm A=\rank\bm\Sigma^\id-\rank\bm\phi$, and $(\bm\Sigma^\id)^{1/2}$ maps $\Ima\bm\Sigma^\id$ bijectively onto itself, which gives the asserted rank. The isotropic formulas follow by inserting $\bm\Sigma^\id=u^\id\bm I_{N_t}$, in which case $\bm W^\top=(\bm\phi^\top\bm\phi)^+\bm\phi^\top=\bm\phi^+$ \citep[Proposition~4.5]{filipovicschneider26fund}.

Sharpe ratio in part \ref{propGLS1}. The law of one price holds, and by Lemma~\ref{lemfactordetails}\ref{lfd2} the factors and residuals are uncorrelated with $\E_t[\bm\epsilon_{t+1}]=\bm 0$. Proposition~3.5 of \citet{filipovicschneider26fund} therefore applies: factor spanning holds, and by their Proposition~2.4 the factor portfolios attain the maximal conditional Sharpe ratio of the returns. The law of one price holds for the factors as well, since their maximal conditional Sharpe ratio is bounded by that of the returns (their Proposition~C.1). Their formula for the maximal conditional Sharpe ratio (Proposition~C.4), applied to the returns and to the factors, then gives $\theta_t=\E_t[\bm f_{t+1}]^\top\cov_t[\bm f_{t+1}]^+\E_t[\bm f_{t+1}]$.

Part \ref{propGLS3}. Since $\E_t[\bm x_{t+1}(1-\bm a^\top\bm x_{t+1})]=\bm\mu-\bm Q\bm a$ vanishes for $\bm a=\bm Q^+\bm\mu$, as $\bm\mu\in\Ima\bm Q$, $x^\ast_{t+1}$ is the conditional least squares projection of $1$ on the span of the returns. Since the law of one price holds, the conditional asset pricing results in \citet[Appendix~C]{filipovicschneider26fund} apply, with $x^\ast_{t+1}$ their projection. Their Corollary~C.3 gives $\bm Q^+\bm\mu=\bm\Sigma^+\bm\mu/(1+\theta_t)$, and their Propositions~C.1 and~C.4 yield the conditional moments of $x^\ast_{t+1}$ and the CMVE property of the portfolios with weights $\bm\Sigma^+\bm\mu$ and $\bm Q^+\bm\mu$.

Part \ref{propSDF3}. Fix $t$. By part \ref{propGLS3}, $\E_t[(x^\ast_{t+1})^2]=\theta_t/(1+\theta_t)\le1$, so $\E[(x^\ast_{t+1})^2]\le1$ by the tower property. Let $y_{t+1}=\bm w_t^\top\bm x_{t+1}$ be square integrable. Then $x^\ast_{t+1}y_{t+1}$ is integrable, and since $\bm\mu_t\in\Ima\bm Q_t$ gives $\bm Q_t\bm Q_t^+\bm\mu_t=\bm\mu_t$, $\E[x^\ast_{t+1}y_{t+1}]=\E\big[\bm\mu_t^\top\bm Q_t^+\E_t[\bm x_{t+1}\bm x_{t+1}^\top]\bm w_t\big]=\E[\bm\mu_t^\top\bm Q_t^+\bm Q_t\bm w_t]=\E[\bm\mu_t^\top\bm w_t]=\E[y_{t+1}]$. Applied to $y=x^\ast$ this gives $\E[(x^\ast_{t+1})^2]=\E[x^\ast_{t+1}]=\E[\theta_t/(1+\theta_t)]=\bar\theta$. Stationarity of $\bm z_t$ makes $\bar\theta$ independent of $t$, which is its only role, and $\bar\theta<1$ because $\theta_t/(1+\theta_t)<1$. By the Cauchy--Schwarz inequality,
\[ \E[y_{t+1}]^2=\E[x^\ast_{t+1}y_{t+1}]^2\le\E[(x^\ast_{t+1})^2]\,\E[y_{t+1}^2]=\bar\theta\big(\var[y_{t+1}]+\E[y_{t+1}]^2\big),\]
which rearranges to \eqref{eqSRuncond}. For $y=x^\ast$ we have $\E[x^\ast_{t+1}]^2=\bar\theta^2$ and $\var[x^\ast_{t+1}]=\bar\theta-\bar\theta^2$, so that equality holds.

\section{Solver}\label{app_vec}

We compute the one-step estimator \eqref{krr1floor}, and the solution of the CMK program \eqref{krr1convex} with any weight, for the isotropic specification by an accelerated projected gradient method \citep{beckteboulle09} in preconditioned coordinates. It works directly with the matrix parameter $\bm U$ and the sufficient statistics of Lemma~\ref{lemsuffstat}. No vectorization is needed.\footnote{An interior-point solver, which we use as a cross-check, requires the half-vectorized form of the quadratic \eqref{eqRquadratic}, obtained with the duplication matrix in the standard way \citep[Chapter 3]{magnusneudecker19}.} 

\paragraph{Preconditioning.} Let $\bar{\bm M}\coloneqq\frac1T\sum_t \bm M_t$, assumed nonsingular, and let $\bm\Psi\coloneqq\bar{\bm M}_{\rm bd}^{1/2}$ be the square root of its block-diagonal part $\bar{\bm M}_{\rm bd}\coloneqq\Diag(\bar M_{00},\bar{\bm M}_{11})$, with $\bar{\bm M}_{11}$ the feature block. The variable $\tilde{\bm U}^\sy\coloneqq\bm\Psi\bm U^\sy\bm\Psi$ has the Hessian block $\frac2T\sum_t \tilde{\bm M}_t\otimes\tilde{\bm M}_t$ with $\tilde{\bm M}_t\coloneqq\bm\Psi^{-1}\bm M_t\bm\Psi^{-1}$. The average of the $\tilde{\bm M}_t$ is the identity when $\bar{\bm M}$ is block diagonal, as for the Frobenius and the Gram-whitened weight. The Hessian block is then close to $2\bm I$ when the $\tilde{\bm M}_t$ vary little over the window. For the loss \eqref{eqlossrho} the role of $\bm M_t$ is played by $\Diag(\kappa,\bm I_{m^\sy})$ in every month, and the Hessian block is exactly $2\bm I$. For the Browne weight the off-diagonal blocks enter the gradient but not the preconditioner. The scalar $u^\id$ is rescaled to the same curvature, the constraint $\bm U^\sy_{00}=1$ becomes $\tilde{\bm U}^\sy_{00}=\bar M_{00}$, the trace penalty becomes $\lsy\langle\bm\Psi^{-2},\tilde{\bm U}^\sy\rangle$, the floor becomes $\tilde{\bm F}_\eta\coloneqq\bm\Psi\bm F_\eta\bm\Psi$, and the PSD property is preserved.

\paragraph{Projection.} The projection of a symmetric matrix $\bm A$ onto $\{\tilde{\bm U}:\tilde{\bm U}-\tilde{\bm F}_\eta\succeq0,\ \tilde U_{00}=\bar M_{00}\}$ is $\tilde{\bm F}_\eta+\Pcal_+(\bm A-\tilde{\bm F}_\eta-\nu\bm E_{00})$, with $\Pcal_+$ the projection onto the PSD cone (eigenvalue truncation) and $\nu$ the root of the monotone scalar equation $[\Pcal_+(\bm A-\tilde{\bm F}_\eta-\nu\bm E_{00})]_{00}=\bar M_{00}$, found by a safeguarded secant iteration. The matrix projection is exact up to the tolerance of the root $\nu$, at the cost of a few eigendecompositions of an $(m^\sy+1)\times(m^\sy+1)$ matrix per step, and it is what produces the exact zero eigenvalues of the excess in Proposition~\ref{proprank}. The scalar $u^\id$ is projected on $[0,\infty)$.

\paragraph{Iteration.} Gradient steps of length $1/\bar L$, $\bar L$ the largest eigenvalue of the joint Hessian in $(\tilde{\bm U}^\sy,u^\id)$, including the cross block that couples the two, computed by power iteration, with Nesterov momentum and the adaptive restart of \citet{odonoghuecandes15}. The path over the $\lsy$ grid is warm-started from the previous solution, and the stopping rule is a relative fixed-point residual of $10^{-7}$.

\paragraph{Statistics.} All gradients are formed from the six per-month statistics of \eqref{eqsuffstatiso}, those of Lemma~\ref{lemsuffstat} for a weighted loss and those of Section~\ref{sec_tikhonov} for the loss \eqref{eqlossrho}, so that a step costs $O(T(m^\sy)^3)$ and no quantity of the size of a cross section is ever formed. On the panel of Section~\ref{sec_empiricsNEW} the method needs 10--20 iterations per value of $\lsy$ and agrees with the interior-point solver Clarabel \citep{goulartchen24} to $10^{-10}$ in the objective. The implementation is \texttt{mk/fista.py} in the replication package.

\clearpage
\counterwithin{table}{section}
\counterwithin{figure}{section}
\counterwithin{equation}{section}
\section*{Online Appendix to ``Conditional Moment Kernels: Joint Modeling of Expected Returns and Covariances in Unbalanced Panels''}
\addcontentsline{toc}{section}{Online Appendix}

\section{Simulation Study}\label{app_sim}

We simulate a data-generating process (DGP) whose population moment kernel is known. We check the convergence rate of Theorem~\ref{thmasy}, the selection of the number of factors of Proposition~\ref{proprank}, and the out-of-sample protocol of Section~\ref{sec_empiricsNEW}, also under misspecification. All results are produced by the script \texttt{simulation.py} of the replication package, with the solver of Appendix~\ref{app_vec}.

\paragraph{DGP.} Each month has $N=500$ stocks with covariates $\bm z_i\sim N(\bm 0,\diag(1,\frac12,\dots,\frac1{20}))$, drawn independently across stocks and months, and the cosine kernel, so that $\bm\phi^\sy(z)=\tilde z\coloneqq z/\|z\|_2$ and $m^\sy=20$ is the correct feature dimension. The decaying variances make the leading eigenvectors of the Gram matrix the informative directions. For $m^\sy<20$ the features are the projections on the $m^\sy$ dominant eigenvectors of the training-window Gram matrix, as for the linear kernel in Section~\ref{sec_dataprotocol}. Returns follow a three-factor model with a pricing error,
\begin{equation}\label{eqpopFM}
\begin{gathered}
  \bm x_{t+1}=\bm\phi^\sy_t\bm b^{\rm pop}+\bm\phi^\sy_t\bm\beta^{\rm pop}\bm g_{t+1}+\bm\epsilon_{t+1},\qquad \bm g_{t+1}\sim N(\bm 0,\bm I_3),\\\qquad \bm\epsilon_{t+1}\sim N(\bm 0,u^{\rm pop}\bm I_N),
\end{gathered}
\end{equation}
with $\bm\phi^\sy_t=\bm\phi^\sy(\bm z_t)$ as in \eqref{eqPhit}, loadings $\bm\beta^{\rm pop}\in\R^{20\times3}$ with independent $N(0,0.04^2)$ entries, and $\bm b^{\rm pop}$ equal to $\bm\beta^{\rm pop}$ times the monthly factor Sharpe ratios $(0.2,0.15,0.1)$ plus a pricing error with independent $N(0,0.001^2)$ entries, and $u^{\rm pop}=0.01$ (a monthly idiosyncratic volatility of 10\%). The population parameter, the true parameter $\bm U_0$ of Section~\ref{sec_population}, is $\bm U^{\rm pop}=(\bm U^{\sy,\rm pop},u^{\rm pop})$ with $\bm U^{\sy,\rm pop}=\begin{bsmallmatrix}1&\bm b^{{\rm pop}\top}\\ \bm b^{\rm pop}&\bm\beta^{\rm pop}\bm\beta^{{\rm pop}\top}+\bm b^{\rm pop}\bm b^{{\rm pop}\top}\end{bsmallmatrix}$, so that $\bm C^{\rm pop}=\bm\beta^{\rm pop}\bm\beta^{{\rm pop}\top}$ has rank three, and the population conditional moments are $\bm\mu_t=\bm\phi^\sy_t\bm b^{\rm pop}$ and $\bm\Sigma_t=\bm\phi^\sy_t\bm C^{\rm pop}\bm\phi^{\sy\top}_t+u^{\rm pop}\bm I_N$. The calibration mimics the orders of magnitude of Section~\ref{sec_empiricsNEW}: the population moments have a predictive $R^2$ of 1.3\%, an $R^{2,2}$ of 2.7\%, a systematic share of variance of 25\%, and an annualized maximum conditional Sharpe ratio of $\sqrt{12\,\bm\mu_t^\top\bm\Sigma_t^{-1}\bm\mu_t}\approx1.1$. The misspecified variant appends to $\tilde z$ three nonlinear features $\psi_k(\tilde z)$, $k=1,2,3$, proportional to $\tilde z_k\tilde z_{k+1}$ and scaled to half the root mean square of $\tilde z_k$, with loadings of the same scale as the linear ones. They carry 17\% of the systematic variance and are invisible to the CMK program, which only sees $\tilde z$.

\paragraph{Convergence rate.} We estimate $\bm U$ from $T$ months with $m^\sy=20$, $\lsy=\lid=0$ and $\eta=0$, with the Frobenius weight of Section~\ref{sec_loss} and with the Gram-whitened weight of Section~\ref{sec_weights}, for $T=50,100,\dots,3200$ and 100 independent panels per $T$, and record the squared distance $\|\bm U^\ast_T-\bm U^{\rm pop}\|^2=\|\bm U^{\sy\ast}_T-\bm U^{\sy,\rm pop}\|_F^2+(u^\ast_T-u^{\rm pop})^2$ in the norm of Section~\ref{sec_uniqueness}. Figure~\ref{fig_sim}(a) shows the mean squared error and its 95\% quantile. The fitted log-log slopes are between $-0.99$ and $-1.00$ for both weights and both statistics, against the rate $-1$ of the bound in Theorem~\ref{thmasy}\ref{thmasy2} and the $T^{-1/2}$ rate for the distance in Theorem~\ref{thmasy}\ref{thmasy3}. The solutions under the two weights are close. The solution with the Gram-whitened weight is more precise by 3 to 5\% at every $T$. Neither shows a bias, and both converge to $\bm U^{\rm pop}$. In absolute terms, the squared distance of $\bm U^{\rm pop}$ from the benchmark $(\bm E_{00},0)$ (zero mean, zero covariance) is $0.0085$, of which the mean squared error is 47\% at $T=50$, 23\% at $T=100$, close to the 108 months of the estimation windows of Section~\ref{sec_empiricsNEW}, and 0.7\% at $T=3200$. The bound $\sigma^2/(\alpha^2T)$ itself, evaluated with the constants of the DGP for the Frobenius weight,\footnote{$\sigma^2=\E\|\bm\zeta(\xi)\|^2=1.6\cdot10^4$ and $\alpha=11.8$ are computed from the sufficient statistics of 2,000 simulated months as the mean of $\|\nabla\Rcal(\bm U^{\rm pop},\xi_t)\|^2$ (here $\nabla\Ecal(\bm U^{\rm pop})=\bm 0$, since the model is well specified) and the minimum over $t$ of the smallest eigenvalue of $\Acal(\xi_t)$ on $\Sm^\circ$, the strong convexity constant of Section~\ref{sec_uniqueness}.} is the dashed line of Figure~\ref{fig_sim}(a). It is a plug-in evaluation with simulated constants, not the bound of the theorem, whose almost-sure strong convexity a Gaussian design does not satisfy uniformly. It has the right slope but exceeds the realized mean squared error by a factor of about 600, the usual price of a bound that only uses strong convexity and the variance of the gradient.

\paragraph{Number of factors along the path.} Figure~\ref{fig_sim}(b) and (c) show, for one panel of $T=108$ months and one of $T=1000$ months, the number of factors and the squared distance to $\bm U^{\rm pop}$ along the regularization path $\lsy\in[10^{-6},10^{-1}]$, for the Frobenius weight, for the Gram-whitened weight without the floor ($\eta=0$) and with the floor at $\eta=1$. The number of factors is the rank of the estimated $\bm C$ for $\eta=0$ and the rank of the excess above the floor for $\eta=1$. Without the floor the count decreases from its unpenalized value, $11$ for the Frobenius weight and $10$ for the Gram-whitened weight, through a plateau at the true rank three for $\lsy$ roughly between $10^{-3}$ and $10^{-2}$, to zero, as in Proposition~\ref{proprank}. With the floor the count is three for every $\lsy$ up to $10^{-2}$, at both $T$. In the 17 directions without systematic variance the estimated excess is noise of the size of $\hat{\bm F}$, and the floor removes it. The distance to $\bm U^{\rm pop}$ is flat up to $\lsy\approx10^{-3}$ and increases beyond. At $T=108$ the best $\lsy$ on the grid improves on $\lsy=0$ by less than $0.2$\%, and at $T=1000$ by 1\% (Frobenius) and 2\% (Gram-whitened). The floor costs $0.3$\% of distance at $T=108$ and $3$\% at $T=1000$, where the floor $\hat{\bm F}$ that it adds is no longer negligible against the estimation error. On these two panels the regularizers identify the number of factors at almost no cost in estimation error. In this design, with $N=500$ stocks and $m^\sy=20$ features, they are not needed for the fit.

\begin{figure}
\centering
\includegraphics[width=\textwidth]{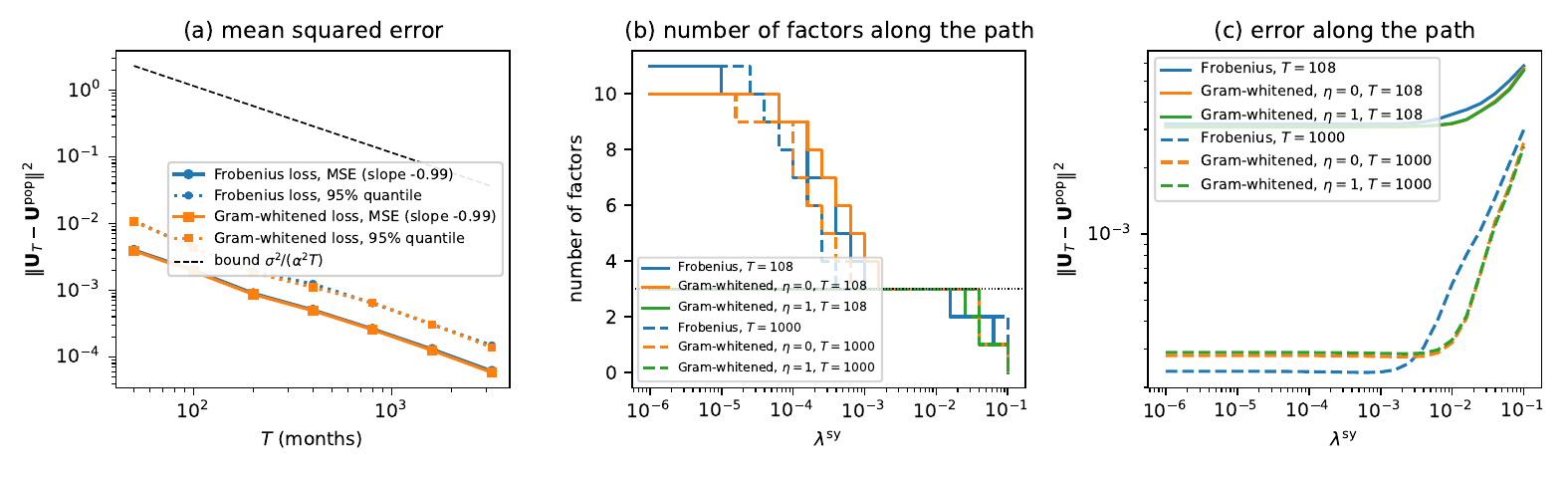}
\caption{\label{fig_sim}Simulation study. (a) Mean squared error $\E\|\bm U^\ast_T-\bm U^{\rm pop}\|^2$ and its 95\% quantile over 100 simulated panels per $T$ for the Frobenius weight and the Gram-whitened weight ($m^\sy=20$, $\lsy=\lid=0$, $\eta=0$), with the fitted log-log slopes, and a plug-in evaluation of $\sigma^2/(\alpha^2T)$ with simulated constants. (b) Number of factors and (c) squared distance to $\bm U^{\rm pop}$ along the regularization path $\lsy$, for the Frobenius weight and for the Gram-whitened weight with $\eta=0$ and $\eta=1$, for one panel of $T=108$ and one of $T=1000$ months. The dotted line in (b) is the true rank three.}
\end{figure}

\paragraph{Out-of-sample protocol.} We run the protocol of Section~\ref{sec_dataprotocol} on one simulated panel of 720 months per DGP. Each window has 96 training and 12 validation months. The features are those of the training window with $m^\sy=5,10,20$, and $(\eta,\lsy)$ is selected on the validation months on the grids of Section~\ref{sec_dataprotocol}. The model is re-estimated on the 108 months and evaluated on the following month, which gives 612 test months. On them we compute the measures of Section~\ref{sec_predictiveNEW} and the Sharpe ratios of Section~\ref{sec_assetpricing_empirical}, for the estimated moments and for the population moments of \eqref{eqpopFM}. Table~\ref{tab_sim} reports the results. Under the well-specified DGP the model with $m^\sy=20$ attains 99\% of the score differential of the population moments and matches their $R^{2,2}$ and total $R^2$ to within a tenth of a percentage point. The validation selects a positive penalty in 97\% of the windows and a positive floor in 3\%, and the estimated number of factors is $4.3$ on average, against the true three. Estimation costs show in the mean, with an $R^2$ of $1.13$\% against $1.31$\%, and hence in the realized Sharpe ratios, $0.75$ against $1.27$. The predicted Sharpe ratio of the estimated model, $1.87$, overstates the realized one for the reason given in Section~\ref{sec_assetpricing_empirical}. The smaller models lose in every measure but the Sharpe ratio, in the same order as in Table~\ref{tab_stat}. Under the misspecified DGP the estimated models lose a fifth of the population score differential (124 against 157 for $m^\sy=20$). The loss is the systematic variance of the nonlinear features, which a linear feature map cannot capture. The ranking across $m^\sy$ is unchanged. The gap is the simulated analog of the gain from the nonlinear features in Table~\ref{tab_stat}.

\begin{table}
\begin{center}
\resizebox{\textwidth}{!}{\input{tables/rev11_tab6_sim.tex}}
\end{center}
\caption{\label{tab_sim}Out-of-sample protocol on simulated data. Measures of Section~\ref{sec_predictiveNEW} (in percent where indicated), average number of factors, and annualized predicted and realized Sharpe ratios of the UMVE and CMVE portfolios over 612 test months, for the solution of the CMK program with the Gram-whitened weight, with $m^\sy=5,10,20$ ($(\eta,\lsy)$ by validation, $\lid=0$, re-estimation on 108 months) and for the population moments of \eqref{eqpopFM}, under the well-specified DGP and under the misspecified DGP with three nonlinear features.}
\end{table}

\section{Additional Empirical Results}\label{app_emp}

We give details of the estimation, the selection and the measures, the sub-period results (Table~\ref{tab_periods}), the robustness checks of Section~\ref{sec_robustness} (Table~\ref{tab_robust}), the comparison of the one-step and the CMK estimator, a characteristic-driven idiosyncratic specification and the sensitivity of the $t$-statistics to the Newey--West lag. Section~\ref{app_jkp} repeats the study of Section~\ref{sec_empiricsNEW} on a second panel, the characteristics of \citet{jensenkellypedersen23} with returns to December 2025.

\subsection{Main Panel}\label{app_emp_main}

\paragraph{Estimation details.} The excess returns use the one-month Treasury bill yield of CRSP at the end of month $t$, divided by 12. It agrees with the riskless rate of \citet{french} (correlation $0.995$ over 1963--2021, mean absolute difference below 2 basis points per month). The median heuristic $\ell_0^2=36.4$ is a squared distance of $0.39$ per characteristic, whose ranks lie in $(-1,1)$. The 200 landmarks capture 77\% of the trace of the Gaussian kernel matrix at $\ell=\ell_0$ and 17\% at $\ell=\ell_0/2$. In the windows up to 1985 the pooled Gram matrix of $[\bm 1,\bm z_t]$ has rank 79 to 94. In the linear feature space at most one eigenvalue of $\bm G_t$ per month lies below $\varrho_t$ from 1985 on. In the full additive feature space about 160 of the 295 eigenvalues lie below it, since the Gaussian features are strongly correlated across stocks, but together they carry $0.02$\% of the trace. The pooled residual variance counts $\mathrm{df}_t=N_t-m$ residual degrees of freedom per month, which is $N_t-\rank\bm\phi^\sy_t$ from 1985 on. The one-step estimator uses $\kappa=100$. For large $\kappa$ its estimate of $\muf$ is close to $\bar{\bm f}$ (Proposition~\ref{propclosedform}). IPCA minimizes its objective in profile form. The factor step, the cross-sectional regressions given $\bm\Gamma_\beta$, is solved in closed form, and $\bm\Gamma_\beta$ is found by L-BFGS, started from the estimate of the previous window and from the leading singular vectors of the characteristic-managed portfolios, keeping the better of the two solutions. The factor step carries a Tikhonov regularization as in \eqref{eqtikhonov}, at $10^{-4}$ times the average eigenvalue of its Gram matrix, which is needed in the windows of the 1970s in which some instruments are identically zero. IPCA has no hyper-parameter besides $K$.

\paragraph{Selection and measures.} The score \eqref{eq_scoringloss} is computed without forming $\bm\Sigma_t^{-1}$. With $\bm\beta_t\coloneqq\bm\phi^\sy_t\bm C^{1/2}$, $\bm\Sigma_t^{-1}=(u^\id)^{-1}[\bm I_{N_t}-\bm\beta_t(u^\id\bm I_m+\bm\beta_t^\top\bm\beta_t)^{-1}\bm\beta_t^\top]$ and $\det\bm\Sigma_t=(u^\id)^{N_t-m}\det(u^\id\bm I_m+\bm\beta_t^\top\bm\beta_t)$. For the additive kernel the penalty is at the upper end of its grid, $6\cdot10^{-3}$, in most windows. For the selected additive kernel of Table~\ref{tab_stat}, the validation picks the full feature space at $\ell=\ell_0$ in 54\% of the windows, $m=152$ at $\ell=\ell_0$ in 30\%, and the linear kernel in 2\%. Extending the grid to $\lsy=0.2$ for the additive kernel at $m=295$, on every fourth window, moves the selection to one of the added values in about 80\% of the windows. The number of factors falls from 51 to 34 at $\ell=\ell_0$ and from 49 to 33 at $\ell=\ell_0/2$. The score differential changes by less than 3 points ($t=1.4$ and $1.6$) and the Sharpe ratio by at most $0.1$. The counts of the additive kernel in Table~\ref{tab_stat} are upper bounds in this sense. For $R^{2,2}$ in \eqref{eq_kellyoos} the squared Frobenius norm of a month grows like $N_t^2$, so even after the factor $N_t^{-1}$ a month counts in proportion to $N_t$. With the weight $N_t^{-2}$ instead, $R^{2,2}$ changes by at most $0.03$ percentage points in Table~\ref{tab_stat}, and only rows within $0.03$ of each other change places.

\begin{table}
\begin{center}
\resizebox{\textwidth}{!}{\input{tables/rev11_tab5_periods.tex}}
\end{center}
\caption{\label{tab_periods}Sub-periods. Predictive $R^2$ and $R^{2,2}$, score differential and annualized Sharpe ratios of the UMVE and the equally weighted portfolio in the full sample, in three sub-periods and in the 57 crash months of Figure~\ref{fig_rolling}, for the selected additive-kernel specification, the linear kernel with $m=95$ and IPCA with $K=20$ and $K=60$. The last two columns give the average monthly score difference to the selected additive-kernel specification and its Newey--West $t$-statistic.}
\end{table}

\begin{table}
\begin{center}
\resizebox{\textwidth}{!}{\input{tables/rev12_tab4_robust.tex}}
\end{center}
\caption{\label{tab_robust}Robustness. Out-of-sample measures of Table~\ref{tab_stat}, annualized Sharpe ratios, and the average monthly score difference to the reference row of the block with its Newey--West $t$-statistic and the $p$-value of the \citet{ledoitwolf08} test of equal UMVE Sharpe ratios. First block, in three panels: the regularizers, for the linear kernel with $m=95$ and the additive kernel with $m=295$ at both bandwidths, with the CMK estimator of Table~\ref{tab_stat} as reference, the CMK estimator without regularizers, the floor alone, the penalty alone, the moment step that adds $\hat{\bm F}$ to $\bm V$ in every direction instead of flooring the excess, and the identity-shaped floor $\eta(\tr\hat{\bm F}/m)\bm I_m$ of Section~\ref{sec_regularization}, with $(\eta,\lsy)$ selected and at $\eta=1$, $\lsy=0$ (for the additive kernel at $\ell=\ell_0$). Second block: the weight and the shape of the floor together, for the linear kernel with $m=95$ on every fourth window, the January, May and September test months, the Gram-whitened weight and the scalar weight $N_t^{-1/2}\bm I_{N_t+1}$ with the floor $\eta\hat{\bm F}$ and with the identity-shaped floor, $(\eta,\lsy)$ selected, and the scalar weight without regularizers, computed by the solver of Appendix~\ref{app_vec} (its rows carry no count of the excess). Third block: the weight. At $m=95$ the Browne weight against the Gram-whitened weight, both without regularizers. At $m=40$ on every second window, the Gram-whitened, the Frobenius (first-moment weight $\kappa$) and the Browne weights, all without regularizers, against the Gram-whitened weight with $(\eta,\lsy)$ selected. Fourth block: the protocol, without re-estimation on the validation months (with the corresponding rows of Table~\ref{tab_stat} as references), with a validation window of 36 months and a training window of 72 months, with one validation month and 107 training months, and with the smallest 20\% of stocks by market capitalization removed from estimation and evaluation. Fifth block: the one-step estimator \eqref{krr1floor}, computed by the solver of Appendix~\ref{app_vec}, against the CMK estimator, at $\eta=1$ and $\lsy=0$, on every fourth window. 594 test months unless noted.}
\end{table}

\begin{table}
\begin{center}
\resizebox{\textwidth}{!}{\input{tables/session8_tab7_floorshape.tex}}
\end{center}
\caption{\label{tab_floorshape}The shape of the floor along the multiplier grid at $\lsy=0$. Score differential, number of factors and predictive $R^{2,2}$ (in percent) of the floor $\eta\hat{\bm F}$ and of the identity-shaped floor $\eta(\tr\hat{\bm F}/m)\bm I_m$ at the multipliers $0$, $\tfrac12$, $1$ and $2$, for the linear kernel with $m=95$ and for the additive kernel with $m=295$ at $\ell=\ell_0$, on the 594 test months.}
\end{table}

\paragraph{Weights in Table~\ref{tab_robust}.} The Frobenius weight with first-moment weight $\kappa$ is $\Om_t=\Diag(\kappa,\bm I_{N_t})$, with no scaling by the size of the month. The scalar weight is $\Om_t=N_t^{-1/2}\bm I_{N_t+1}$. The Browne weight is
\begin{equation}\label{eqbrowne}
   \Om^{\rm B}_t\coloneqq\frac{u^{\rm F}}{\sqrt{\bar N}}\,q_{\bm U^{\rm F}}(\Zt_t,\Zt_t^\top)^{-1},\qquad q_{\bm U^{\rm F}}(\Zt_t,\Zt_t^\top)^{-1}=\begin{bmatrix}1&\bm 0^\top\\-\bm\mu^{\rm F}_t&\bm I_{N_t}\end{bmatrix}^\top\Diag\big(1,(\bm\Sigma^{\rm F}_t)^{-1}\big)\begin{bmatrix}1&\bm 0^\top\\-\bm\mu^{\rm F}_t&\bm I_{N_t}\end{bmatrix},
\end{equation}
where $\bm U^{\rm F}=(\bm U^{\sy,\rm F},u^{\rm F})$ is a first-stage estimate, with the Frobenius weight at $\lsy=0$ and $\kappa=1$, $\bm\mu^{\rm F}_t$ and $\bm\Sigma^{\rm F}_t$ are its conditional moments \eqref{eqmutSigmatest}, and $\bar N$ is the average cross-sectional size of the window. It whitens the augmented second-moment matrix by the inverse of its first-stage estimate, which is the Fisher-scoring weight of the Gaussian quasi-likelihood, and the factor $u^{\rm F}/\sqrt{\bar N}$ puts the loss on the scale of the scalar weight. All quantities are $O(m^\sy)$-dimensional by the Woodbury identity.

\paragraph{Identity-shaped floor in the additive space.} In the additive feature space with $m=295$ at $\ell=\ell_0$ the validation selects a positive multiplier in 3 of the 594 windows, so that the identity-shaped floor nearly reduces to the penalty alone, 22 points below $\eta\hat{\bm F}$ ($t=4.4$), with a Sharpe ratio of $4.8$ against $5.1$. At $\lsy=0$ the score falls from 845 at $\eta=0$ to 566 at $\eta=1$, whereas $\eta\hat{\bm F}$ raises it to 854. The predictive $R^{2,2}$ of the identity-shaped floor turns negative from $\eta=\frac12$ on, since the added variance in the strong directions is a bias of the second moment. In the additive space most eigenvalues of the Gram matrix are small, so the same trace spread evenly over the directions is a large shift in the few strong ones. Table~\ref{tab_floorshape} gives the scores along the multiplier grid.

\paragraph{Weight and floor shape.} The second block of Table~\ref{tab_robust} varies the weight and the shape of the floor together, on the January, May and September test months: the scalar weight, which is the Frobenius weight scaled by $N_t^{-1/2}$, with the floor $\eta\hat{\bm F}$ and with the identity-shaped floor, $(\eta,\lsy)$ selected by validation. Under this weight the CMK program has no closed form, and the rows are computed by the solver of Appendix~\ref{app_vec}.\footnote{These 148 test months have larger score differentials on average than the other months, January in particular. The comparison is paired on the same months.} Both score below the CMK estimator with the Gram-whitened weight, 998 and 981 against 1,045, with Sharpe ratios of $3.4$ and $3.3$ against $4.9$ ($p<0.01$). The score differences are not significant ($t=1.7$ and $1.1$). They come from a few early months in which some features are collinear, where the one-step and the CMK estimator differ (paragraph on the one-step estimator below), so that before 1985 the comparison mixes the weight with the method of computation. From 1985 on the three scores are within 30 points of each other. The Sharpe ratios are not. The scalar weight lets the weight of a month grow with $N_t$ in the second-moment term, which Section~\ref{sec_weights} argues is too much, and the portfolio built on its moments earns about 30\% less in Sharpe ratio. Without regularizers the scalar weight scores 900 with a Sharpe ratio of $3.3$. The third block varies the weight alone. The Frobenius weight, the natural starting point of the theory, weights the moment terms of a month by its Gram matrix, whereas the errors of the single-month estimates do not shrink with the cross section (Section~\ref{sec_weights}). At $m=40$ and without regularizers, on every second window, it scores 814 against 848 for the Gram-whitened weight without regularizers and 859 with them ($t=2.4$ against the latter), with a Sharpe ratio of $3.4$ against $3.9$, for both $\kappa=1$ and $\kappa=100$. Without regularizers the Browne weight and the Gram-whitened weight score 847 and 848 at $m=40$, and 788 and 798 at $m=95$ ($t=2.7$), where both have a Sharpe ratio of $4.4$.

\paragraph{Validation window and universe.} The fourth block of Table~\ref{tab_robust} varies the protocol. Re-estimating on the training and validation months is worth 82 to 105 points and $0.3$ to $0.4$ of Sharpe ratio, and without it the models keep their ranking. For the linear kernel with $m=95$, a validation window of 36 months with a training window of 72 months, everything re-estimated on the union, costs 5 points of score differential ($t=3.8$) and $0.2$ of Sharpe ratio: the hyper-parameters are chosen on more months, but the training window is shorter. A single validation month with 107 training months costs 7 points ($t=5.8$) and $0.4$ of Sharpe ratio: the selection is noisier. The 96/12 split of the paper is the better compromise on this panel. Removing the smallest 20\% of stocks by market capitalization from both estimation and evaluation raises the predictive $R^{2,2}$ from $0.87$ to $1.74$\% and the score differential from 817 to 927, and lowers the Sharpe ratio from $4.9$ to $4.4$.

\paragraph{Tikhonov regularization.} We set $\varrho=0$ and treat the eigenvalues of $\bm G_t$ below $10^{-9}$ times the largest as zero. The CMK estimator then performs worse than with $\varrho=10^{-3}$. For the linear kernel with $m=95$ and $(\eta,\lsy)$ selected the score differential falls by $4.3$ points ($t=4.3$) and the Sharpe ratio from $4.91$ to $4.83$ ($p=0.04$). The loss is 12.9 points before 1985, where the predictive $R^{2,2}$ becomes $-18$\%, and $1.4$ points from 1985 on ($t=3.5$). For the additive kernel with $m=295$ at $\ell=\ell_0$, on 174 test months (every second window to 1990 and every sixth window from 1985), it is 59 points ($t=4.0$), and 28 points from 1985 on. The CMK estimator fixes $\muf$ at the sample mean $\bar{\bm f}$, and in a weak direction the coordinate has a variance of order $u^\id/\lambda_k(\bm G_t)$. The floor regularizes the covariance in these directions but not the mean. The Tikhonov regularization shrinks the coordinates and thereby the mean. A hard truncation at the same level, the eigenvalues below $10^{-3}$ times their mean set to zero, performs like the Tikhonov regularization in the linear feature space (score differential $816.4$ against $817.5$, Sharpe ratio $4.89$ against $4.91$), where the rank deficiency is exact. In the additive feature space, where the eigenvalues of $\bm G_t$ decay without a gap, it discards about half of the directions in every month and loses 85 points on 99 test months ($t=2.7$). To first order the regularization multiplies the target mean $\muf$ from the left, and the target covariance $\bm C$ from both sides, by the window average of $(\bm G_t+\varrho_t\bm I_m)^{-1}\bm G_t$. We measure this on ten windows at the selected hyper-parameters. In relative Euclidean norm, the implied conditional mean of the following test month moves by 0.4 to 10\% in the linear feature space, by less than 2\% from 1990 on, and by 9 to 15\% in the additive feature space. In relative Frobenius norm, the implied systematic covariance matrix moves by 0.2 to 4\% in the linear and by 3 to 8\% in the additive feature space.

\paragraph{One-step and CMK estimators.} The fifth block compares the one-step estimator \eqref{krr1floor}, computed by the solver of Appendix~\ref{app_vec}, with the CMK estimator of Section~\ref{sec_twostage} that the paper uses, at $\eta=1$, $\lsy=0$ and $m=95$ on every fourth window. From 1985 on the two coincide: $u^\id$ agrees to five digits and the test scores differ by less than one point in the median month. Before 1985, when some characteristics are not yet available and the windows contain rank-deficient months and many weak directions, the scores differ by 42 points in the median month, in favor of the CMK estimator, and the one-step estimator fits the second moment of these early months worse, as the $R^{2,2}$ column of the block shows.

\paragraph{Idiosyncratic specification.} Idiosyncratic variances of stock returns vary with characteristics such as size and past idiosyncratic volatility. The specification $\bm\phi^\id(z)=[1,\ \text{size},\ \text{idiosyncratic volatility}]$, with the rank-transformed \texttt{mvel1} and \texttt{idiovol} of \citet{gukellyxiu20b}, makes $\bm\Sigma^\id_t=\Diag(\bm\phi^\id_t\bm U^\id\bm\phi^{\id\top}_t)$ a PSD quadratic function of the two characteristics with six parameters. The kernel identifies $\bm U^\id$ when the products $\phi^\id_i\phi^\id_j$, $i\le j$, are linearly independent functions (Remark~\ref{remidentU} in Appendix~\ref{app_proofs}). It is estimated jointly with the systematic block by the same CMK program, in the linear feature space with the scalar weight $\Om_t=N_t^{-1/2}\bm I_{N_t+1}$ and without floor, on the protocol without re-estimation. At $m=20$ and $\lsy=0$, with $\bm U^\id$ averaged over the windows and the two characteristics at their rank-transformed values $(0,0)$, $(1,-1)$ and $(-1,1)$, the implied monthly idiosyncratic volatility is 10\% for the median stock, 16\% for the largest stock with the lowest past volatility and 33\% for the smallest stock with the highest past volatility. The fitted $3\times3$ matrix $\bm U^\id$ has rank two in 61\% and rank three in 26\% of the windows, and the predictive $R^{2,2}$ rises from $0.67$ to $0.68$\% (paired $t=5.4$) while the predictive $R^2$ is unchanged. The specification is thus estimable and informative about the cross section of idiosyncratic risk, at a modest gain in the second-moment fit. Its consequences for the factor portfolios, which then differ from the ordinary least squares portfolios (Section~\ref{sec_factorSDF}), are left for future work.

\paragraph{Covariance forecasts.} Table~\ref{tab_gmv} evaluates the covariance matrices directly. The GMV portfolio of a model has the weights $\bm\Sigma_t^{-1}\bm 1/(\bm 1^\top\bm\Sigma_t^{-1}\bm 1)$ and the predicted variance $1/(\bm 1^\top\bm\Sigma_t^{-1}\bm 1)$. Both follow from the sufficient statistics by the Woodbury identity, since $\bm 1$ is the first column of $[\bm 1,\bm z_t]$. A lower out-of-sample variance of the GMV return indicates a better covariance forecast \citep{engleledoitwolf19}. The paired test takes the difference of the squared demeaned GMV returns of a model and of the linear kernel with $m=95$. Its $t$-statistic uses Newey--West standard errors with 12 lags and is positive when the model has the higher variance. The variance ratio is the sample variance of the standardized return, the return divided by its predicted standard deviation, for the GMV and for the UMVE portfolio. It is one for a calibrated forecast. The covariance part of the score differential is the average of $\Scal(\bm x_{t+1},\bm 0,u^{\rm bm}\bm I_{N_t})-\Scal(\bm x_{t+1},\bm 0,\bm\Sigma_t)$, the score differential with the means set to zero. The mean part is the remainder, the average of $\Scal(\bm x_{t+1},\bm 0,\bm\Sigma_t)-\Scal(\bm x_{t+1},\bm\mu_t,\bm\Sigma_t)$. The IPCA rows come from a rerun that adds these quantities. For $K=20$, $40$ and $60$ the nonconvex fit of the rerun ends in slightly different local optima than the runs of Table~\ref{tab_stat}, and its score differentials differ from those of Table~\ref{tab_stat} by at most 0.3 points.

The linear kernel with $m=95$ has the lowest GMV volatility, 8.7\% a year over the full sample and 8.0\% from 1985 on. IPCA with $K=80$ has the highest volatility, 10.8\%, because of months before 1985, and its difference to the linear kernel is not significant. More features lower the GMV volatility of the linear kernel slightly but not significantly. The Gaussian features raise the covariance part of the score by 58 points ($t=5.2$) and the mean part by 1.9 points ($t=4.5$), but they do not lower the GMV volatility. Every model underpredicts the variance of its GMV portfolio, the linear kernel by less than IPCA. The UMVE variance ratios lie between 1.7 and 3.1. Without the mean part, the score differences between the linear kernel with $m=95$ and IPCA keep their sign and their significance. The mean parts differ by at most 1.7 points.

\begin{table}
\begin{center}
\resizebox{\textwidth}{!}{\input{tables/rev16_tab_gmv.tex}}
\end{center}
\caption{\label{tab_gmv}Covariance forecasts. Annualized out-of-sample volatility (in percent) of the GMV portfolio over the 594 test months and over the 444 test months with returns from January 1985 on, with the Newey--West $t$-statistic of the paired difference of squared demeaned returns to the linear kernel with $m=95$. Variance ratio: sample variance of the GMV and the UMVE return divided by its predicted variance. Covariance part and mean part of the score differential, as differences to the linear kernel with $m=95$, with Newey--West $t$-statistics. Every CMK row selects $(\eta,\lsy)$ by validation, and the selected additive kernel also selects the bandwidth and $m$, as in Table~\ref{tab_stat}. A positive $t$-statistic of the GMV volatility means a higher variance than the linear kernel with $m=95$. Variance ratios and score parts are over the 594 test months. The IPCA rows come from a rerun whose score differentials differ from those of Table~\ref{tab_stat} by at most 0.3 points.}
\end{table}

\paragraph{Newey--West lags.} The $t$-statistics of the paper use 12 lags. With 24 and 36 lags, the $t$-statistic of the score difference between the selected additive-kernel specification and the linear kernel with $m=95$ is $4.8$ at both lags instead of $5.4$, that of its difference to IPCA with $K=60$ is $4.2$ and $4.1$ instead of $4.7$, and the $t$-statistics of the mean UMVE returns of Table~\ref{tab_ap} fall from $14$--$19$ to $11$--$15$ and $10$--$13$. The differences between IPCA and the linear kernel with $m=95$ remain significant at every lag ($|t|$ between $2.7$ and $4.8$). No conclusion of Section~\ref{sec_empiricsNEW} depends on the lag.

\subsection{Second Panel}\label{app_jkp}

Table~\ref{tab_jkp} repeats the study on the stock-level data of \citet{jensenkellypedersen23} (JKP), which replace both the characteristics and the sample of Section~\ref{sec_dataprotocol}. The panel consists of the common stocks on NYSE, AMEX and NASDAQ in the JKP data of Wharton Research Data Services, from February 1963 to December 2025 (755 months). JKP's common-stock filter does not condition on incorporation in the US, so the panel also contains the common stocks of firms incorporated abroad, about 7\% of the stock-months, rising from 1\% of a cross section in the 1960s to 17\% at the end of 2025. Of their 153 characteristics we keep the 121 that are missing in less than 30\% of the stock-months of the whole sample. This choice uses no returns, but it uses the availability of the characteristics after the test months. We remove the nano caps, the stocks below the first percentile of NYSE market equity, together with the stocks without a JKP size group, and the stock-months with more than a third of the characteristics missing. The covariate vector $z_{t,i}\in\R^{121}$ consists of these characteristics, rank-transformed to $(-1,1)$ within the month with missing values set to zero. The ranks are taken before the stocks without a return in month $t+1$ are dropped. The drop removes $0.3$\% of the stock-months, almost all of them stocks that left CRSP during month $t$ and whose delisting return is part of their return in month $t$. The panel also carries the eight macroeconomic predictors of \citet{welchgoyal08}, which are not used. The return over month $t+1$ is JKP's, which compounds the delisting return and imputes $-30$\% where CRSP records a performance delisting without a return. The riskless rate is the one-month Treasury bill return of CRSP that JKP subtract, and the rate of \citet{french} for 2025, which the CRSP series in the data does not cover. The two rates agree closely (correlation $0.995$, mean absolute difference $1.4$ basis points per month). The cross section has 749 stocks in the first month, 3,480 on average and at most 6,271 in 1997. The feature matrix $[\bm 1,\bm z_t]$ has 122 columns. Its rank is 122 from the return month November 1971 on and 112 to 118 before. Everything else is as in Section~\ref{sec_dataprotocol}, with the feature dimensions adjusted to the covariate dimension. The linear kernel is run with $m\in\{20,40,122\}$, and $m=122$ is the full feature space and the reference of the paired tests. The additive kernel is run with 200 landmarks and $m\in\{179,229,322\}$, which keeps the distances of the paper's $m=152$ and $202$ to the full feature space, here $m=322$. The median heuristic gives $\ell_0^2=67.6$, a squared distance of $0.56$ per characteristic, and the landmarks capture 89\% of the trace of the Gaussian kernel matrix at $\ell=\ell_0$ and 35\% at $\ell=\ell_0/2$. Every month has more stocks than the largest specification has features, so all 647 test months from February 1972 to December 2025 are reported.

\begin{table}
\begin{center}
\resizebox{\textwidth}{!}{\input{tables/rev11_tab_jkp.tex}}
\end{center}
\caption{\label{tab_jkp}The panel of \citet{jensenkellypedersen23}. Out-of-sample measures of Table~\ref{tab_stat}, the annualized Sharpe ratio of the UMVE portfolio, its maximum drawdown (``max DD'', in percent and negative, the largest loss from a peak of the cumulative return, compounded at 10\% annual volatility), its average monthly turnover at unit gross exposure and its Sharpe ratio at unit gross exposure before costs (``SR net 0'') and net of proportional transaction costs of 25 basis points of turnover, as in Table~\ref{tab_ap}, on the panel described in the text. The linear kernel with $m=122$ is the full linear feature space, the reference of the column ``diff.\ lin.'', the additive kernel with $m=322$ the full additive feature space, and the selected additive kernel the reference of the column ``diff.\ Gauss.''. 647 test months, February 1972 to December 2025.}
\end{table}

The levels of the measures differ from those of Table~\ref{tab_stat}, since the characteristics, the universe and the sample period differ. The score differential and the predictive $R^{2,2}$ are higher on this panel. Within the panel, every score comparison of Section~\ref{sec_empiricsNEW} has the same sign as in Table~\ref{tab_stat}. For the linear kernel the score differential rises with the number of features, from 870 at $m=20$ to 949 at $m=122$, and the validation selects $m=122$ in 98\% of the windows. At $m=122$ the validation selects a positive floor in 93\% and a positive penalty in 55\% of the windows, and the regularizers reduce the number of factors from 57 to 26 and add 34 points of score ($t=5.8$). The additive kernel improves on the linear kernel by 24 to 49 points ($t$ between $4.1$ and $6.6$), and by 56 points ($t=5.3$) when the bandwidth and $m$ are selected, against 60 points ($t=5.4$) in Table~\ref{tab_stat}. Its score differential is positive in 96\% of the test months and in every calendar year, and its improvement on the linear kernel is significant in each sub-period: 38 points in 1972--1989 ($t=3.7$), 69 in 1990--2009 ($t=3.7$), 62 in 2010--2025 ($t=2.6$) and 87 in the crash months ($t=3.2$). With 24 and 36 Newey--West lags these four $t$-statistics lie between $2.3$ and $3.8$. The last four years, 2022 to 2025, lie after the end of the main panel. In these 48 test months the additive kernel improves on the linear kernel by 97 points ($t=3.9$) and the linear kernel improves on IPCA with $K=60$ by 110 points ($t=4.6$). The Sharpe ratios do not keep their order. They are $3.0$ for the selected additive kernel, $2.3$ for the linear kernel with $m=122$ and $3.1$ for IPCA with $K=60$, which exceeds the linear kernel at the 1\% level. The predictive $R^2$ is negative for every model in these years.

In the full linear feature space ($m=122$) the CMK estimator improves on IPCA by 54 to 68 points of score ($t$ between $4.7$ and $6.0$), and the selected additive kernel improves on IPCA with $K=60$ by 111 points ($t=5.1$). On this panel the Sharpe ratios separate the models less clearly than in Table~\ref{tab_ap}. The UMVE Sharpe ratio of the linear kernel with $m=122$, $5.0$, is close to that of IPCA with $K=40$ and $K=60$ ($p=0.76$ against $K=60$) and exceeds that of IPCA with $K=20$, $4.7$ ($p=0.02$). That of the selected additive kernel, $5.3$, is above IPCA's at the 10\% level ($p=0.06$ against $K=60$), and the full additive feature space reaches $5.6$ at both bandwidths ($p<0.01$).

At 10\% annual volatility the maximum drawdowns of the UMVE portfolios lie between 11 and 21\%, 13\% for the selected additive kernel, against 36\% for the equally weighted portfolio. The portfolios turn over 87 to 93\% of their book per month. Net of 25 basis points the Sharpe ratios are $2.9$ to $3.8$ for the linear kernel, $4.1$ to $4.3$ for the additive kernel and $3.7$ to $3.8$ for IPCA, so that the additive kernel stays ahead after costs. The linear kernel with $m=122$ is within $0.1$ of IPCA with $K=40$ and with $K=60$ before and after costs. The five factors of \citet{famafrench15} explain at most 4\% of the variation of the UMVE returns and leave annualized intercepts of 37 to 54\% at 10\% volatility.

The robustness checks of Table~\ref{tab_robust} were run on this panel as well. At $m=122$ the floor alone recovers 28 of the 34 points of the two regularizers, the penalty alone none, and the identity-shaped floor scores 19 points below the floor $\eta\hat{\bm F}$ ($t=10.8$). The re-estimation on the training and validation months is worth 87 to 112 points. Two checks differ from the main panel. Removing the smallest 20\% of the stocks, here of the stocks above the nano caps, lowers the score differential of the linear kernel with $m=122$ from 949 to 906 and its Sharpe ratio from $5.0$ to $4.0$, whereas on the main panel it raises the score differential. The one-step and the CMK estimator agree to $0.1$ points of score in the median month before 1985 and to solver tolerance after, whereas on the main panel they differ before 1985.

Table~\ref{tab_decomp_jkp} repeats the decomposition of Table~\ref{tab_decomp}. IPCA with $K=95$ was not run on this panel, so the moment step is compared with IPCA with $K=60$. The score changes have the same signs as on the main panel, and the Sharpe ratio rises with the floor.

\begin{table}
\begin{center}
\input{tables/rev19_tab_decomp_jkp.tex}
\end{center}
\caption{\label{tab_decomp_jkp}Decomposition of the gain over IPCA on the second panel, as in Table~\ref{tab_decomp}, from IPCA with $K=60$ to the CMK estimator in the full linear feature space, $m=122$. Score differential, its change from the row above with the paired Newey--West $t$-statistic, and the annualized Sharpe ratio (SR) of the UMVE portfolio. Rows 2 to 4 have $\lsy=0$, rows 2 and 3 have $\eta=0$.}
\end{table}

\end{appendix}

\end{document}

%% file: tables/rev11_tab1_stat.tex
\begin{tabular}{lrrrrrrrrrrr}
\toprule
model & \# factors & $R^2$ (\%) & $R^{2,2}$ (\%) & score & total $R^2$ (\%) & adj.\ total $R^2$ (\%) & $\rho^{\bm f}$ (\%) & diff.\ lin. & $t$ & diff.\ Gauss. & $t$ \\
\midrule
CMK, linear, $m=20$ & 10.6 & 0.20 & 0.82 & 763 & 14.0 & 13.5 & 18.4 & -54.3 & -8.1 & -114.5 & -7.0 \\
CMK, linear, $m=40$ & 15.3 & 0.23 & 0.84 & 791 & 15.1 & 14.1 & 19.9 & -26.3 & -8.2 & -86.6 & -6.3 \\
CMK, linear, $m=95$ & 22.0 & 0.29 & 0.87 & 817 & 17.3 & 15.0 & 22.8 & 0.0 & -- & -60.2 & -5.4 \\
CMK, linear, $m$ selected & 21.9 & 0.29 & 0.87 & 817 & 17.2 & 15.0 & 22.7 & -0.1 & -0.6 & -60.3 & -5.4 \\
\midrule
CMK, linear + Gaussian ($\ell=1\times$), $m=152$ & 21.5 & 0.35 & 0.96 & 844 & 19.5 & 15.9 & 24.2 & 26.2 & 6.8 & -34.0 & -4.3 \\
CMK, linear + Gaussian ($\ell=1\times$), $m=202$ & 30.5 & 0.36 & 0.99 & 857 & 21.1 & 16.2 & 26.5 & 39.2 & 5.9 & -21.0 & -4.3 \\
CMK, linear + Gaussian ($\ell=1\times$), $m=295$ & 50.7 & 0.39 & 0.98 & 872 & 23.8 & 16.5 & 30.2 & 54.8 & 4.5 & -5.4 & -2.5 \\
CMK, linear + Gaussian ($\ell=0.5\times$), $m=152$ & 19.8 & 0.33 & 0.96 & 838 & 19.3 & 15.6 & 24.2 & 20.8 & 6.0 & -39.5 & -4.8 \\
CMK, linear + Gaussian ($\ell=0.5\times$), $m=202$ & 29.0 & 0.35 & 0.96 & 848 & 20.6 & 15.7 & 26.5 & 30.6 & 5.1 & -29.7 & -5.2 \\
CMK, linear + Gaussian ($\ell=0.5\times$), $m=295$ & 48.5 & 0.36 & 0.96 & 860 & 23.1 & 15.8 & 30.1 & 43.0 & 3.8 & -17.2 & -7.1 \\
CMK, linear + Gaussian, $(\ell,m)$ selected & 34.7 & 0.38 & 1.03 & 878 & 21.9 & 16.6 & 27.4 & 60.2 & 5.4 & -- & -- \\
\midrule
IPCA, $K=20$ & 20.0 & 0.28 & 0.83 & 774 & 14.4 & 13.9 & 21.0 & -43.4 & -4.8 & -103.6 & -5.2 \\
IPCA, $K=40$ & 40.0 & 0.28 & 0.84 & 780 & 15.4 & 14.4 & 22.2 & -37.4 & -4.0 & -97.6 & -4.8 \\
IPCA, $K=60$ & 60.0 & 0.29 & 0.84 & 785 & 16.1 & 14.7 & 22.8 & -33.0 & -3.7 & -93.2 & -4.7 \\
IPCA, $K=80$ & 80.0 & 0.29 & 0.84 & 788 & 16.8 & 14.9 & 23.2 & -29.6 & -3.3 & -89.8 & -4.5 \\
IPCA, $K=95$ & 94.2 & 0.30 & 0.84 & 793 & 17.3 & 15.0 & 23.2 & -25.0 & -3.2 & -85.2 & -4.5 \\
\bottomrule
\end{tabular}

%% file: tables/rev19_tab3_decomp.tex
\begin{tabular}{lrrrr}
\toprule
Step & Score & Change & $t$ & SR \\
\midrule
IPCA, $K=60$ & 785 &  &  & 4.55 \\
IPCA, $K=95$ & 793 & $+8.0$ & $5.2$ & 4.49 \\
Moment step, full feature space, Tikhonov $\varrho=10^{-3}$ & 808 & $+15.5$ & $2.4$ & 4.48 \\
Noise correction: subtract $\hat u\bm G^+_{t,\varrho}$ & 798 & $-10.4$ & $-4.0$ & 4.44 \\
Floor at $\eta=1$ & 814 & $+16.3$ & $4.1$ & 4.79 \\
Select $(\eta,\lsy)$: CMK estimator & 817 & $+3.6$ & $2.6$ & 4.91 \\
\bottomrule
\end{tabular}

%% file: tables/rev11_tab2_ap.tex
\begin{tabular}{lrrrrrrrrrrrrrrr}
\toprule
model & SR pred. & SR UMVE & $t$ & SR CMVE & $p$ vs $1/N$ & $p$ vs IPCA & intercept (\%) & $t$ (int.) & adj.\ $R^2$ (\%) & turnover & SR net 0 & SR net 10 & SR net 25 & SR net 50 & SR net 100 \\
\midrule
CMK, linear, $m=20$ & 5.72 & 3.96 & 14.4 & 3.77 & 0.00 & 0.07 & 37.7 & 13.5 & 3.0 & 0.74 & 3.87 & 3.59 & 3.16 & 2.41 & 0.81 \\
CMK, linear, $m=40$ & 6.74 & 4.37 & 16.3 & 4.12 & 0.00 & 0.47 & 42.0 & 15.2 & 2.7 & 0.78 & 4.20 & 3.90 & 3.43 & 2.61 & 0.88 \\
CMK, linear, $m=95$ & 8.53 & 4.91 & 18.7 & 4.48 & 0.00 & 0.02 & 47.8 & 17.7 & 2.0 & 0.93 & 4.91 & 4.52 & 3.91 & 2.85 & 0.61 \\
\midrule
CMK, linear + Gaussian ($\ell=1\times$), $m=295$ & 8.55 & 5.08 & 19.4 & 4.61 & 0.00 & 0.01 & 49.6 & 18.2 & 1.3 & 0.85 & 4.90 & 4.56 & 4.04 & 3.15 & 1.29 \\
CMK, linear + Gaussian ($\ell=0.5\times$), $m=295$ & 8.62 & 4.97 & 18.6 & 4.52 & 0.00 & 0.01 & 48.3 & 17.2 & 2.2 & 0.85 & 4.85 & 4.51 & 3.98 & 3.07 & 1.16 \\
CMK, linear + Gaussian, $(\ell,m)$ selected & 8.89 & 5.05 & 19.0 & 4.54 & 0.00 & 0.01 & 49.2 & 17.6 & 1.7 & 0.89 & 4.95 & 4.60 & 4.07 & 3.16 & 1.25 \\
\midrule
IPCA, $K=20$ & 6.78 & 4.74 & 17.1 & 4.27 & 0.00 & 0.31 & 45.5 & 15.9 & 3.6 & 0.93 & 4.74 & 4.38 & 3.83 & 2.88 & 0.89 \\
IPCA, $K=40$ & 7.59 & 4.72 & 17.2 & 4.38 & 0.00 & 0.06 & 45.7 & 16.0 & 2.7 & 0.89 & 4.77 & 4.38 & 3.79 & 2.77 & 0.63 \\
IPCA, $K=60$ & 7.98 & 4.55 & 16.9 & 4.33 & 0.00 & -- & 44.2 & 16.1 & 2.0 & 0.84 & 4.61 & 4.23 & 3.63 & 2.61 & 0.47 \\
\midrule
$1/N_t$ & -- & 0.50 & 3.3 & -- & -- & -- & 0.6 & 1.3 & 92.8 & -- & -- & -- & -- & -- & -- \\
\bottomrule
\end{tabular}

%% file: tables/rev11_tab6_sim.tex
\begin{tabular}{llrrrrrrrr}
\toprule
DGP & model & $R^2$ (\%) & $R^{2,2}$ (\%) & score & total $R^2$ (\%) & \# factors & SR pred. & SR UMVE & SR CMVE \\
\midrule
well specified & $m=5$ & 0.72 & 0.93 & 64.6 & 12.3 & 2.9 & 1.31 & 0.86 & 0.84 \\
 & $m=10$ & 0.92 & 1.44 & 85.3 & 16.5 & 3.8 & 1.51 & 0.86 & 0.84 \\
 & $m=20$ & 1.13 & 2.59 & 133.6 & 25.0 & 4.3 & 1.87 & 0.75 & 0.73 \\
 & population & 1.31 & 2.67 & 134.5 & 25.0 & 3.0 & 1.12 & 1.27 & 1.27 \\
\midrule
misspecified & $m=5$ & 0.62 & 0.89 & 62.7 & 12.1 & 3.0 & 1.37 & 0.88 & 0.85 \\
 & $m=10$ & 0.75 & 1.33 & 81.3 & 15.9 & 4.1 & 1.53 & 0.78 & 0.77 \\
 & $m=20$ & 0.95 & 2.35 & 124.4 & 23.8 & 4.5 & 2.03 & 0.82 & 0.82 \\
 & population & 1.34 & 3.32 & 157.1 & 28.3 & 3.0 & 1.38 & 1.32 & 1.32 \\
\bottomrule
\end{tabular}

%% file: tables/rev11_tab5_periods.tex
\begin{tabular}{llrrrrrrrr}
\toprule
period & model & months & $R^2$ (\%) & $R^{2,2}$ (\%) & score & SR UMVE & SR $1/N$ & diff.\ score & $t$ \\
\midrule
full & CMK, linear + Gaussian, $(\ell,m)$ selected & 594 & 0.38 & 1.03 & 878 & 5.05 & 0.50 & -- & -- \\
full & CMK, linear, $m=95$ & 594 & 0.29 & 0.87 & 817 & 4.91 & 0.50 & -60.2 & -5.4 \\
full & IPCA, $K=60$ & 594 & 0.29 & 0.84 & 785 & 4.55 & 0.50 & -93.2 & -4.7 \\
full & IPCA, $K=20$ & 594 & 0.28 & 0.83 & 774 & 4.74 & 0.50 & -103.6 & -5.2 \\
\midrule
1972--1989 & CMK, linear + Gaussian, $(\ell,m)$ selected & 210 & 0.38 & 1.97 & 1000 & 6.93 & 0.36 & -- & -- \\
1972--1989 & CMK, linear, $m=95$ & 210 & 0.28 & 1.94 & 945 & 6.84 & 0.36 & -55.5 & -3.2 \\
1972--1989 & IPCA, $K=60$ & 210 & 0.28 & 1.86 & 894 & 6.75 & 0.36 & -106.3 & -3.4 \\
1972--1989 & IPCA, $K=20$ & 210 & 0.25 & 1.84 & 889 & 6.52 & 0.36 & -111.0 & -3.5 \\
\midrule
1990--2009 & CMK, linear + Gaussian, $(\ell,m)$ selected & 240 & 0.42 & 0.97 & 929 & 5.89 & 0.49 & -- & -- \\
1990--2009 & CMK, linear, $m=95$ & 240 & 0.30 & 0.76 & 858 & 5.98 & 0.49 & -70.9 & -4.1 \\
1990--2009 & IPCA, $K=60$ & 240 & 0.30 & 0.74 & 835 & 5.64 & 0.49 & -94.3 & -3.2 \\
1990--2009 & IPCA, $K=20$ & 240 & 0.29 & 0.74 & 821 & 5.59 & 0.49 & -107.9 & -3.7 \\
\midrule
2010--2021 & CMK, linear + Gaussian, $(\ell,m)$ selected & 144 & 0.29 & 0.60 & 614 & 2.86 & 0.76 & -- & -- \\
2010--2021 & CMK, linear, $m=95$ & 144 & 0.27 & 0.53 & 565 & 2.58 & 0.76 & -49.2 & -1.9 \\
2010--2021 & IPCA, $K=60$ & 144 & 0.27 & 0.51 & 542 & 2.15 & 0.76 & -72.1 & -1.6 \\
2010--2021 & IPCA, $K=20$ & 144 & 0.27 & 0.51 & 528 & 2.53 & 0.76 & -85.8 & -1.8 \\
\midrule
crashes & CMK, linear + Gaussian, $(\ell,m)$ selected & 57 & -0.55 & 1.91 & 1804 & 3.54 & -1.07 & -- & -- \\
crashes & CMK, linear, $m=95$ & 57 & -0.66 & 1.54 & 1710 & 3.99 & -1.07 & -93.7 & -3.2 \\
crashes & IPCA, $K=60$ & 57 & -0.66 & 1.52 & 1666 & 3.37 & -1.07 & -137.2 & -2.7 \\
crashes & IPCA, $K=20$ & 57 & -0.63 & 1.52 & 1643 & 3.80 & -1.07 & -160.4 & -3.0 \\
\bottomrule
\end{tabular}

%% file: tables/rev12_tab4_robust.tex
\begin{tabular}{lrrrrrrrrrr}
\toprule
specification & \# factors & $R^2$ (\%) & $R^{2,2}$ (\%) & score & total $R^2$ (\%) & SR UMVE & SR CMVE & diff.\ score & $t$ & $p$ SR \\
\midrule
linear $m=95$: $(\eta,\lsy)$ selected (Table~\ref{tab_stat}) & 22.0 & 0.29 & 0.87 & 817 & 17.3 & 4.91 & 4.48 & -- & -- & -- \\
linear $m=95$: $\eta=0$, $\lsy=0$ & 45.1 & 0.29 & 0.84 & 798 & 17.3 & 4.44 & 4.22 & -19.9 & -4.7 & 0.00 \\
linear $m=95$: $\eta=1$, $\lsy=0$ & 28.9 & 0.29 & 0.84 & 814 & 17.3 & 4.79 & 4.52 & -3.6 & -2.6 & 0.13 \\
linear $m=95$: $\eta=0$, $\lsy$ selected & 41.6 & 0.29 & 0.85 & 798 & 17.3 & 4.43 & 4.34 & -19.8 & -4.7 & 0.00 \\
linear $m=95$: moment step, $\hat{\bm F}$ in every direction & 94.2 & 0.29 & 0.84 & 808 & 17.3 & 4.48 & 4.27 & -9.5 & -4.8 & 0.00 \\
linear $m=95$: identity-shaped floor, $(\eta,\lsy)$ selected & 27.5 & 0.29 & 0.86 & 808 & 17.3 & 4.63 & 4.03 & -9.4 & -8.4 & 0.13 \\
linear $m=95$: identity-shaped floor, $\eta=1$, $\lsy=0$ & 17.4 & 0.29 & 0.83 & 783 & 17.3 & 4.52 & 4.37 & -34.1 & -14.9 & 0.11 \\
\midrule
linear + Gaussian ($\ell=1\times$) $m=295$: $(\eta,\lsy)$ selected (Table~\ref{tab_stat}) & 50.7 & 0.39 & 0.98 & 872 & 23.8 & 5.08 & 4.61 & -- & -- & -- \\
linear + Gaussian ($\ell=1\times$) $m=295$: $\eta=0$, $\lsy=0$ & 90.3 & 0.39 & 0.83 & 845 & 23.8 & 4.87 & 4.32 & -27.2 & -6.4 & 0.08 \\
linear + Gaussian ($\ell=1\times$) $m=295$: $\eta=1$, $\lsy=0$ & 66.7 & 0.39 & 0.77 & 854 & 23.8 & 5.04 & 4.51 & -18.3 & -9.5 & 0.63 \\
linear + Gaussian ($\ell=1\times$) $m=295$: $\eta=0$, $\lsy$ selected & 74.4 & 0.39 & 0.93 & 851 & 23.8 & 4.84 & 4.43 & -21.8 & -4.5 & 0.01 \\
linear + Gaussian ($\ell=1\times$) $m=295$: moment step, $\hat{\bm F}$ in every direction & 107.0 & 0.39 & 0.86 & 858 & 23.8 & 4.94 & 4.46 & -14.2 & -3.1 & 0.22 \\
linear + Gaussian ($\ell=1\times$) $m=295$: identity-shaped floor, $(\eta,\lsy)$ selected & 74.3 & 0.39 & 0.76 & 850 & 23.8 & 4.84 & 4.43 & -22.5 & -4.4 & 0.01 \\
linear + Gaussian ($\ell=1\times$) $m=295$: identity-shaped floor, $\eta=1$, $\lsy=0$ & 53.8 & 0.39 & -104.62 & 566 & 23.8 & 3.96 & 3.57 & -306.7 & -70.6 & 0.00 \\
\midrule
linear + Gaussian ($\ell=0.5\times$) $m=295$: $(\eta,\lsy)$ selected (Table~\ref{tab_stat}) & 48.5 & 0.36 & 0.96 & 860 & 23.1 & 4.97 & 4.52 & -- & -- & -- \\
linear + Gaussian ($\ell=0.5\times$) $m=295$: $\eta=0$, $\lsy=0$ & 88.5 & 0.36 & 0.83 & 837 & 23.1 & 4.88 & 4.30 & -23.9 & -6.8 & 0.53 \\
linear + Gaussian ($\ell=0.5\times$) $m=295$: $\eta=1$, $\lsy=0$ & 64.3 & 0.36 & 0.78 & 842 & 23.1 & 4.95 & 4.47 & -18.0 & -8.7 & 0.81 \\
linear + Gaussian ($\ell=0.5\times$) $m=295$: $\eta=0$, $\lsy$ selected & 72.0 & 0.36 & 0.92 & 842 & 23.1 & 4.92 & 4.47 & -18.4 & -4.5 & 0.73 \\
linear + Gaussian ($\ell=0.5\times$) $m=295$: moment step, $\hat{\bm F}$ in every direction & 107.0 & 0.36 & 0.86 & 849 & 23.1 & 4.93 & 4.43 & -11.7 & -3.0 & 0.68 \\
\midrule
linear $m=95$, Gram-whitened weight, floor $\eta\hat{\bm F}$, $(\eta,\lsy)$ selected & 21.9 & 0.70 & 0.89 & 1045 & 19.0 & 4.91 & 4.59 & -- & -- & -- \\
linear $m=95$, Gram-whitened weight, identity-shaped floor, $(\eta,\lsy)$ selected & 28.4 & 0.70 & 0.89 & 1033 & 19.0 & 4.45 & 4.12 & -12.2 & -5.2 & 0.06 \\
linear $m=95$, scalar weight $N_t^{-1/2}\bm I$, floor $\eta\hat{\bm F}$, $(\eta,\lsy)$ selected & -- & 0.57 & 0.62 & 998 & 19.0 & 3.44 & 2.59 & -46.9 & -1.7 & 0.01 \\
linear $m=95$, scalar weight $N_t^{-1/2}\bm I$, identity-shaped floor, $(\eta,\lsy)$ selected & -- & 0.70 & 0.64 & 981 & 19.0 & 3.25 & 2.05 & -63.9 & -1.1 & 0.00 \\
linear $m=95$, scalar weight $N_t^{-1/2}\bm I$, $\eta=0$, $\lsy=0$ & -- & 0.67 & 0.88 & 900 & 19.0 & 3.28 & 2.93 & -144.7 & -2.1 & 0.00 \\
\midrule
linear $m=95$, Gram-whitened weight, $\eta=0$, $\lsy=0$ & 45.1 & 0.29 & 0.84 & 798 & 17.3 & 4.44 & 4.22 & -- & -- & -- \\
linear $m=95$, Browne weight, $\eta=0$, $\lsy=0$ & 44.4 & 0.31 & 0.83 & 788 & 17.3 & 4.44 & 4.16 & -9.6 & -2.7 & 0.97 \\
linear $m=40$, Gram-whitened weight, $(\eta,\lsy)$ selected & 15.4 & 0.36 & 0.85 & 859 & 15.8 & 4.26 & 4.08 & -- & -- & -- \\
linear $m=40$, Gram-whitened weight, $\eta=0$, $\lsy=0$ & 26.0 & 0.36 & 0.85 & 848 & 15.8 & 3.91 & 3.67 & -11.2 & -4.3 & 0.04 \\
linear $m=40$, Frobenius weight, $\kappa=1$ & 23.1 & 0.32 & 0.84 & 814 & 15.8 & 3.43 & 3.22 & -44.5 & -2.4 & 0.00 \\
linear $m=40$, Frobenius weight, $\kappa=100$ & 23.1 & 0.32 & 0.84 & 814 & 15.8 & 3.38 & 3.14 & -44.7 & -2.4 & 0.00 \\
linear $m=40$, Browne weight, $\eta=0$, $\lsy=0$ & 26.2 & 0.36 & 0.85 & 847 & 15.8 & 3.97 & 3.78 & -12.1 & -4.0 & 0.09 \\
\midrule
no re-estimation, linear $m=95$ & 22.0 & 0.18 & 0.76 & 726 & 17.3 & 4.52 & 4.43 & -91.9 & -4.4 & 0.00 \\
no re-estimation, linear + Gaussian selected & 34.7 & 0.27 & 0.93 & 796 & 21.9 & 4.74 & 4.42 & -82.1 & -4.4 & 0.00 \\
no re-estimation, IPCA $K=60$ & 60.0 & 0.17 & 0.72 & 679 & 16.0 & 4.16 & 4.06 & -105.3 & -4.5 & 0.00 \\
36 validation months, 72 training months, linear $m=95$ & 22.3 & 0.29 & 0.85 & 813 & 17.3 & 4.72 & 4.22 & -4.6 & -3.8 & 0.00 \\
1 validation month, 107 training months, linear $m=95$ & 20.0 & 0.29 & 0.84 & 811 & 17.3 & 4.51 & 3.78 & -6.6 & -5.8 & 0.00 \\
smallest 20\% removed, linear $m=95$ & 25.1 & 0.27 & 1.74 & 927 & 22.5 & 4.40 & 3.96 & -- & -- & -- \\
\midrule
CMK estimator, $\eta=1$, $\lsy=0$, linear $m=95$ & 28.9 & 0.70 & 0.88 & 1041 & 19.0 & 4.72 & 4.46 & -- & -- & -- \\
one-step estimator by the solver, $\eta=1$, $\lsy=0$, linear $m=95$ & -- & 0.70 & 0.35 & 1031 & 19.0 & 4.71 & 4.48 & -9.7 & -2.0 & 0.84 \\
\bottomrule
\end{tabular}

%% file: tables/session8_tab7_floorshape.tex
\begin{tabular}{lrrrrrrrrrrrrr}
\toprule
specification & months & score (0) & factors (0) & $R^{2,2}$ (0) & score (0.5) & factors (0.5) & $R^{2,2}$ (0.5) & score (1) & factors (1) & $R^{2,2}$ (1) & score (2) & factors (2) & $R^{2,2}$ (2) \\
\midrule
linear $m=95$, floor $\eta\hat{\bm F}$ & 594 & 798 & 45.1 & 0.84 & 813 & 35.9 & 0.84 & 814 & 28.9 & 0.84 & 807 & 19.8 & 0.84 \\
linear $m=95$, identity-shaped floor & 594 & 798 & 45.1 & 0.84 & 801 & 24.0 & 0.84 & 783 & 17.4 & 0.83 & 755 & 11.2 & 0.76 \\
\midrule
linear + Gaussian ($\ell=1\times$) $m=295$, floor $\eta\hat{\bm F}$ & 594 & 845 & 90.3 & 0.83 & 854 & 77.4 & 0.80 & 854 & 66.7 & 0.77 & 848 & 50.3 & 0.70 \\
linear + Gaussian ($\ell=1\times$) $m=295$, identity-shaped floor & 594 & 845 & 90.3 & 0.83 & 622 & 62.6 & -25.29 & 566 & 53.8 & -104.62 & 507 & 41.0 & -424.24 \\
\bottomrule
\end{tabular}

%% file: tables/rev16_tab_gmv.tex
\begin{tabular}{lrrrrrrrrrr}
\toprule
 & \multicolumn{4}{c}{GMV volatility} & \multicolumn{2}{c}{Variance ratio} & \multicolumn{4}{c}{Score, difference to linear $m=95$} \\
\cmidrule(lr){2-5}\cmidrule(lr){6-7}\cmidrule(lr){8-11}
 & 1972--2021 & $t$ & 1985--2021 & $t$ & GMV & UMVE & Cov. & $t$ & Mean & $t$ \\
\midrule
CMK, linear, $m=20$ & 9.06 & 1.5 & 8.35 & 1.6 & 1.59 & 1.95 & $-51.9$ & $-7.7$ & $-2.4$ & $-6.1$ \\
CMK, linear, $m=40$ & 8.87 & 0.7 & 8.21 & 1.2 & 1.71 & 2.01 & $-24.5$ & $-7.5$ & $-1.9$ & $-5.7$ \\
CMK, linear, $m=95$ & 8.71 &  & 7.97 &  & 1.93 & 2.46 &  &  &  &  \\
\addlinespace
CMK, linear + Gaussian, $(\ell,m)$ selected & 9.03 & 1.8 & 8.35 & 1.9 & 2.43 & 3.09 & $+58.3$ & 5.2 & $+1.9$ & 4.5 \\
\addlinespace
IPCA, $K=20$ & 9.01 & 0.8 & 8.55 & 2.7 & 4.27 & 1.73 & $-42.0$ & $-4.6$ & $-1.7$ & $-4.3$ \\
IPCA, $K=40$ & 9.62 & 3.0 & 8.55 & 2.4 & 3.66 & 2.14 & $-36.7$ & $-3.9$ & $-0.9$ & $-2.3$ \\
IPCA, $K=60$ & 9.87 & 3.2 & 8.50 & 1.9 & 3.19 & 2.48 & $-32.4$ & $-3.7$ & $-0.3$ & $-0.7$ \\
IPCA, $K=80$ & 10.78 & 1.6 & 8.32 & 1.5 & 3.84 & 2.63 & $-29.6$ & $-3.3$ & $0.0$ & 0.0 \\
IPCA, $K=95$ & 9.32 & 2.3 & 8.22 & 1.2 & 2.27 & 2.77 & $-25.1$ & $-3.2$ & $+0.1$ & 0.3 \\
\addlinespace
equally weighted & 20.75 & 8.3 & 20.01 & 6.6 &  &  &  &  &  &  \\
\bottomrule
\end{tabular}

%% file: tables/rev11_tab_jkp.tex
\begin{tabular}{lrrrrrrrrrrrrrrr}
\toprule
model & \# factors & $R^2$ (\%) & $R^{2,2}$ (\%) & score & total $R^2$ (\%) & adj.\ total $R^2$ (\%) & diff.\ lin. & $t$ & diff.\ Gauss. & $t$ & SR UMVE & max DD (\%) & turnover & SR net 0 & SR net 25 \\
\midrule
CMK, linear, $m=20$ & 11.9 & 0.17 & 1.36 & 870 & 17.2 & 16.7 & -79.8 & -8.4 & -136.3 & -7.5 & 3.95 & -13.6 & 0.88 & 3.88 & 2.90 \\
CMK, linear, $m=40$ & 17.7 & 0.18 & 1.34 & 898 & 18.6 & 17.5 & -51.2 & -7.4 & -107.7 & -6.7 & 4.22 & -21.3 & 0.87 & 4.11 & 3.13 \\
CMK, linear, $m=122$ & 26.0 & 0.28 & 1.40 & 949 & 22.5 & 19.4 & 0.0 & -- & -56.5 & -5.3 & 4.98 & -13.0 & 0.93 & 4.89 & 3.78 \\
\midrule
CMK, linear + Gaussian ($\ell=1\times$), $m=179$ & 22.1 & 0.34 & 1.52 & 973 & 24.8 & 20.3 & 23.8 & 6.5 & -32.7 & -4.5 & 5.18 & -12.9 & 0.92 & 5.12 & 4.12 \\
CMK, linear + Gaussian ($\ell=1\times$), $m=229$ & 29.2 & 0.35 & 1.52 & 984 & 26.5 & 20.7 & 34.5 & 5.4 & -21.9 & -4.8 & 5.39 & -14.1 & 0.91 & 5.26 & 4.25 \\
CMK, linear + Gaussian ($\ell=1\times$), $m=322$ & 46.4 & 0.36 & 1.52 & 998 & 29.3 & 21.0 & 48.6 & 4.1 & -7.8 & -2.7 & 5.57 & -13.5 & 0.89 & 5.39 & 4.32 \\
CMK, linear + Gaussian ($\ell=0.5\times$), $m=179$ & 21.4 & 0.35 & 1.49 & 973 & 24.6 & 20.2 & 23.8 & 6.6 & -32.6 & -4.3 & 5.15 & -14.2 & 0.92 & 5.10 & 4.08 \\
CMK, linear + Gaussian ($\ell=0.5\times$), $m=229$ & 28.6 & 0.36 & 1.52 & 983 & 26.1 & 20.4 & 34.0 & 5.6 & -22.5 & -4.5 & 5.32 & -13.8 & 0.91 & 5.18 & 4.16 \\
CMK, linear + Gaussian ($\ell=0.5\times$), $m=322$ & 48.0 & 0.36 & 1.51 & 996 & 29.0 & 20.7 & 46.8 & 4.2 & -9.7 & -4.2 & 5.62 & -12.3 & 0.89 & 5.39 & 4.30 \\
CMK, linear + Gaussian, $(\ell,m)$ selected & 32.4 & 0.35 & 1.54 & 1006 & 27.1 & 21.1 & 56.5 & 5.3 & -- & -- & 5.34 & -13.1 & 0.91 & 5.28 & 4.23 \\
\midrule
IPCA, $K=20$ & 20.0 & 0.26 & 1.36 & 882 & 17.8 & 17.4 & -67.8 & -6.0 & -124.3 & -5.7 & 4.66 & -11.0 & 0.92 & 4.65 & 3.69 \\
IPCA, $K=40$ & 40.0 & 0.27 & 1.36 & 887 & 19.0 & 18.0 & -62.7 & -5.2 & -119.2 & -5.3 & 4.98 & -14.5 & 0.89 & 4.91 & 3.82 \\
IPCA, $K=60$ & 60.0 & 0.28 & 1.36 & 895 & 19.9 & 18.5 & -54.4 & -4.7 & -110.9 & -5.1 & 4.93 & -17.1 & 0.87 & 4.88 & 3.71 \\
\midrule
$1/N_t$ & -- & -- & -- & -- & -- & -- & -- & -- & -- & -- & 0.37 & -35.5 & -- & -- & -- \\
\bottomrule
\end{tabular}

%% file: tables/rev19_tab_decomp_jkp.tex
\begin{tabular}{lrrrr}
\toprule
Step & Score & Change & $t$ & SR \\
\midrule
IPCA, $K=60$ & 895 &  &  & 4.93 \\
Moment step, full feature space, Tikhonov $\varrho=10^{-3}$ & 929 & $+33.9$ & $3.9$ & 4.73 \\
Noise correction: subtract $\hat u\bm G^+_{t,\varrho}$ & 915 & $-13.6$ & $-4.9$ & 4.71 \\
Floor at $\eta=1$ & 943 & $+27.7$ & $5.2$ & 5.03 \\
Select $(\eta,\lsy)$: CMK estimator & 949 & $+6.4$ & $4.0$ & 4.98 \\
\bottomrule
\end{tabular}